\documentclass[12pt]{article}
\usepackage[utf8]{inputenc}
\usepackage[english]{babel}
\usepackage{hyperref}
\usepackage{textcomp}
\usepackage{notations}
\usepackage{bm}
\usepackage{comment}
\usepackage[autostyle]{csquotes}
\usepackage{authblk}

\usepackage{amsthm}
\usepackage{amsfonts}
\usepackage{amssymb}
\usepackage{graphicx}
\usepackage{comment}
\usepackage{geometry}
\usepackage{subfiles}
\usepackage{subfig}
\usepackage{mathtools}
\usepackage{pgfplots}
\pgfplotsset{/pgf/number format/use comma,compat=newest}
\usepackage{url}
\usepackage{tikz}
\usepackage{bbm}
\usepackage{floatflt}
\usetikzlibrary{arrows,decorations.pathmorphing,backgrounds,positioning,fit,petri}

\usepgfplotslibrary{fillbetween}

\newcommand{\jean}[1]{{ #1}}

\newcommand{\E}{\mathbb{E}}
\newcommand{\R}{\mathbb{R}}

\newcommand{\bPhi}{\boldsymbol \Phi}
\newcommand{\bmu}{\boldsymbol \mu}
\newcommand{\bSigma}{\boldsymbol \Sigma}
\newcommand{\bDelta}{\boldsymbol \Delta}

\theoremstyle{plain} 
\newtheorem{theorem}{Theorem}

\newtheorem{lemma}[theorem]{Lemma} 
\newtheorem{proposition}[theorem]{Proposition} 

\theoremstyle{definition}

\newtheorem{result}{Result}
\newtheorem{hypothesis}{Hypothesis}

\theoremstyle{remark}

\title{\textbf{Fundamental limits in structured PCA, \\ and how to reach them}}

\author{Jean Barbier$^{\dagger,*}$, Francesco Camilli$^{\dagger,*}$, \\
Marco Mondelli$^\diamond$, Manuel S\'aenz$^\triangledown$}
\affil{$\dagger$ \emph{Abdus Salam International Center for Theoretical Physics, Italy}\\ 
$\diamond$ \emph{Institute of Science and Technology, Austria}\\
$\triangledown$ \emph{Universidad de La Rep\'ublica, Uruguay}}
\date{}

\begin{document}
\maketitle

{
\let\thefootnote\relax\footnotetext{$*$ Corresponding authors: \url{jbarbier@ictp.it}, \url{fcamilli@ictp.it}}
}

\vspace{-10pt}

\begin{abstract}
How do statistical dependencies in measurement noise influence high-dimensional inference? 
    To answer this, we study the paradigmatic spiked matrix model of principal components analysis (PCA), 
    where a rank-one matrix is corrupted by additive noise. We go beyond the usual independence assumption on the noise entries, 
    by drawing the noise from a low-order polynomial orthogonal matrix ensemble.
    The resulting noise correlations make the setting relevant for applications but analytically challenging. We provide the first 
    characterization of the Bayes-optimal limits of inference in this model. If the spike is rotation-invariant, 
    we show that standard spectral PCA is optimal. However, for more general priors, both PCA and the existing approximate message passing algorithm (AMP) 
    fall short of achieving the information-theoretic limits, which we compute using the replica method from statistical mechanics. We thus propose a novel AMP, inspired by the theory of Adaptive Thouless-Anderson-Palmer 
    equations, which saturates the theoretical limit. This AMP comes with a rigorous 
    state evolution analysis tracking its performance. 
    Although we focus on specific noise distributions, our methodology can be generalized to a wide class of trace matrix ensembles at the cost 
    of more involved expressions. Finally, despite the seemingly strong assumption of rotation-invariant noise, our theory empirically predicts algorithmic performance on real data, pointing at remarkable universality properties.
\end{abstract}

\newpage

\tableofcontents

\part{Main part}

\section{Introduction}
The success of inference and learning algorithms depends strongly 
on the structure of the high-dimensional noisy data they process. Consequently, quantifying how this structure helps algorithms to overcome the curse of dimensionality has become a central topic in statistics and machine learning.
Classical examples include sparsity in compressed sensing \cite{donoho2006compressed}, low-rank structure in matrix recovery \cite{candes2006exact}, or community structure in community detection \cite{SBM_abbe18}. In all these models, structure is usually assumed only at the signal's level. But the decomposition of the data into ``signal'' (the component considered of interest) and ``noise'' (the rest) is often arbitrary and application-dependent. E.g., in classification of ``dogs/cats'', the training images contain a lot of information unrelated to dogs and cats, e.g., on the notions of ``inside/outside'', ``day/night'', etc. Yet, this highly structured potential source of information is discarded as random noise (independent, Gaussian, etc.). Most of the research effort has thus focused on understanding how the signal structure alone helps inferring it. In contrast, much less is known on the role of the noise structure and how to exploit it to improve inference.

Given their ubiquitous appearance in the statistics literature, spiked matrix models, which were originally formulated as models for probabilistic principal component analysis (PCA) \cite{johnstone2001distribution}, are now a paradigm in high-dimensional inference. Thanks to their universality features they, and their generalizations, find numerous applications in other central problems, including community detection \cite{SBM_abbe18}, group synchronization \cite{perry2018message} and sub-matrix localization or high-dimensional clustering \cite{lesieur2015mmse}. They thus offer the perfect benchmark to quantify the influence of noise structure. In this paper we focus on the following estimation problem: a statistician needs to extract a rank-one matrix (the \emph{spike}) $\mathbf{P}^*:=\mathbf{X}^*\mathbf{X}^{*\intercal}$, $\mathbf{X}^*\in\mathbb{R}^N$, from the data
\begin{align}\label{channel0_main}
    \mathbf{Y}=\frac{\lambda}{N}\mathbf{P}^*+\mathbf{Z}\in\mathbb{R}^{N\times N}
\end{align}
with ``noise'' $\mathbf{Z}$ and \emph{signal-to-noise ratio} (SNR) $\lambda\ge 0$. 

The spectral properties of finite rank perturbations of large random matrices like \eqref{channel0_main} were intensively investigated in random matrix theory (see e.g. 
\cite{baik2005phase,baik2006eigenvalues,benaych2011eigenvalues}), showing the presence of a threshold phenomenon coined BBP transition (in reference to the authors of \cite{baik2005phase}): when $\lambda$ is large enough, the top eigenvalue of $\mathbf{Y}$ detaches from the bulk of eigenvalues. Its corresponding eigenvector has then a non-trivial projection onto the sought ground truth $\mathbf{X}^*$, and can be used as its estimator. The problem has also been approached from the angle of Bayesian inference \cite{korada2009exact,it-PCAMontanari14,XXT,2016arXiv161103888L}. In particular, besides the previous spectral estimator, there exists a whole family of iterative algorithms, known as approximate message passing (AMP), that can be tailored
to take further advantage of prior structural information 
about the signal and noise. AMP algorithms were first proposed for estimation in linear models 
\cite{kabashima2003cdma,Donoho09_Compressed_sensing}, but have since been applied to a range of statistical estimation problems, including generalized linear models 
\cite{barbier2019optimal,rangan2011generalized} and low-rank matrix estimation 
\cite{it-PCAMontanari14,montanari2017estimation}.
An attractive feature of AMP is that its performance in the high-dimensional limit can often be characterized
by a succinct recursion called state evolution \cite{Bayati_Montanari11,bolthausen2014iterative}
. Using the state evolution analysis, it has been proved that AMP achieves Bayes-optimal performance for some models \cite{it-PCAMontanari14,
montanari2017estimation,
barbier2019optimal}, and a conjecture posits that for a wide range of estimation problems, AMP is optimal among polynomial-time algorithms \cite{montanari2022equivalence}.

The references mentioned above rely on the assumption of independent and identically distributed (i.i.d.) noise, often taken Gaussian $Z_{ij}=Z_{ji}\sim \mathcal{N}(0,1)$, under which 
\eqref{channel0_main} is the well-known spiked Wigner model \cite{
johnstone2001distribution}. This independence, or ``absence of structure'', in the noise simplifies greatly the analysis.
In order to relax this property, we may seek inspiration from the statistical physics literature on disordered systems. An idea that was first brought forth in \cite{MultiSK_Contucci,Panchenko_multi13} for the Sherrington-Kirkpatrick model, and later imported also in high dimensional inference \cite{
DBMNL,Guionnet_inhomogeneous_noise} is that of giving an inhomogeneous variance profile to the noise matrix elements (we mention that this idea in inference is similar to the earlier definition of ``spatially coupled systems'' \cite{felstrom1999time,kudekar2011threshold} in coding theory, see \cite{XXT} for its use in the present context). The procedure makes the $(Z_{ij})$ no longer identically distributed, but it leaves them independent. This is an important step towards more structure in the noise. Yet, the independence assumption is a rather strong one. In fact, \cite{Guionnet_inhomogeneous_noise} showed that a broad class of observation models, as long as the independence assumption holds, are information-theoretically equivalent to one with independent Gaussian noise. 

One way to go beyond is to consider noises belonging to the wider class of \emph{rotationally invariant matrices}. Since the appearance of the seminal works \cite{Marinari_ParisiI,Marinari_ParisiII,Parisi_Potters}, there has been a remarkable development in this direction, as evidenced by the rapidly growing number of papers on spin glasses 
\cite{opper2016theory,adatap,maillard2019high}
and inference 
\cite{gerbelot2020asymptotic,fan2022approximate,zhong2021approximate,venkataramanan2022estimation} that 
take into account structured disorder, including the present one. Indeed, we hereby consider a spiked model in which the noise $\mathbf{Z}$ is drawn from an orthogonal matrix ensemble different from the Gaussian orthogonal ensemble (the only one with independent entries). 
Intuitively, the presence of dependencies in the noise should be an advantage for an algorithm sharp enough to see patterns within it and use them to retrieve the sought low-rank matrix. 
Going in that direction, \cite{fan2022approximate} proposed a version of AMP designed for rotationally invariant noises (using earlier ideas of \cite{adatap,opper2016theory}). 
Furthermore, in a recent work \cite{price_of_ignorance22}, part of the authors analysed a Bayes estimator and an AMP, both assuming Gaussian noise, whereas the actual noise in the data was drawn from a generic orthogonal matrix ensemble. However, besides intuition and the mentioned works, to the best of our knowledge there is little theoretical understanding of the true role played by noise structure in spiked matrix estimation and more generically in inference. In particular, prior to our work there was no theoretical prediction of optimal performance to benchmark practical inference algorithms.

\section{Setting and main results}
Our analysis focuses on two types of signal's distributions: the factorized prior $dP_X(\mathbf{x})=\prod_{i\leq N}dP_X(x_i)$ and a uniform prior measure over the $N$-dimensional sphere of radius $\sqrt N$. By convention $\int x^2\,dP_X(x)=1$; which amounts to rescale $\lambda$. The noise matrix $\mathbf{Z}$ is drawn from a trace random matrix ensemble, defined by a certain potential $V:\mathbb{R}\mapsto
\mathbb{R}$. $V$ is extended to matrices as follows: if $\mathbf{A}=\text{diag}(a_1,\dots,a_N)$ then
$V(\bA)=\text{diag}(V(a_1),\dots,V(a_N))$. For real symmetric matrices $\bM=\bU\bA\bU^\intercal$, with $\mathbf{U}$ orthogonal, $V(\mathbf{M})=\bU V(\bA)\bU^\intercal$. With these notations we can write the density of the trace ensemble (with
normalization constant $C_V$) as
\begin{align}\label{Z-ensemble_main}
    dP_Z(\mathbf{Z})= C_V \exp\big(-\frac{N}{2}\Tr V(\mathbf{Z})\big)\prod_{i\leq j}dZ_{ij}\,.
\end{align}
Instances of such ensembles have a spectral decomposition $\mathbf{Z}=\mathbf{O}\mathbf{D}\mathbf{O}^\intercal$, with $\mathbf{O}$ uniformly distributed over $N\times N$ orthogonal matrices. The distribution of the eigenvalues in the diagonal matrix $\mathbf{D}$, which is independent of $\mathbf{O}$, can be explicitly written, see the Supporting Information (SI), Sec. 1.2. Only the special case $V(x)=x^2/(2\sigma)$, corresponding to the Gaussian orthogonal ensemble, induces independent (Gaussian
distributed) matrix entries. Any other potential generates dependencies among matrix elements and thus \emph{structure}. E.g., if we
take $V(x)=x^4/4$, the probability density would be proportional to $\prod\exp(-\frac N 8
Z_{ij}Z_{jk}Z_{kl}Z_{li})$, which is clearly not factorizable over matrix entries.

Analysing the model for a generic potential $V$ is \emph{possible} through the novel methodology presented in this paper. {Indeed, as discussed in Appendix \ref{sec:approx_potential}, this can be done by studying the inference problem whose noise's potential is a polynomial approximation of $V$.} However, if we take a generic polynomial potential $V$, the higher the order, the more technical and cumbersome our derivations become. Therefore, for the sake of clarity, we focus on a concrete example of non-trivial correction to i.i.d.\ noise: 
the quartic matrix potential $V(x)=\mu x^2/2+\gamma x^4/4$, where $\mu$ and $\gamma$ are two non-negative real numbers \cite{quartic_Parisi}.
We could have also considered a non-symmetric potential with a cubic term too, but for simplicity we restrict ourselves to that case as symmetry slightly simplifies the computations.
The noise $\bZ$ drawn from the quartic matrix ensemble has a known $N\to\infty$ asymptotic eigenvalue distribution \cite{potters2020first}
\begin{align}\label{eq:density_main}
    \rho(x)dx=(\mu+2a^2\gamma+\gamma x^2)\sqrt{4a^2-x^2}/(2\pi) \,dx,
\end{align}
where $a^2:=(\sqrt{\mu^2+12\gamma}-\mu)/(6\gamma)$. In order to have a coherent definition of SNR, we also fix $\int x^2d\rho(x)=1$, which implies $$\gamma= \gamma(\mu)=(8-9\mu+\sqrt{64-144\mu+108\mu^2-27\mu^3})/27\,.$$
When $\mu=1, \gamma(1)=0$ and we recover the pure Wigner case. On the contrary, $(\mu=0,\gamma(0)=16/27)$ corresponds to a purely quartic case with unit variance, the ``most structured'' ensemble in this class. 
Therefore, $\mu$ allows us to interpolate between unstructured and structured noise ensembles.

{We emphasize that, despite this model may seem rather academic at first sight, we will see that our main assumption, i.e., the rotational invariance of the noise, turns out to yield a theory which accurately predicts the empirical performance of algorithms for inference of low-rank matrices hidden in noise coming from \emph{real data sets} from various application domains. This is probably a consequence of strong universality properties, yet to be understood from a theoretical perspective, along the lines of \cite{dudeja2022universality2,dudeja2022universality}. We thus argue that our assumptions are in fact rather mild, making our novel inference algorithms relevant for potential future applications.}

We now introduce the Bayesian framework we are going to analyse. Let $\mathbf{P}:=\mathbf{x}\mathbf{x}^{\intercal}$. The posterior measure reads
\begin{align}
    dP_{X\mid Y}(\mathbf{x}\mid \bY)\!=\!\frac{C_V}{P_Y(\bY)}dP_X(\mathbf{x})\exp\big(\!-\frac{N}{2}\Tr V\big(\bY-\frac{\lambda}{N}\mathbf{P}\big)\!\big).\label{posterior_main}
\end{align}
The evidence $P_Y(\bY)$ is simply the integral of the numerator. We stress that the prior $P_X$ and the likelihood $P_{Y\mid X}$ match respectively the distribution of the signal and the noise density $P_Z$, and $\lambda$ is known. Therefore we are in the \emph{Bayes-optimal setting}. Studying the limits of inference in this setting draws a fundamental line between what is information-theoretically possible and what is not in terms of performance of inference.

A main object of interest is the \emph{free entropy}, which is minus the Shannon entropy of the data: $F_N(\bY):= -H(\bY)=\EE\ln P_Y(\bY)$. It is related to the mutual information between signal and data through the identity $I(\bP^*; \bY)
=-F_N(\bY)+\ln C_V-\frac{N}{2}\mathbb{E}\Tr  V(\mathbf{Z})$. {The relevance of the latter is extensively discussed in Sec. \ref{sec:results}.}
Using the form of the observation model in \eqref{channel0_main} it reads
\begin{align}\label{mutual-info_main}
-\text{I}(\bP^*; \bY)&=\mathbb{E}\ln\int dP_X(\mathbf{x})e^{-H_N(\mathbf{x};\mathbf{Z},\bX^*)}=:\mathbb{E}\ln\mathcal{Z}\,,
\end{align}
where the Hamiltonian linked to the partition function $\mathcal{Z}$ is
\begin{align}
    H_N(\mathbf{x};\mathbf{Z},\bX^*):=\frac N2\Tr\big[ V\big(\mathbf{Z}+\frac{\lambda}{N}(\mathbf{P}^*-\mathbf{P})\big)-V(\mathbf{Z})\big]\,.\label{Hamilt_main}
\end{align}
In this way, the problem is mapped onto a statistical mechanics model with ``quenched randomness'' $\mathbf{Z},\bX^*$ and ``spins'' $\bx$ with Gibbs-Boltzmann distribution associated to this Hamiltonian (i.e., the posterior). This Hamiltonian is tricky to directly deal with, so a key point will be to ``convert'' it into a more tractable form, see Sec. \ref{sec:methods} and Sec. \ref{sec:quadraticModel}.

\begin{figure*}[ht!!!]
    \begin{center}                
        \includegraphics[width=0.32\linewidth,trim={0.7cm 0 1cm 1cm},clip]{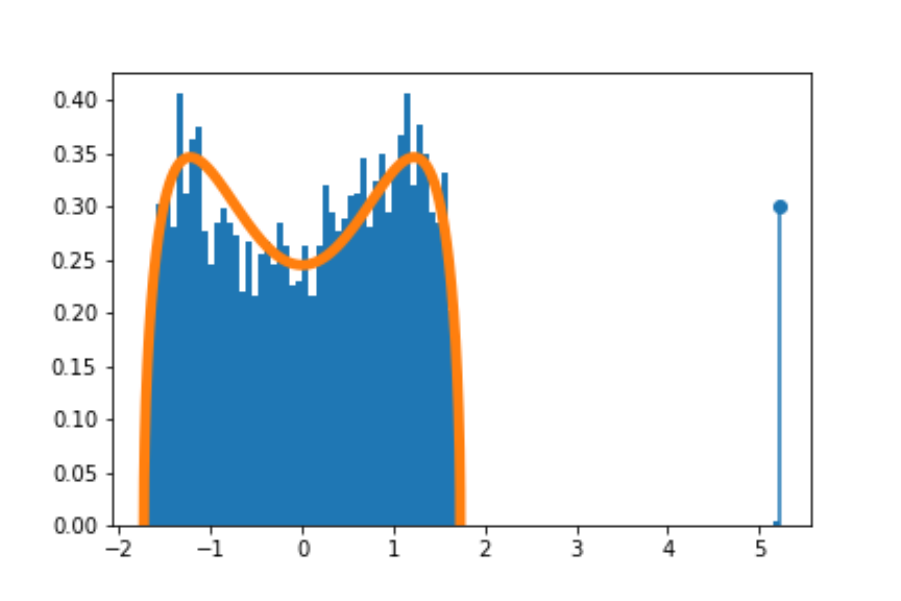}
        \includegraphics[width=0.32\linewidth,trim={1cm  0 1cm 1cm},clip]{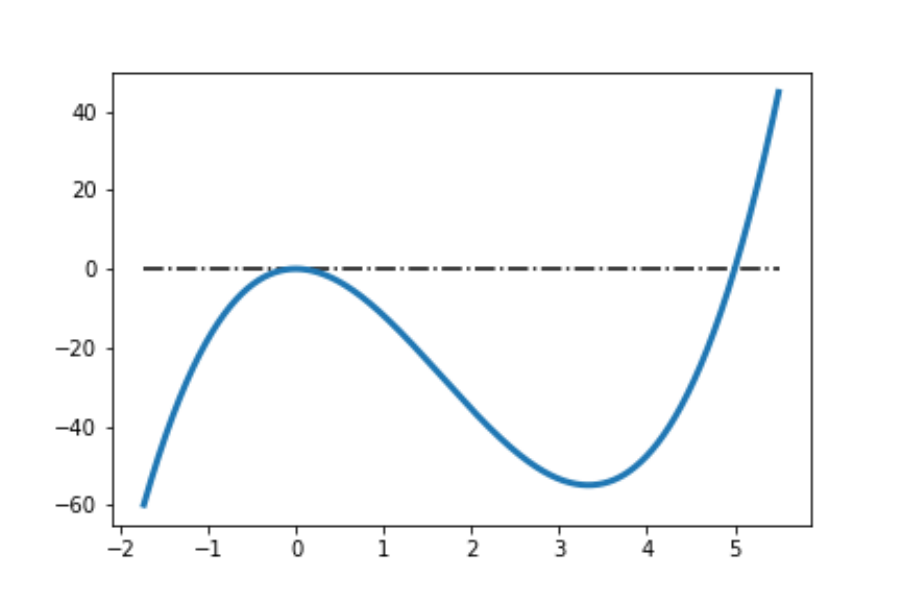}
        \includegraphics[width=0.32\linewidth,trim={0.7cm 0 1cm 1cm},clip]{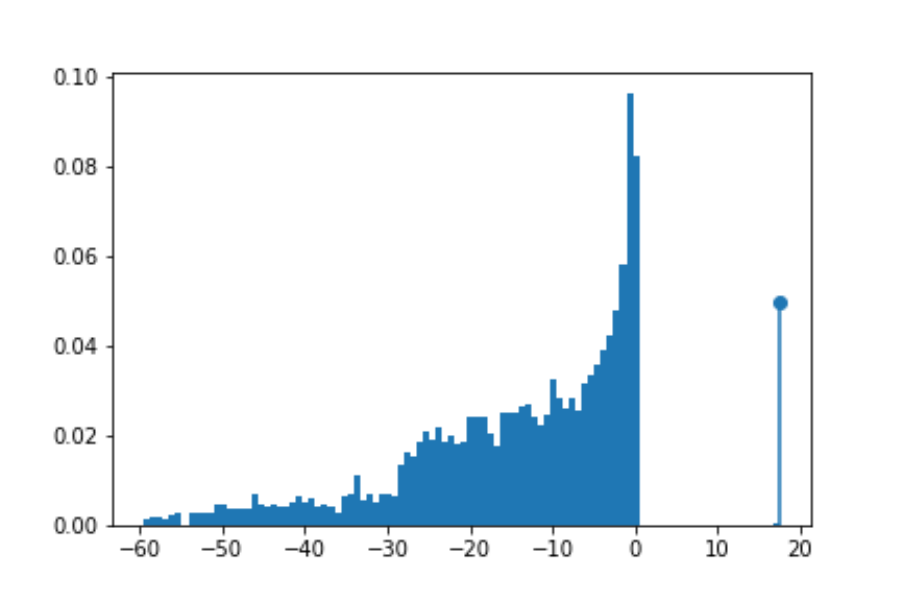}
    \end{center}
     \vspace{-17pt}
\caption{We have set $\mu=0,\gamma(0)=16/27,\lambda=5, N=4000$ and generated one instance of the data model \eqref{channel0_main}. (Left) Histogram of the eigenvalues of $\mathbf{Y}$. The leading eigenvalue is emphasized and the orange curve is the density in \eqref{eq:density_main}. (Middle) Optimal pre-processing function $J(x)=\mu\lambda x-\gamma \lambda^2x^2+\gamma\lambda x^3$. (Right) Histogram of the eigenvalues of $J(\bY)$. The pre-processing $J$ flushes the bulk to the negative axis while pushing \emph{only} the leading eigenvalue even further from the bulk in the positive direction.}\label{fig:spectral2_main}
\vspace{-5pt}
\end{figure*}

\subsection{Result 1: Information-theoretical limits}
Our first result is a variational formula for the mutual information via the celebrated \emph{replica method} \cite{mezard2009information} outlined in Sec. \ref{sec:methods}: if we  let $\btau_*:= {\rm argmax}\{ f_\rho(\btau):\btau \in\mathbb{R}^{13}, \nabla f_\rho(\btau)=\boldsymbol{0}\}$ then we have the following low-dimensional expression for the mutual information between hidden spike and the data:
\begin{equation}\label{var_princ_main}
    \textstyle{\frac{1}{N}I(\bP^*; \bY)\xrightarrow{N\to\infty}-\,f_\rho(\btau_*)\,.}
\end{equation}
The argmax is selected and not the argmin as $f_\rho$ is a free entropy (i.e., \emph{minus} free energy, the free energy being minimized in physics). $f_\rho$ and its derivation are reported in Sec. \ref{replica_comp_section}. The $13$ coupled fixed point equations coming from $\nabla f_\rho=\boldsymbol{0}$ {will reduce to only $2$ (see \eqref{m_replica_simplified}--\eqref{tildeV_eq_BO}) thanks to special symmetries inherent to the Bayes-optimal nature of our analysis}. One of the two remaining order parameters, denoted $m^2$ and called (squared) ``magnetization'', quantifies the asymptotic trace inner product between the minimum mean-square error (MMSE) estimator $\int  dP_{X\mid Y}(\mathbf{x}\mid \bY)\mathbf{x}\mathbf{x}^\intercal $ and the spike $\bX^*\bX^{*\intercal}$. It allows us to compute the MMSE as
\begin{equation}\label{eq:MMSE_main}
    \textstyle{\frac{1}{2N^2}\EE\| \bX^*\bX^{*\intercal}-{\textstyle \int } dP_{X\mid Y}(\mathbf{x}\mid \bY)\mathbf{x}\mathbf{x}^\intercal\|_{\rm F}^2\xrightarrow{N\to\infty}\frac{1-m^2}2\,,}
\end{equation}
with $m$ solving the aforementioned system of equations.

\subsection{Result 2: Optimality of PCA for rotationally invariant priors}
The above results hold for a factorized prior $P_X^{\otimes N}$. Nevertheless, if $\bX^*$ is uniformly distributed on the sphere, a variational formula analogous to \eqref{var_princ_main} can still be derived, as shown Sec. \ref{sec:PCAopt}, and the related MMSE computed. Analytical arguments and numerical experiments show that the latter can be achieved using the naive spectral estimator $C\boldsymbol \nu \boldsymbol \nu^\intercal$ of $\bP^*$ obtained from the principal eigenvector $\boldsymbol \nu=\boldsymbol \nu(\bY)$ of $\bY$ 
properly re-scaled by a certain factor $C(\lambda,\rho)$, see \cite{benaych2011eigenvalues}.

\subsection{Result 3a: Optimal pre-processing of the data}
{
Instead of using an AMP with iterates based on $\mathbf{Y}$, we introduce a pre-processing procedure driven by the AdaTAP formalism 
\cite{adatap}. The end result is an \emph{effective} quadratic model (i.e., with only pairwise interactions) which is ``equivalent'' (in a proper sense described below) to the original one, with coupling matrix 
\begin{align}\label{eq:pre-processed_main}
    J(\bY)=\mu\lambda \bY-\gamma \lambda^2\bY^2+\gamma\lambda\bY^3.
\end{align}
This new model being quadratic is now solvable using AdaTAP/AMP, and possesses the same thermodynamic properties (free entropy, phase transitions, etc.) as well as the same marginal means and variances as the model in \eqref{posterior_main} when $N\to\infty$ (and thus equivalent for our purposes). Therefore, to approximate the MMSE estimator, one can simply ``pre-process'' $\bY$ by applying $J(\bY)$ and then efficiently compute the marginals of the resulting quadratic model by AdaTAP/AMP, see next section. AdaTAP allows to parametrize the free entropy (i.e., log-partition function) of a model with \emph{quadratic} Hamiltonian, for a given instance of the interaction matrix, in terms of $O(N)$ order parameters, some of which correspond to the sought marginal means $(\langle x_i\rangle)_{i\le N}$ and associated variances. The extremization w.r.t. them yields equations that can be solved iteratively and identified with an AMP algorithm. However, the Hamiltonian in \eqref{Hamilt_main} is \emph{not} quadratic in $\mathbf{x}$, but can be made so by fixing certain order parameters as outlined in Sec. \ref{sec:methods}. The resulting coupling matrix depends on $\mathbf{Y}$ and on the fixed order parameters, whose values are constrained by Bayes-optimality (see Sec. \ref{sec:results}). Using these values, for an initial quartic $V(x)=\mu x^2/2+\gamma x^4/4$, we get the above interaction matrix in \eqref{eq:pre-processed_main} (see Sec. \ref{sec:adatapFreeEn} and \ref{statCondADATAP}). }

The ``cleaning effect'' of $J(\bY)$ is illustrated in Fig. \ref{fig:spectral2_main}. {In general, for a $(K+1)$-order polynomial matrix potential, the pre-processed matrix is a polynomial $J(\bY)=\sum_{k\le K}c_k\bY^k$}, with $(c_k)_{k\le K}$ depending on $V$. 
E.g., for $V(x)=\xi x^6/6$ (with $\xi=27/80$ to select unit variance) the pre-processing (derived similarly to the quartic case, see Sec. \ref{sec:preprocessing_sestic}) is $J_6(x)=\xi \lambda x^5-\xi\lambda^2 x^4-\xi\lambda^2 x^2$; it has an effect similar to that in Fig~\ref{fig:spectral2}.
We point out that the statistics of the noise could be only partially known. This issue can be overcome by learning the $(c_k)$ from the data, see Appendix \ref{sec:EM}. 

\begin{figure*}[t!]
        \begin{center}
                \includegraphics[width=0.45\linewidth,
                clip]{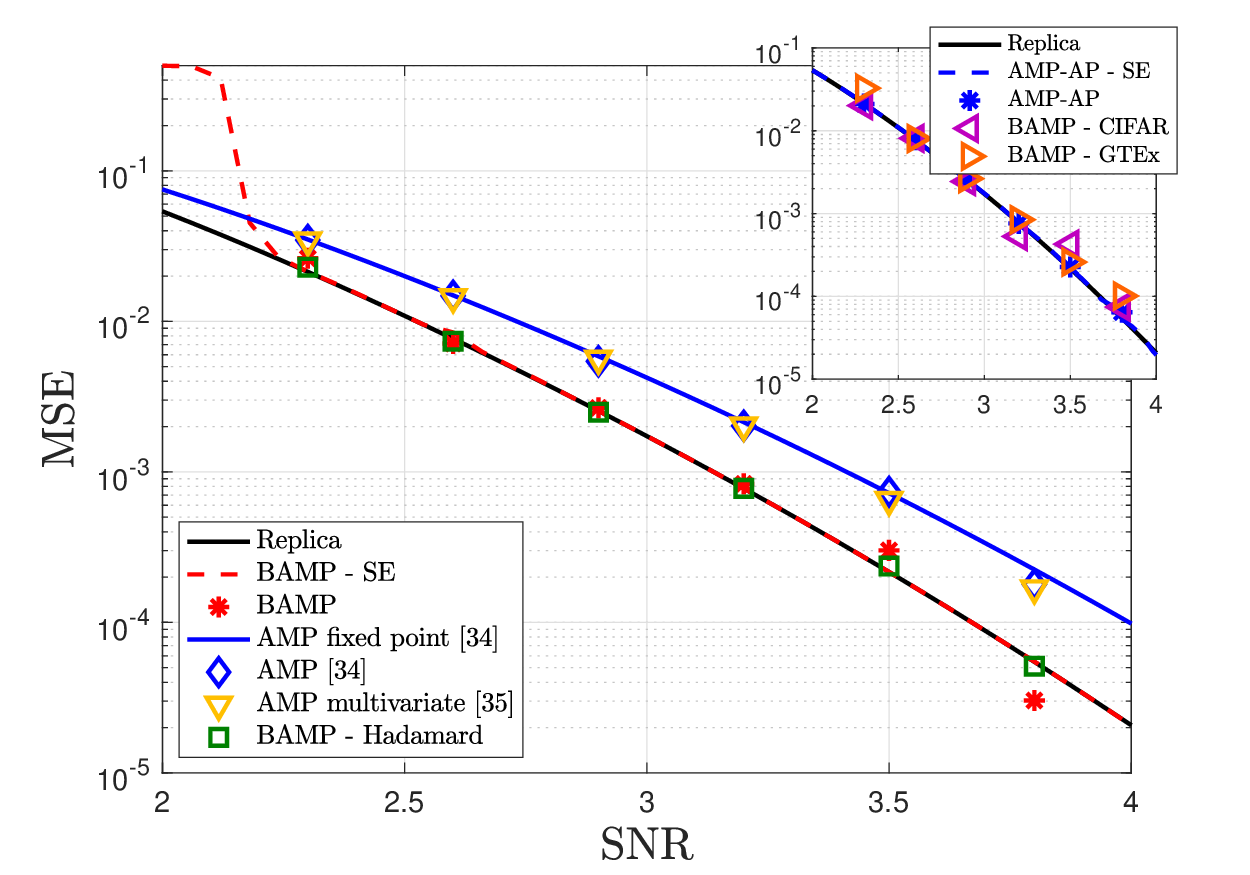}
                \hspace{2em}
                \includegraphics[width=0.45\linewidth,clip]{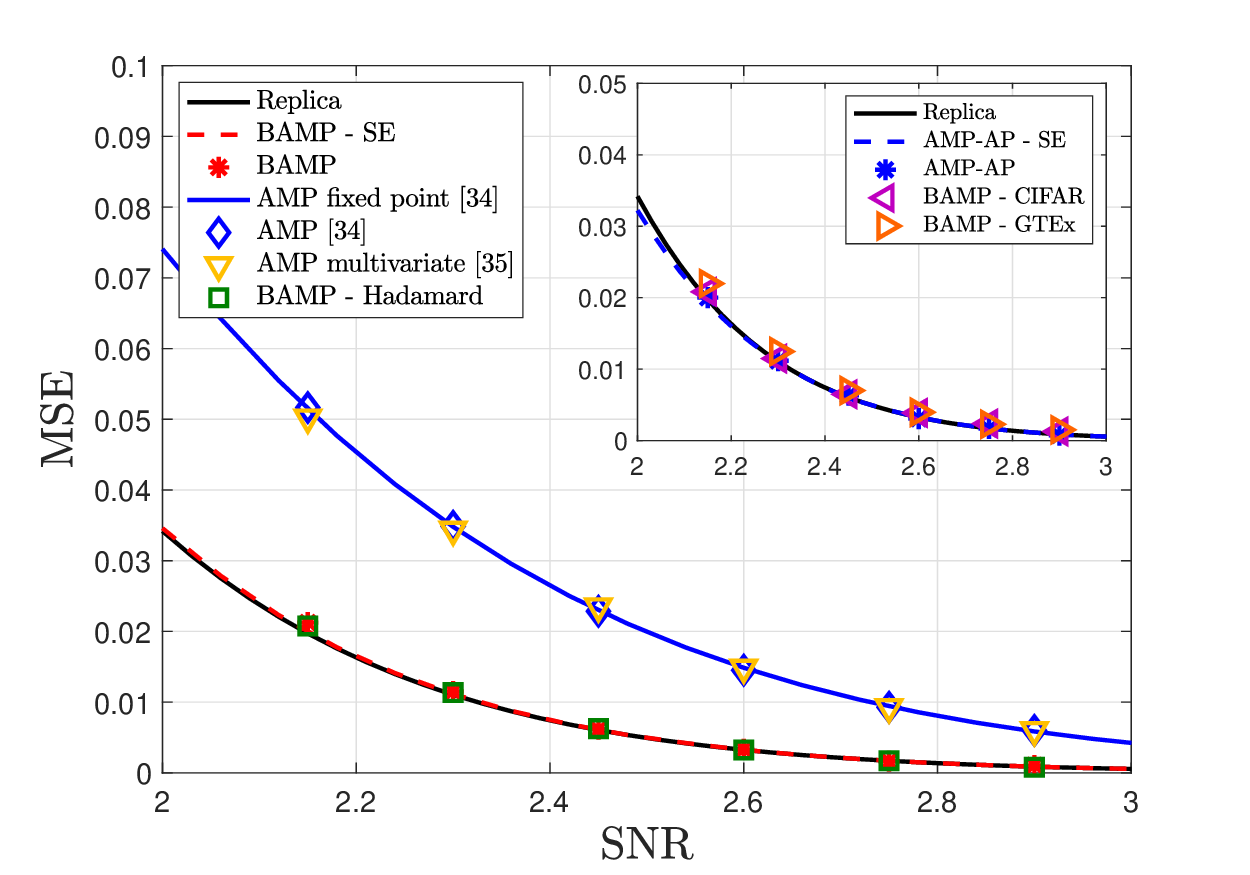}
        \end{center}
        \vspace{-15pt}
\caption{{Quartic potential with $\mu= 0$ (left) and pure power six potential (right). Comparison of the following inference procedures: \emph{(i)} (black) replica prediction of the MMSE, \eqref{eq:MMSE_main}. \emph{(ii)} (red) performance of the BAMP algorithm, where $g_{t+1}$ is the single-iterate posterior mean denoiser $g_{t+1}(f)=\mathbb{E}[X^*\mid F_t=f]$. The red line corresponds to the fixed point of the MSE given by the state evolution recursion, and the red stars denote the MSE obtained by running BAMP \eqref{BAMP_main} with the proper pre-processing. \emph{(iii)} (blue) performance of the AMP proposed in \cite{fan2022approximate}. The blue line corresponds to the fixed point MSE obtained with a single-iterate posterior mean as denoiser, and the blue diamonds denote the MSE obtained by running the AMP of \cite{fan2022approximate} with the same denoiser. \emph{(iv)} (ochre squares) MSE obtained by the AMP of \cite{zhong2021approximate} (without the pre-processing of $\mathbf{Y}$), which employs a full memory posterior mean denoiser: $h_{t+1}(f_1, \ldots, f_t) = \mathbb E[X^*\mid (F_1, \ldots, F_t) = (f_1, \ldots, f_t)]$. Finally, \emph{(v)} (green triangles) performance of BAMP when the uniformly distributed matrix $\bO$ (appearing in the spectral decomposition of the noise $\bZ$) is replaced by the product of the Hadamard-Walsh matrix and a diagonal matrix with i.i.d.\ Rademacher entries as in \cite{dudeja2022universality}. In the smaller plots in the top-right corner, we report the performance of AMP-AP (blue) and of BAMP for our universality experiments involving the CIFAR-10 ``plane'' class (purple) and the ``muscle skeletal'' GTEx dataset (orange).}}\label{fig:BAMP_main}
\vspace{-5pt}
\end{figure*}

\subsection{Result 3b: Bayes-optimal AMP}

{First, we show in Sec. \ref{sec:sub-opt} that existing AMPs \cite{fan2022approximate,zhong2021approximate} do not saturate the MMSE predicted by \eqref{eq:MMSE_main}. We provide a replica-based theory showing that despite these existing AMPs are aware of the noise structure/statistics, they nevertheless make an implicit mismatched assumption of i.i.d. Gaussian noise: the noise structure is ``only'' exploited to enforce convergence despite the mismatch,  rather than as a source of greater statistical accuracy, in contrast to the proposed AMP we explain now.}

{To cure this issue we propose to employ the processed $J(\bY)$ in AMP}, which leads to our Bayes-optimal approximate message passing (BAMP) algorithm defined by the recursion
\begin{align}
    {\textstyle \bff^t = J(\bY) \bu^t - \sum_{i\le t}{\sf c}_{t, i}\bu^i, \quad \bu^{t+1} = g_{t+1}(\bff^t), \quad t\ge 1\,,}\label{BAMP_main}
\end{align}
with $g_{t+1}$ applied component-wise. For simplicity, we assume to have access to an initialization $\bu^1\in \mathbb R^N$ independent of the noise $\bZ$ and with a strictly positive correlation with $\bX^*$, i.e., 
\begin{equation}\label{eq:AMPinit_main}
(\bX^*, \bu^1)\stackrel{\mathclap{W_2}}{\longrightarrow} (X^*, U_1), \ \ \mathbb E[X^* \,U_1]:=\epsilon>0, \ \ \mathbb E[U_1^2]= 1.
\end{equation}
This requirement 
is rather standard in the analysis of AMP algorithms \cite{barbier2019optimal,fan2022approximate,feng2022unifying}. However, as having access to such an initialization is often impractical, recent work \cite{montanari2017estimation,
mondelli2021pca,zhong2021approximate} has designed AMPs initialized with the top eigenvector $\boldsymbol \nu(\bY)$.

By carefully choosing the Onsager coefficients $\{{\sf c}_{t, j}\}_{j
\in [t]}$, we rigorously obtain BAMP's state evolution characterization.

\begin{theorem}[State evolution of BAMP]\label{th:SE_main}
Let $J(\bY)=\sum_{i\le K} c_i \bY^i$. Consider the AMP of \eqref{BAMP_main} initialized as \eqref{eq:AMPinit_main}, with Onsager coefficients $\{{\sf c}_{t, j}\}_{j
\in [t]}$ given in Sec. \ref{subsec:OnsSE}, and where $(g_{t+1})_{t\ge 1}$ are $\mathcal{C}^1$ and Lipschitz. Then, the following limit holds almost surely for any order 2 pseudo-Lipschitz function\footnote{A function $\psi\colon\mathbb R^m\to\mathbb R $ is \emph{pseudo-Lipschitz of order $2$} if there exists a constant $ C>0 $ such that, for all $ \bx,\by\in\mathbb R^m$, $\|\psi(\bx) - \psi(\by)\|_2 \le C(1 + \|\bx\|_2 + \|\by\|_2) \|\bx - \by\|_2$. 
} $\psi: \mathbb R^{2t+2}  \to \mathbb R$ and $t\ge 1$:
\begin{align}\label{eq:conv_main}
    & \frac{1}{N} \sum_{i\le N} \psi(u_i^1, \ldots, u_i^{t+1}, f_i^1, \ldots, f_i^t, X^*_i) \nonumber\\ &\qquad\qquad\xrightarrow{N\to\infty}\mathbb E \,\psi(U_1, \ldots, U_{t+1}, F_1, \ldots, F_t, X^*) \,.
\end{align}
Here, for $i\in [t]$, $U_{i+1}=g_{i+1}(F_t)$ and $(F_1, \ldots, F_t)=\boldsymbol{\mu}_t X^*+(W_1, \ldots, W_t)$, with $(W_i)_{i\le t}$ a multivariate Gaussian vector whose covariance as well as $\boldsymbol{\mu}_t$ are given in Sec. \ref{subsec:OnsSE}.
\end{theorem}

\eqref{eq:conv_main} provides a high-dimensional characterization of our proposed BAMP. A suitable choice of $\psi$ readily gives the MSE of the BAMP iterates. We also note that our result is equivalent to the almost sure convergence in Wasserstein-2 distance of the joint empirical distribution of $(\bu^1, \ldots, \bu^{t+1}, \bff^1, \ldots, \bff^t, \bX^*)$ to 
$(U_1, \ldots, U_{t+1}, F_1, \ldots, F_t, X^*)$, see Corollary 7.21 of \cite{feng2022unifying}.


We emphasize that our BAMP algorithm is \emph{not} the usual AMP of \cite{fan2022approximate}, where the data matrix $\bY$ is just replaced by the pre-processed matrix $J(\bY)$. Indeed, tuning the Onsager coefficients $\{{\sf c}_{t, i}\}$ entering BAMP requires a novel type of ``multi-stage'' state evolution recursion which is completely different from the one in \cite{fan2022approximate}. The novel acronym we introduce stresses this crucial distinction. 
While our replica prediction for the MMSE is non-rigorous, the state evolution analysis of BAMP is rigorous. 
In Sec. \ref{sec:num}, we show that BAMP improves over the AMP in \cite{fan2022approximate} by comparing their fixed points. This improvement is thus a rigorous conclusion, while the conjecture is that BAMP saturates the Bayes-optimal performance.

{Finally, the ``multi-stage'' state evolution of BAMP suggests a choice of the denoisers in the AMP of \cite{fan2022approximate}, which differs from the greedy strategy of \cite{zhong2021approximate} (i.e., picking the full posterior mean denoiser at every iteration). The numerical results of Sec. \ref{sec:num} also show that this denoiser selection --motivated by BAMP-- meets the BAMP performance and, hence, the replica prediction of the Bayes-optimal error.}

\section{Numerical results and discussion}\label{sec:num}

\subsection{BAMP vs the replica prediction}
{The left plot of Fig. \ref{fig:BAMP_main} considers the quartic ensemble for $\mu=0$, and the right one refers to the pure power six potential.} The signal $\bX^*$ has a Rademacher prior $X_i^*\sim \frac12 (\delta_1+\delta_{-1})$. The estimators of the spike $\bX^*\bX^{*\intercal}$ are compared in terms of the MSE achieved at the fixed point, as a function of the SNR $\lambda$. All algorithms are run for $N=8000$, {they are initialized with $\bu^1$ that satisfies \eqref{eq:AMPinit_main}}, and the results are averaged over $50$ trials; the state evolution recursions and the replica prediction are for $N\to\infty$. {In Sec. \ref{sec:7.2} we provide additional numerical results for a \emph{sparse}} Rademacher prior, which display a similar qualitative behavior.

We observe that all algorithms converge rapidly: $10$ iterations are sufficient to reach the corresponding fixed points. A few remarks concerning the results displayed in Fig. \ref{fig:BAMP_main} are now in order. First, in all settings, the fixed point of the BAMP state evolution (red) matches the replica prediction (black). 
This is a strong numerical evidence supporting our conjecture that the proposed BAMP algorithm is Bayes-optimal. These theoretical curves for $N\to\infty$ are also remarkably close to the MSE achieved by the BAMP algorithm at $N=8000$.

Secondly, 
there is a clear performance gap between our proposed BAMP (red) and the existing AMP algorithms \cite{fan2022approximate,zhong2021approximate} (single-step denoiser in blue, and multi-step in ochre). For $V(x)=\xi x^6/6$ the gap is even more evident. As predicted by our theory, the gap is reduced when $\mu$ approaches $1$ with all curves collapsing for $\mu=1$, see 
Sec. \ref{sec:7.2}. 

{Thirdly, we consider a choice of denoisers in the AMP of \cite{fan2022approximate} which is motivated by our BAMP: if the potential has degree $K$, every $K$-th non-linearity is the full memory posterior mean denoiser, and all the other denoisers are chosen to be the identity. The algorithm is dubbed AMP with Alternating Posteriors (AMP-AP), and its connection to BAMP is discussed at the end of Sec. \ref{sec:methods}. As evident from the smaller plots in the top right corner, AMP-AP (blue) matches the performance of BAMP and of the replica prediction as well.}
     
Lastly, BAMP is numerically unstable for low SNR. For 
the quartic potential and $\lambda =2.3$, 5 out of 50 trials 
do not reach the state evolution fixed point (and are thus discarded). Furthermore, BAMP's state evolution detaches from the replica prediction as the SNR gets smaller. Considering an initialization closer to the fixed point mitigates the issue. This instability is likely due to the fact that BAMP's state evolution corresponds to an auxiliary AMP that multiplies the number of iterations, see Sec. \ref{sec:methods}, and which thus amplify errors. 


\subsection{Universality of the rotational invariance assumption} 
We believe that our results apply beyond the rotational invariance assumption to cases where the eigenbasis of the noise is invariant under more restrictive transformations (such as permutations), or even ``quasi deterministic''. This intuition comes from recent works \cite{dudeja2022universality2,dudeja2022universality} showing that, when AMP or its linearized version are used,
the class of rotationally invariant matrices leads to the same performance as a much broader class of matrices (with same spectral density). While the existing literature considers a setting different than ours, 
this still suggests that our predictions should remain true more generally. {To confirm this, we plot 
in Fig. \ref{fig:BAMP_main} 
the performance of BAMP when the uniformly distributed matrix $\bO$ (i.e., the noise $\bZ$ eigenbasis) is replaced by \emph{(i)} the product of the Hadamard-Walsh matrix and a diagonal matrix with i.i.d.\ Rademacher entries, as in 
\cite{dudeja2022universality} (green squares), or \emph{(ii)} the eigenbasis of the covariance matrix for two popular datasets in computer vision and quantitative genetics, i.e., the CIFAR-10 ``plane'' class and the ``muscle skeletal'' GTEx dataset \cite{lonsdale2013genotype} (purple and orange markers, respectively, in the top-right plots). The excellent match 
clearly supports the universality of our predictions. Additional validations are contained in Sec. \ref{sec:7.2}. These results can be understood from the fact that \emph{any} eigenbasis $\bO$ is typical w.r.t. the Haar measure, so for a fixed instance, as long as $\bO$ is sufficiently independent from the eigenvalues, the universality should hold. This suggests that, in practice, our rotational invariance assumption effectively corresponds to assuming decoupling between eigenbasis and eigenvectors.} 

\section{Methods}\label{sec:methods}

\subsection{Outline of the replica computation}
The starting point of the replica method is the ``replica trick'' $$\lim_{N\to\infty}\EE\ln \mathcal{Z}(\mathbf{Y})/N=\lim_{n\to 0}\lim_{N\to\infty} \ln \EE \frac{1}{Nn}\mathcal{Z}^n(\mathbf{Y})$$
that implicitly assumes the commutation of the $n,N$ limits. Another key assumption is to consider $n\in\mathbb{N}$ in the computation and then assume an analytic continuation to $n$ close to $0_+$. 
The expectation is with respect to $\bY$ or equivalently the independent $\bO,\bX^*$; concerning $\mathbf{D}$ we only need that its empirical eigenvalue distribution converges weakly to $\rho$ and that it has asymptotically no outliers. 
When computing $\mathcal{Z}^n$ we get multiple integrals over $(\mathbf{x}_\ell)_{0\leq\ell\leq n}$, with $\mathbf{x}_0\equiv\mathbf{X}^*$, and a sum of $n$ Hamiltonians as in \eqref{Hamilt_main} in the exponential. Expanding the exponent we identify some order parameters: for $1\leq\ell\leq n$,
\begin{align*}
        &v_\ell:=\frac{\Vert\mathbf{x}_\ell\Vert^2}{N},\,\, 
        M_{(k)\ell}:=\frac{\bx_\ell^\intercal \bZ^k\bx_\ell}{N},\,\,\kappa_\ell:=\frac{\bx_\ell^\intercal \bZ\bx_0}{N},\,\, 
        m_\ell:=\frac{\bx_0^\intercal \bx_{\ell}}{N},
\end{align*}
After fixing these using the Fourier representation of the Dirac delta function, the replicated partition function reads
\begin{align*}
    \mathbb{E}\mathcal{Z}^n=\, \mathbb{E}_{\mathbf{Z},\bx_0}\int  \prod_{\ell=1}^n dP_X(\bx_\ell) d \btau_\ell d\hat\btau_\ell \,e^{-H_N(\btau_{\ell},\hat \btau_{\ell},\bx_{\ell};\bx_0,\bZ)},
\end{align*}
where $H_N(\btau,\hat \btau,\bx;\bx_0,\bZ):=Nh(\btau,\hat \btau)  + \bx^\intercal \bJ_1(\btau,\hat \btau,\bZ)\bx+\bx^\intercal\bJ_0(\btau,\hat \btau, \bZ)\bx_0$, and $\btau_\ell:=(v_\ell,M_{(1)\ell},\kappa_\ell,m_{\ell})$ with $\hat \btau_\ell$ being the Fourier conjugate. The definitions of $(h,\bJ_1,\bJ_0)$ can be found in Sec. \ref{sec:quadraticModel}.
This point is crucial as it allows us to write the $n$ Hamiltonians (one per $\mathbf{x}_\ell$) as at most quadratic functions of $\mathbf{x}_\ell$. Due to the quartic nature of the potential, the original $H_N$ would instead have quartic interactions, or higher order ones for polynomial $V$ of degree greater than four. Yet, by identifying the proper order parameters, a similar reduction to effective quadratic Hamiltonians would still be possible.

In $\mathbb{E}\mathcal{Z}^n$ the replicas are coupled in the system only through the expectation over the \emph{quenched} noise, that can be rewritten as an expectation over the Haar distributed noise eigenbasis $\mathbb{E}_\mathbf{O}$. The entire computation then boils down to the evaluation of an \emph{inhomogeneous log-spherical integral} that we introduced and defined as follows: let the matrices $\bC_{\ell\ell'}={\rm diag}((C_{i,\ell\ell'})_{i\le N})$, $\bC_i=(C_{i,\ell\ell'})_{\ell,\ell'\le n}$, and vectors $\bh_\ell=(h_{i,\ell})_{i\le N}$, $\bh_i=(h_{i,\ell})_{\ell\le n}$ all having bounded entries uniformly in $N$. The sequence $(\bh_i\in\mathbb{R}^{n},\bC_i\in\mathbb{R}^{n\times n})_{i\le N}$ is assumed to have an empirical law tending to that of the random variable $(\bh\in\mathbb{R}^{n},\bC\in\mathbb{R}^{n\times n})$. The inhomogeneous log-spherical integral is defined as
\begin{align}
    \mathcal{I}_N\!:=\!\frac1N \ln \mathbb{E}_{\bO} \exp\!\big(\sum_{\ell,\ell'\le n} (\bO\bx_\ell)^\intercal \bC_{\ell\ell'} \bO\bx_{\ell'}\!+\!\sum_{\ell\le n}(\bO\bx_\ell)^\intercal \bh_{\ell}  \big).\label{gene_sphInt_main}
\end{align}
Its limit depends only on the law of $(\mathbf{C},\mathbf{h})$ and on the overlaps $q_{\ell\ell'}:=\frac{1}{N}\bx_\ell^\intercal \bx_{\ell'}$, $\ell\le \ell'$, that we need to fix with additional Dirac deltas in addition to the previous order parameters. We find that $\lim_{N\to\infty}\mathcal{I}_N$ is expressed by a variational formula, see Sec. \ref{sec:sphInt}. This integral is a natural generalization of the standard spherical integral \cite{guionnet2005fourier} and thus may have an interest beyond the present model, in particular in random matrix theory or spin glasses. 

The final ingredient is a \emph{replica symmetric ansatz}, justified by the strong concentration-of-measure effects taking place in the Bayes-optimal setting \cite{nishimori01,barbier_strong_CMP}. It amounts to assume that all order parameters entering the model are independent of the replica index $\ell$. Finally, a saddle point yields an extremization over $\mathbb{R}^{13}$ of an effective action. Eqs. (\ref{var_princ_main}) and (\ref{eq:MMSE_main}) follow directly. 

Concerning the reduction from $13$ to $2$ order parameters (saddle point equations): this is possible thanks to a symmetry arising as a consequence of Bayes rule which is specific to the Bayes-optimal setting, and often called \emph{Nishimori identity}. It allows to ``interchange'' the ground-truth signal $\bX^*$ with a sample $\bx$ from the posterior \eqref{posterior_main} inside joint expectations over the posterior and data, see, e.g., \cite{barbier_strong_CMP}, and as a consequence to automatically fix the value of most order parameters.


\subsection{Auxiliary AMP and Onsager coefficients}
The Onsager coefficients $\{{\sf c}_{t, i}\}_{i\in [t], t\ge 1}$ are designed so that, 
conditioned on the signal, the empirical distribution of the iterate $\bff^t$ is Gaussian, namely $(\bff^1, \ldots, \bff^t)\stackrel{\mathclap{W_2}}{\longrightarrow} (F_1, \ldots, F_t) := \boldsymbol{\mu}_t X^* + \bW_t$, with $\bW_t\sim\normal (0, \boldsymbol{\Sigma}_t)$ for some mean vector $\boldsymbol{\mu}_t$ and covariance matrix $\boldsymbol{\Sigma}_t$. 
For the AMP in \cite{fan2022approximate}, this condition is enforced via the reduction to an \emph{auxiliary} AMP, which also allows to track the iterates of the original algorithm and yields the state evolution parameters, such as $\boldsymbol{\mu}_t$ and $\boldsymbol{\Sigma}_t$ above. This reduction crucially relies on splitting the matrix $\mathbf{Y}$ that multiplies the iterate $\mathbf{u}^t$, into the rank-one signal plus the noise matrix. In contrast, in \eqref{BAMP_main}, the iterate is multiplied by the pre-processed matrix $J(\bY)$, which \emph{cannot} be directly split in a similar fashion. Hence, we track all the contributions $(\mathbf{Y}^k\mathbf{u}^t)_{k\leq K}$, so that we can split them as $\mathbf{Y}^k\mathbf{u}^t=\bY\,\bY^{k-1}\bu^t=\frac\lambda N\,\bX^*\langle\bX^{*},\mathbf{Y}^{k-1}\mathbf{u}^t\rangle+\mathbf{Z}\mathbf{Y}^{k-1}\mathbf{u}^t$.


The key idea is to map the first $T$ iterations of \eqref{BAMP_main} to the first $K \times T$ iterations of an auxiliary AMP with iterates $(\tilde \bz^t, \tilde \bu^t)_{t\in [KT]}$ and denoisers $\{\tilde h_{t+1}\}_{t\in [K T]}$,
\begin{align}\label{eq:auxbd_main}
    \tilde{\bz}^t = \bZ \tilde{\bu}^t - \sum_{i\le t} \bar {\sf b}_{t, i}\tilde{\bu}^i, \ \tilde{\bu}^{t+1} = \tilde{h}_{t+1}(\tilde{\bz}^1, \ldots, \tilde{\bz}^t, \bu^1, \bX^*)\,,
\end{align}
whose state evolution can instead be deduced from \cite{fan2022approximate}. The denoisers $\{\tilde h_{t+1}\}_{t\in [KT]}$ of this multi-stage auxiliary AMP are chosen so that, for $t\in [T]$  and $\ell\in [K]$, 
\begin{align}
    \frac{1}{N}\|\tilde\bu^{K(t-1)+\ell} - \bY^{\ell-1}\bu^t\|_2^2\xrightarrow{N\to\infty}0\,.\label{eq:mapping1_main}
\end{align}
More specifically, for $t\in [T]$  and $\ell\in \{2, \ldots, K\}$, the denoiser $\tilde h_{K(t-1)+\ell}$ giving $\tilde\bu^{K(t-1)+\ell}$ is a linear combination of past iterates $\tilde \bu^1, \ldots, \tilde \bu^{K(t-1)+\ell-1}$ and of $\tilde \bz^{K(t-1)+\ell-1}$; furthermore, the coefficients of these linear combinations are chosen to ensure that $\tilde\bu^{K(t-1)+\ell}\approx \bY^{\ell-1}\bu^t$. Hence, by using \eqref{eq:auxbd_main} with $Kt$ in place of $t$, one gets $(\bY^{\ell}\bu^t)_{\ell\in [K]}$ from $\tilde \bz^{Kt}$ and $(\tilde \bu^{K(t-1)+\ell})_{\ell\in \{2, \ldots, K\}}$ (up to an $o_N(1)$). Thus, $J(\bY) \bu^t$ can be expressed as a linear combination of $(\tilde \bu^1, \ldots, \tilde \bu^{Kt}, \tilde \bz^{Kt})$, which in turn is a linear combination of \emph{(i)} the past iterates $\{\bu^i\}_{i\in [t]}$, \emph{(ii)}
the signal $\bX^*$, plus \emph{(iii)} independent Gaussian noise. By inspecting the coefficients of this linear combination, one deduces \emph{(a)} the Onsager coefficients $\{{\sf c}_{t, i}\}_{i\in [t], t\ge 1}$ (as the coefficients multiplying the past iterates $\{\bu^i\}_{i\in [t]}$), \emph{(b)} the mean $\mu_t$ (as the coefficient multiplying the signal $\bX^*$), and \emph{(c)} the covariance matrix $\boldsymbol{\Sigma}_t$ (as the covariance matrix of the remaining noise terms). Finally, by making $\tilde h_{Kt+1}$ depend on $g_{t+1}$, we enforce that $\tilde\bu^{Kt+1}\approx \bu^{t+1}$. 
The description of the auxiliary AMP is deferred to SI, Appendix \ref{appsubsec:aux}, and its state evolution follows in Appendix \ref{appsubsec:auxAMPSE}. 

{In summary, the derivation of BAMP's Onsager coefficients involves approximating $\{\bY^{k}\bu^t\}_{k\le K-1}$. This suggests an alternative choice of denoisers leading to the algorithm dubbed AMP-AP: for each batch of $K$ iterations, we pick linear denoisers in the first $K-1$ of them, as this allows to construct $\{\bY^{k}\bu^t\}_{k\le K-1}$; then, at the $K$-th iteration, we pick the posterior mean using \emph{all} the past iterates, as this --in principle-- allows to assemble the vectors $\{\bY^{k}\bu^t\}_{k\le K-1}$ to obtain $J(\bY)\bu^t$ as in BAMP. We note that AMP-AP does not require the coefficients of the polynomial $J(\bY)$, but it rather leaves to the posterior mean denoiser to learn them from the data. As such, it provides an efficient alternative to our proposed BAMP.}

\vspace{-3pt}

\vspace{-5pt}

\paragraph{Acknowledgements} {J. Barbier was funded by the European Union (ERC, CHORAL, project number 101039794). Views and opinions expressed are however those of the author(s) only and do not necessarily reflect those of the European Union or the European Research Council. Neither the European Union nor the granting authority can be held responsible for them. M. Mondelli was supported by the 2019 Lopez-Loreta prize. {The authors would like to thank the reviewers for the insightful comments and, in particular, for suggesting the BAMP-inspired denoisers leading to AMP-AP.}}

\paragraph{Codes} The codes used for this work are available \href{https://github.com/fcamilli95/Structured-PCA-}{here}.

\newpage
\part{Supplementary Information}

\section{Introduction, problem setting and main results}\label{sec:intro}

\subsection{Introduction and related works}

Given their ubiquitous appearance in the statistics literature, spiked matrix models, which were originally formulated as probabilistic models for principal component analysis (PCA) \cite{johnstone2001distribution}, are now a paradigm in high dimensional inference. Thanks to their universality features they, and their generalizations, find numerous applications in other central problems such as community detection \cite{SBM_abbe18,kim2017community}, group synchronization \cite{perry2018message,perry2016optimality}, sub-matrix localization or high-dimensional clustering \cite{lesieur2015mmse}; see \cite{lesieur2017constrained,perry2018optimality} for more applications.

In this paper we focus on the following estimation problem: a statistician needs to extract a rank-one matrix (the spike) $\mathbf{X}^*\mathbf{X}^{*\intercal}$, $\mathbf{X}^*\in\mathbb{R}^N$, from the data
\begin{align}\label{channel0}
    \mathbf{Y}=\frac{\lambda}{N}\mathbf{X}^*\mathbf{X}^{*\intercal}+\mathbf{Z}\in\mathbb{R}^{N\times N}
\end{align}
with some additive noise $\mathbf{Z}$. The positive parameter $\lambda$, referred to as signal-to-noise ratio (SNR), sets the strength of the signal with respect to that of the noise.  

The spectral properties of finite rank perturbations of large random matrices like \eqref{channel0} were intensively investigated in random matrix theory \cite{baik2005phase,baik2006eigenvalues,rmt-Peche2006,feral2007largest,capitaine2009largest,Nadakuditi_Jacobi,benaych2011eigenvalues,benaych2012singular,bai2012sample}, showing the presence of a spectral transition often called BBP transition (in reference to the authors of \cite{baik2005phase}): when $\lambda$ is large enough, the top eigenvalue of $\mathbf{Y}$ detaches from the bulk of the eigenvalue distribution. Its corresponding eigenvector has then a non trivial projection onto the sought ground truth $\mathbf{X}^*$, and can be used as its estimator. 

The problem has also been approached from the angle of Bayesian inference. In particular, besides the previous spectral estimator, there exists a whole family of iterative algorithms, known as approximate message passing (AMP), that can be tailored
to take further advantage of prior structural information known about the signal. AMP algorithms were first proposed for estimation in linear models \cite{kabashima2003cdma,Bayati_Montanari12,Bayati_Montanari11,Donoho09_Compressed_sensing,krzakala2012probabilistic,maleki2013asymptotic}, but have since been applied to a range of statistical estimation problems, including generalized linear models \cite{barbier2019optimal,ma2019optimization,maillard2020phase,mondelli2021approximate,rangan2011generalized,schniter2014compressive,sur2019modern} and low-rank matrix estimation \cite{it-PCAMontanari14,Rangan_iterative18,kabashima2016phase,lesieur2017constrained,montanari2017estimation,barbier2020all}.
An attractive feature of AMP is that under
suitable model assumptions, its performance in the high-dimensional limit is precisely characterized
by a succinct deterministic recursion called state evolution \cite{Bayati_Montanari11,bolthausen2014iterative,javanmard2013state}. Using the state evolution analysis, it has been proved that AMP achieves Bayes-optimal performance for some models \cite{it-PCAMontanari14,donoho2013information,montanari2017estimation,barbier2019optimal}, and a conjecture from statistical physics posits that for a wide range of estimation problems, AMP is optimal among polynomial-time algorithms.

The references mentioned above rely on the assumption of Gaussian identically and independently distributed (i.i.d.) noise $Z_{ij}\sim \mathcal{N}(0,1)$, under which the model identified by \eqref{channel0} is the well-known Wigner spiked model \cite{el2018estimation,BarbierMacris2019,alaoui2017finite,johnstone2001distribution}. This independence, or ``absence of structure'', in the noise  has many advantages from the theoretical point of view due to the numerous simplifications it generates.

%
In order to relax this property, we can seek inspiration from the statistical physics literature on disordered systems. An idea that was first brought forth in \cite{MultiSK_Contucci,Panchenko_multi13} for the Sherrington-Kirkpatrick model, and later imported also in high dimensional inference \cite{MSKNL,DBMNL,Guionnet_inhomogeneous_noise}, is that of giving an inhomogeneous variance profile to the noise matrix elements (we mention that this idea in inference is similar to the earlier definition of ``spatially coupled systems'' \cite{felstrom1999time,kudekar2011threshold} in coding theory, see \cite{XXT,XX_long} for its use in the present context). This procedure makes the $(Z_{ij})$ no longer identically distributed, but it leaves them independent. This an important step towards more structure in the noise (and therefore the data). Yet, the independence assumption is a rather strong one. Actually, \cite{Guionnet_inhomogeneous_noise} showed that for a broad class of observation models, as long as the independence assumption holds, the model is information-theoretically equivalent to one with independent Gaussian (possibly inhomogeneous) noise. 

One way to go beyond this last assumption is to consider noises that belong to the wider class of \emph{rotationally invariant matrices}. Since the appearance of the seminal works \cite{Marinari_ParisiI,Marinari_ParisiII,Parisi_Potters}, there has been a remarkable development in this direction, as evidenced by the rapidly growing number of papers on spin glasses \cite{bhattacharya2016high,opper2016theory,adatap,fan2021replica,maillard2019high,Foini_Kurchan_annealed} and inference \cite{benaych2011eigenvalues,benaych2012singular,gabrie2019entropy,gerbelot2020asymptotic,impact_sensing_Ma,takahashi2020macroscopic,fan2022approximate,zhong2021approximate,venkataramanan2022estimation} that try to take into account structured disorder, including the present one. Indeed, we hereby consider a spiked model in which the noise matrix $\mathbf{Z}$ is drawn from an orthogonal matrix ensemble different from the Gaussian orthogonal ensemble (which is the only rotationally invariant ensemble such that the matrix entries are independent). Intuitively, the presence of dependencies in the noise should be exploitable by an algorithm that is sharp enough to see patterns within it and use them to retrieve the sought rank one matrix more efficiently. Going in that direction, in \cite{fan2022approximate} the author proposed a version of AMP designed for rotationally invariant noises (using earlier ideas of \cite{adatap,opper2016theory}) and provided also a rigorous state evolution analysis for it. Furthermore, in a recent work \cite{price_of_ignorance22}, part of the authors performed a rigorous analysis of a Bayes estimator and an AMP, both assuming Gaussian noise, whereas the actual noise in the data was drawn from a generic orthogonal matrix ensemble. However, besides intuition and the mentioned works, to our best knowledge there is little theoretical understanding of the true role played by noise structure in spiked matrix estimation and more generically in inference. In particular, prior to our work there was no theoretical prediction of optimal performance to benchmark practical inference algorithms.

\subsubsection*{Organization.} 

The end of this section 
properly defines the model and the quartic random matrix ensemble we 
consider.
In Section \ref{appendix_spherical_generalized}, we define and analyze an integral dubbed \emph{inhomogeneous spherical integral}, that 
will play an essential role in the analysis. For those interested mainly in the main results this section can be skipped at first reading. Section \ref{sec:replica} contains the core information-theoretic analysis based on the replica method. We also show at the end of it that, for rotationally invariant priors, the spectral estimator is Bayes-optimal in the MMSE sense. Next, in Section \ref{sec:sub-opt}, we analyze both the fixed point performance of the previously proposed AMP for structured PCA \cite{fan2022approximate} and our replica prediction for the MMSE. We deduce that, in general, the AMP in \cite{fan2022approximate} is sub-optimal, and we provide an explanation for why this is the case. Using the theory of adaptive TAP equations \cite{adatap}, Section \ref{sec:AdaTAPtoward} lays the foundations for defining an optimal AMP: the main outcome is an optimal pre-processing polynomial function that depends on the statistical properties of the noise and has to be applied to the data, in order to achieve 
Bayes-optimality. Section \ref{sec:AMP} demonstrates that, by exploiting this pre-processing function, a novel AMP can be written down which \emph{does match} the MMSE predicted by the replica theory. This algorithm comes with a scalar state evolution recursion which rigorously tracks its performance in the limit of large size. The Onsager reaction coefficients of our AMP are different from those in \cite{fan2022approximate} and their calculation, as well as the state evolution analysis, requires new ideas. To highlight these differences and emphasize the match with the replica MMSE, this new algorithm is dubbed Bayes-optimal AMP, or \emph{BAMP}. {Furthermore, the structure of BAMP suggests a choice of the denoisers in the existing AMP which differs from those previously proposed in \cite{fan2022approximate,zhong2021approximate}. We refer to this algorithm as AMP with Alternating Posteriors (AMP-AP), since it alternates linear denoisers to a full posterior mean denoiser using all the previous iterates.  
In the final Section \ref{sec:numerics} we provide a numerical confirmation of our theoretical predictions, and we show that both BAMP and AMP-AP match the replica MMSE.}
Appendix \ref{sec:approx_potential} is dedicated to showing that studying polynomial potentials acting on the eigenvalues on the noise is sufficient in order to study more general ensembles. In Appendix \ref{sec:EM}, we provide expectation-maximization (EM) equations that learn the optimal pre-processing function to be used by BAMP, when noise statistics are not known.
In the last technical Appendix \ref{appsec:BoptAMP}, we gather the proofs of the various results needed to reach the state evolution of our BAMP algorithm.

\subsubsection*{Notations.} 
Bold notations are reserved for vectors and matrices. By default a vector $\bx$ is a column vector, and its transpose $\bx^\intercal$ is therefore a row vector. Thus the usual $L_2$ norm $\|\bx\|^2=\bx^\intercal \bx$ and $\bx\bx^\intercal$ is a rank-one projector. The notation $\bx\stackrel{\mathclap{W_2}}{\longrightarrow} X$ denotes convergence of the empirical distribution of the random vector $\bx$ to the random variable $X$ in Wasserstein-2 distance. Symbol $\propto$ means ``equality up to a constant'' (often, a normalization constant) and $:=$ is an equality by definition. $\Tr$ is the usual trace operator. For a vector $\bx$, the matrix ${\rm diag}(\bx)$ is diagonal with $\bx$ on its diagonal. For a diagonal matrix $\bA$ and a function $F:\mathbb{R}\mapsto \mathbb{R}$ the matrix $F(\bA)$ is diagonal with $F$ applied component-wise to each diagonal entry of $\bA$. A function $F$ applied to a real symmetric $N\times N$ matrix diagonalizable as $\bM=\bU\bA\bU^\intercal$ acts in the standard way: $F(\bM):=\bU F(\bA)\bU^\intercal$. $\EE_A$ is an expectation with respect to the random variable $A$; $\EE$ is an expectation with respect to all random variables entering the ensuing expression. For a function $F$ of one argument we denote $F'$ its derivative. Notations like $i\le N$ always implicitly assume that the index $i$ starts at $1$. Notation $[t]:=\{1,2,\cdots,t\}=\{i\le t\}$. Powers for vectors apply componentwise (this is however \emph{not} the case for matrices). We often compactly write $\EE(\cdots)^2=\EE[(\cdots)^2]\ge (\EE(\cdots))^2$ and similarly for other functions, we denote equivalently $\EE[f(\cdots)]$ and $\EE f(\cdots)$. Matrix $I_N$ is the identity of size $N$.

\subsection{Probabilistic model of PCA with structured noise}
Consider a vector $\mathbf{X}^*=(X_i^*)_{i\leq N}$ whose components are drawn i.i.d. from a given distribution $P_X$ with support bounded uniformly in $N$. Two cases will be considered: the factorized case $$dP_X(\bX^*)=\prod_{i\leq N}dP_X(X^*_i)=\prod_{i\leq N}P_X(X^*_i) dX^*_i,$$ and the case where $dP_X$ is the uniform measure over the $N$-sphere of radius $\sqrt N$. If not specified the first case is assumed. We will always consider priors with unit second moment $$\int dP_X(x)x^2=1.$$ This is just a convention as if one wants to consider a different normalization, it can simply be included through a proper rescaling of the SNR $\lambda$.

The inference task we are interested in is the retrieval of the rank-one ``spike'' $\mathbf{P}^*:=\mathbf{X}^*\mathbf{X}^{*\intercal}$ from the following observed matrix
\begin{align}\label{channel}
    \mathbf{Y}=\frac{\lambda}{N}\mathbf{P}^*+\mathbf{Z},
\end{align}
where $\mathbf{Z}$ is a unknown noise matrix, $\lambda \ge  0$ is the SNR. Whenever $\bZ$ is a Wigner matrix this model corresponds to the usual Wigner spike model. But here we no longer assume that the noise is \emph{unstructured} (namely, has independent entries). More specifically, we will assume that is drawn from a certain orthogonal rotationally invariant random matrix ensemble defined by a potential $V:\mathbb{R}\mapsto \mathbb{R}$ and a density (with normalization constant $C_V$)
\begin{align}\label{Z-ensemble}
    dP_Z(\mathbf{Z})= C_V d\mathbf{Z}\exp\Big(-\frac{N}{2}\Tr V(\mathbf{Z})\Big).
\end{align}
Rotational invariance means that $\bZ$ equals in distribution $\bU^\intercal \bZ\bU$ for any orthogonal matrix $\bU$ (this follows from the trace in the exponent) \cite{potters2020first}. More precisely, when changing variables from matrix $\bZ$ to eigenvalues $\bD$ and eigenbasis $\bO$ via $\bZ=\bO^\intercal \bD\bO$ we have
\begin{align}
    dP_Z(\bD,\bO)= C_V d\bO \,d\bD\exp\Big(-\frac{N}{2}\Tr V(\bD)\Big)\prod_{i<j} |D_i-D_j|.\label{densityOD}
\end{align}
The measure $d\bO$ is the Haar measure, i.e., uniform measure over the orthogonal group $\mathbf{O}(N)$, and the last term coupling all eigenvalues in a pairwise long-range fashion is the Vandermonde determinant.
Note that only the special case $V(x)=x^2/(2\sigma)$ corresponding to the Gaussian orthogonal ensemble induces independent (Gaussian distributed) matrix entries (up to symmetry). Any other potential generates dependencies among matrix elements and thus \emph{structure}. E.g., if we take $V(x)=x^4/4$,
\begin{align}\label{Z-ensemble}
    dP_Z(\mathbf{Z})= C_V d\mathbf{Z}\prod_{i,j,k,l} \exp\Big(-\frac N 8 Z_{ij}Z_{jk}Z_{kl}Z_{li}\Big)
\end{align}
which clearly is not factorizable over matrix entries.

We now introduce the Bayesian framework which we are going to analyse. Let the projector $\mathbf{P}:=\mathbf{x}\mathbf{x}^{\intercal}$. This allows us to write the posterior measure of the inference problem:
\begin{align}
    dP_{X\mid Y}(\mathbf{x}\mid \bY)=\frac{C_V}{P_Y(\bY)}dP_X(\mathbf{x})\exp\Big(-\frac{N}{2}\Tr V\Big(\bY-\frac{\lambda}{N}\mathbf{P}\Big)\Big).\label{posterior}
\end{align}
Because both the prior $P_X(\mathbf{x})$ matches the density of the signal and the likelihood $P_{Y\mid X}$ matches the noise density $P_Z$ and moreover the SNR $\lambda$ is known, the posterior written above is the ``correct'' one and we are in the \emph{Bayesian-optimal setting}. Studying the limits of inference in this setting draws a fundamental line between what is information-theoretically possible and what is not in terms of performance of inference. The evidence reads
\begin{align}
    P_Y(\bY)=C_V\int dP_X(\mathbf{x})\exp\Big(-\frac{N}{2}\Tr V\Big(\bY-\frac{\lambda}{N}\mathbf{P}\Big)\Big).
\end{align} 

One of the main object of interest is the \emph{free entropy} (or minus the \emph{free energy}), which is nothing else than minus the Shannon entropy of the data: 
\begin{align}
F_N(\bY):= -H(\bY)=\EE\ln P_Y(\bY).  \label{freeEntropy}  
\end{align}
Therefore the free entropy is related to the mutual information by an additive constant corresponding to the entropy of the noise, and is therefore simply computed (while the free entropy is not):
\begin{align}
 I(\bP^*; \bY)&=-F_N(\bY)-H(\bY\mid \bX^*)\nn
 &=   -F_N(\bY)-H(\bZ)\nn
 &=-F_N(\bY)-\ln C_V+\frac{N}{2}\mathbb{E}\Tr  V(\mathbf{Z}).
\end{align}
 Using the explicit form of the observation model \eqref{channel} the free entropy reads
\begin{align}\label{free-energy}
        F_N(\mathbf{Y})&=\mathbb{E}\ln\int dP_X(\mathbf{x})\exp\Big(-\frac{N}{2}\Tr\Big[ V\Big(\mathbf{Z}+\frac{\lambda}{N}(\mathbf{P}^*-\mathbf{P})\Big)-V(\mathbf{Z})\Big]\Big)\nn
        &\qquad+\ln C_V-\frac{N}{2}\mathbb{E}\Tr  V(\mathbf{Z}).
\end{align}
We extracted the noise entropy in the second line so that we can isolate the mutual information and to make the argument of the integrated exponential of order $N$. In this way the problem is naturally mapped onto a statistical mechanics model with extensive Hamiltonian given by minus the log-likelihood:
\begin{align}
    H_N(\mathbf{x};\mathbf{Z},\bX^*)&:=\frac N2\Tr\Big[ V\Big(\mathbf{Z}+\frac{\lambda}{N}(\mathbf{P}^*-\mathbf{P})\Big)-V(\mathbf{Z})\Big]\label{Hamilt}.
\end{align}
Indeed, our Hamiltonian can be rewritten as
\begin{align}
    \frac12\Tr(\mathbf{P}^*-\mathbf{P})\int_0^\lambda dt V'\Big(\mathbf{Z}+\frac{t}{N}(\mathbf{P}^*-\mathbf{P})\Big).\nonumber
\end{align}
The difference between the two projectors has only two eigenvalues of order $N$ and the matrix inside the potential derivative has $O(1)$ eigenvalues, hence the previous is of $O(N)$ too. The free entropy is thus directly linked to the expected log-partition function associated to this Hamiltonian:
\begin{align}
  \EE\ln \mathcal{Z}(\bY):=\EE\ln \int dP_X(\mathbf{x})\exp\big(-H_N(\mathbf{x};\mathbf{Z},\bX^*)\big) \label{partFunc} .
\end{align}
The notation $;$ in $H_N(\mathbf{x};\mathbf{Z},\bX^*)$ emphasizes that $\mathbf{Z},\bX^*$ are quenched variables while $\bx$ fluctuates according the Gibbs-Boltzmann distribution associated to this Hamiltonian (i.e., the posterior). The same notation with same meaning for Hamiltonians will be used later on.

 
{
\subsection{Concrete examples: the quartic and sestic ensembles}\label{sec:quartic}
}
Analysing this model for a generic potential $V$ is \emph{possible} through the novel methodology presented in this paper. But as it will become apparent, if we take a generic polynomial potential $V$, the higher the order of this polynomial, the more technical and cumbersome it becomes. So for the sake of pedagogy we focus in the present contribution on a very concrete example of non trivial correction to the i.i.d. noise hypothesis. As a matter of fact, the simplest inference problem with correlated noise elements is that with the quartic matrix potential, {that is,} for two positive real numbers $\mu$ and $\gamma${,}
\begin{align}
    \label{quartic_potential}
    V(x)=\frac{\mu}{2} x^2+\frac{\gamma}{4}x^4.
\end{align}
This was first studied by Br\'ezin et al in \cite{quartic_Parisi} to study the planar approximation of quantum field theories with large internal symmetry groups. We could have also considered a non-symmetric potential with a cubic term too, but for simplicity we restrict ourselves to that case as symmetry will slightly simplify the computations (but there is no barrier to applying our methods to that a more general, possibly non-even, potential).

\begin{figure}[t!] 
\begin{center}
 \includegraphics[width=0.7\linewidth,trim={0 0 0 0},clip]{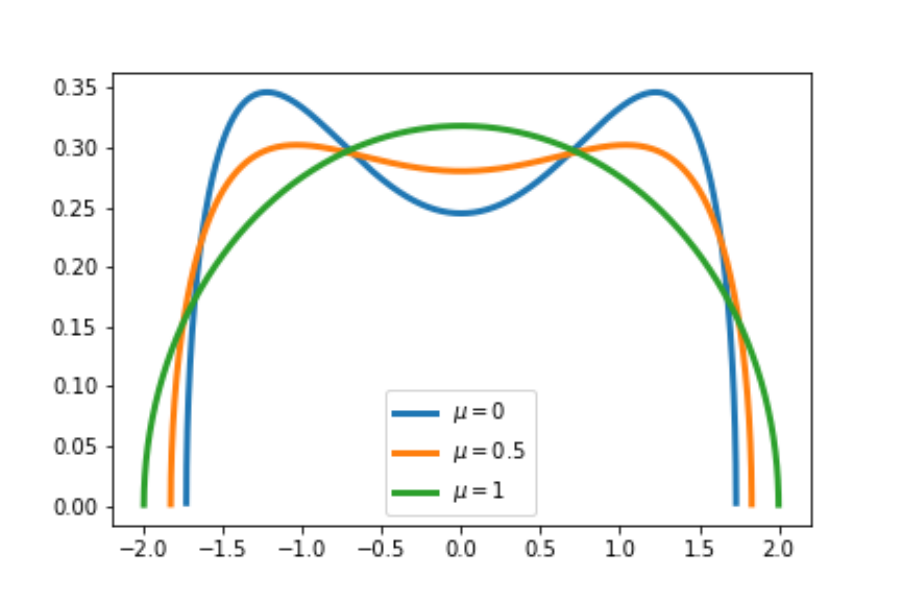}
\caption{Asymptotic spectral density \eqref{densityRho} of the random noise ensemble defined by the potential \eqref{quartic_potential} from less structured (with independent entries) at $(\mu=1, \gamma=0)$, corresponding to the standard semi-circle law, to the more structured $(\mu=0,\gamma=16/27)$ (recall relation \eqref{gamma(mu)}.} \label{fig:density}
\end{center}
\end{figure}

The matrix ensemble defined by \eqref{quartic_potential} has a known Stieltjes transform $\mathcal{S}$ and asymptotic eigenvalue density $\rho$, see, e.g., \cite{potters2020first}: if $\bZ$ is a sequence of matrices of increasing size $N$ drawn from \eqref{Z-ensemble} with the above quartic potential and whose sequence of eigenvalues is $(D_i)_{i\le N}$, then
\begin{align}
    &\frac1N\sum_{i\le N} \delta_{D_i,x}\xrightarrow{N\to\infty} \rho(x)=\frac{1}{2\pi}(\mu+2a^2\gamma+\gamma x^2)\sqrt{4a^2-x^2},\label{densityRho}\\
    &\mathcal {S}(z) = \int \frac{d\rho(x)}{z-x} =\frac{1}2\big(\mu z+\gamma z^3-(\mu+2a^2\gamma+\gamma z^2)\sqrt{z^2-4a^2}\big),
\end{align}
for a $z$ lying outside of the support of $\rho$, and where
\begin{align}
    &a^2:=\frac{\sqrt{\mu^2+12\gamma}-\mu}{6\gamma}.
\end{align}
It is evident that when $\gamma\to 0^+$ one has $a^2\to 1/\mu$
and consequently $\rho(x)\to\rho_{\rm sc}(x)$ the standard semi-circle law, see Figure \ref{fig:density}. In principle the choice of $\gamma$ and $\mu$ is totally free, as long as\footnote{We use implicitly the convexity of the potential, which requires $\mu,\gamma>0$, to obtain the density of eigenvalues \cite{potters2020first}. But we believe that this condition can be relaxed if one can get an associated well-defined asymptotic spectral density and that our analysis would still hold.} $\gamma>0$. However, we are interested in a noise with unit variance in order to be able to make a meaningful comparison with models with unstructured noise. By enforcing this unitarity constraint one finds a relation between $\gamma$ and $\mu$:
\begin{align}
    \gamma= \gamma(\mu)=\frac1{27}\big(8-9\mu+\sqrt{64-144\mu+108\mu^2-27\mu^3}\big).\label{gamma(mu)}
\end{align}
With this choice one can check that $$\int d\rho(x)x^2=1 \quad \mbox{for any} \quad \mu\in[0,1].$$  When $(\mu=1, \gamma(1)=0)$ we recover the pure Wigner case already analyzed in great detail. On the contrary $(\mu=0,\gamma(0)=16/27)$ corresponds to a purely quartic case with unit variance, and to the ``most structured'' ensemble in {parametric class of ensembles.} 
Therefore, $\mu$ can be thought of as a parameter allowing to interpolate between unstructured and structured noise ensembles. Even for this simple family of potentials, as soon as $\mu <1$, neither the Bayes-optimal nor the algorithmic limits of inference are known (except for those of a simple spectral algorithm, see \cite{benaych2011eigenvalues}).

{
As an additional example, we push our analysis further to the sestic matrix potential
\begin{align}\label{eq:sestic_potential}
    V(x)=\frac{\mu}{2}x^2+\frac{\gamma}{4}x^4+\frac{\xi}{6}x^6\,.
\end{align}
As for the quartic ensemble, using the same techniques illustrated in \cite{potters2020first}, we were able to derive the Stieltjes transform $\mathcal{S}_6$ and asymptotic eigenvalue density $\rho_6$ for this ensemble:
\begin{align}
    &\frac1N\sum_{i\le N} \delta_{D_i,x}\xrightarrow{N\to\infty} \rho_6(x)=\frac{\big[
\mu+2a^2\gamma+6a^4\xi+(\gamma +2a^2\xi)x^2+\xi x^4\big]\sqrt{4a^2-x^2}}{2\pi},\label{densityRho_6}\\
    &\mathcal {S}_6(z)  =\frac{1}2\big(\mu z+\gamma z^3+\xi z^5-\big[
\mu+2a^2\gamma+6a^4\xi+(\gamma +2a^4\xi)z^2+\xi z^4
\big]\sqrt{z^4-4a^2}\big),
\end{align}
for a $z$ lying outside of the support of $\rho$, and $a^2$ solving the cubic equation
\begin{align}
    10 \xi x^3+3\gamma x^2 +\mu x-1=0\,.
\end{align}
As we are interested in the most structured and accessible case, we set the coefficients of the lower order monomials to $0$: $\mu=\gamma=0$. In this case, after imposing the constraint $\int d\rho_6(x) x^2=1$, one readily gets
\begin{align}
    \xi=\frac{27}{80}\,,\quad a=\sqrt{\frac{2}{3}}\,.
\end{align}

}

\subsection{Main results}\label{sec:results}

{Our main contributions can be divided in two categories: those on the fundamental, information-theoretic, limitations of inference in structured PCA, and new algorithmic ideas.

Our information-theoretic results boil down to low-dimensional explicit variational formulas for, firstly, the asymptotic limit of the mutual information between the spike and the data. This limit contains the location of the fundamental phase transition in the problem, which corresponds to its non-analytic points (as a function of the signal-to-noise ratio). Such a transition, often called information-theoretic phase transition, defines the limit below which inference is typically poor or even impossible. Secondly, we obtain a formula for the minimum mean-square error (MMSE), which represents the fundamental lower bound on the mean-square error any algorithm, efficient or not, can possibly achieve. Having access to the MMSE then provides a clear benchmark for any practical algorithm. Analytically computing these quantities in the present model was not explored prior to our work. In absence of low-dimensional asymptotic formulas such as those given below, practitioners aiming at approximating them would rely on exact sampling procedures (such as Monte-Carlo Markov Chain) and wait long enough for convergence, which is not guaranteed in reasonable times in certain regions of the phase diagram (often at low signal-to-noise ratio). Thus, in general, only analytical characterizations are able to quantify the MMSE in the whole region of parameters of the problem.}

Complementary to that, we will introduce novel algorithmic ideas allowing to match these Bayes-optimal limits efficiently. Both type of results require conceptual insights and technical advances that we emphasize. We gather here these results and state them informally; we refer to the main sections for precise statements.

\subsubsection*{Information-theoretic results}
\begin{itemize}
    \item Our analysis of the information-theoretic (Bayes-optimal) performance based on the non-rigorous replica method yields first a low-dimensional variational formulation for the free entropy (log-partition function) of the model when $P_X$ is factorized:
\end{itemize}

\begin{result}[Free entropy]
{For the quartic potential,} the free entropy (i.e., minus Shannon entropy of the data) verifies in the limit of large size the following characterization:
\begin{align*}
   \frac1NF_N(\bY)=- \frac1N H(\bY)\xrightarrow{N\to\infty}{\rm extr}\, f_\rho(\btau)
\end{align*}
where $\btau \in\mathbb{R}^{13}$ and for an explicit real-valued function $f_\rho:\mathbb{R}^{13}\mapsto \mathbb{R}$ depending on the noise asympotic spectral density $\rho$. See \eqref{free_energy_RS} for the complete statement. Here and everywhere in the paper ${\rm extr}$ stands for the following ``extremization'' procedure: if $f:\mathbb{R}^{k}\mapsto \mathbb{R}$ then $${\rm extr}\, f(\btau):= f(\btau_*) \ \ \mbox{where}\ \ \btau_*:= \underset{\{\btau \in\mathbb{R}^{k}\,:\,\nabla f(\btau)=\boldsymbol{0}\}}{{\rm argmax}} f(\btau).$$
\end{result}

We will see that despite the apparent mess, the $13$-dimensional system of equations defining $\btau_*$ will reduce to a much simpler $2$-dimensional one (see eqs. \eqref{m_replica_simplified}--\eqref{hatmReplica}) thanks to special symmetries inherent to the Bayes-optimal nature of our analysis, and known as Nishimori identity in physics, which is a simple consequence of Bayes rule \cite{barbier_strong_CMP}:\\

\noindent {\bf Nishimori identity.}\ \ For any bounded function $f$ of the signal $\bX^*$, the data $\bY$ and of conditionally i.i.d. samples from the posterior $\bx^j\sim P_{X\mid Y}(\,\cdot \mid \bY)$, $j=1,2,\ldots,n$, we have that
\begin{align}
    \EE\langle f(\bY,\bX^*,\bx^2,\ldots,\bx^{n})\rangle=\EE\langle f(\bY,\bx^1,\bx^2,\ldots,\bx^{n})\rangle
\end{align}
where the bracket notation $\langle \,\cdot\,\rangle$ is used for the joint expectation over the posterior samples $(\bx^j)_{j\le n}$, $\EE$ is over the signal $\bX^*$ and data $\bY$. \\

The reduction of the replica saddle point equations thanks to this identity is done in Section \ref{sec:RSFP}. As a consequence only two scalar quantities will remain after reduction, one denoted $m$ and called ``magnetization'' quantifying the overlap between the minimum mean-square error (MMSE) estimator and the signal.

\begin{itemize}
    \item From the solution of this variational problem we deduce our second main result, namely, an asympotically exact expression for the minimum mean-square error of inference of the hidden spike with factorized prior:
\end{itemize}

\begin{result}[Minimum mean-square error]
 The minimum mean-square error verifies
\begin{align*}
    \lim_{N\to\infty}\frac{1}{2N^2}\EE\| \bX^*\bX^{*\intercal}-\EE[\bX^*\bX^{*\intercal}\mid \bY]\|_{\rm F}^2=\frac{1}2(1-m^2)
\end{align*}
where $m$ is one component of the solution $\btau_*$ to the variational problem for the free entropy, studied in Section \ref{sec:RSFP}.
\end{result}
{
All the above results hold in an analogous form also for the sestic potential \eqref{eq:sestic_potential} with $\mu=\gamma=0$. The variational principle outlined in Result 1 involves more order parameters, but after the application of the Nishimori identities the saddle point equations can be reduced to only 5, see Section \ref{sec:sestic_replica}. The replica prediction for the MMSE in Result 2 appears still in the same form, with $m$ being again one of these stationary solutions of the saddle point equations.

}

The main technical and conceptual novelties which lead to these formulas are:
\begin{itemize}
    \item To the best of our knowledge, we provide the first adaptation of the replica method to the analysis of the fundamental limits of inference in a model with a \emph{noise} having strongly dependent random entries (instead of a measurement operator, or matrix of covariates, in a regression setting). See Section \ref{sec:replica}.
    \item If the structure of the noise (i.e., its statistical properties) is encoded by a polynomial potential $V$ of order $K+1$, then this induces in the posterior distribution $k$-wise interactions between the signal's estimator entries, for all $k\le K+1$. Said differently, the underlying factor graph is an hypergraph with hyperedges of degrees $K+1,K,\ldots,1$. However, we discovered that by exploiting the low-rank structure of the signal, all these interactions can be reduced to effective pair-wise interations. This allows to reduce the model to an Ising model more convenient for theoretical analysis (a similar reduction is useful for algorithmic approaches too, see next section). The reduction we propose is general and systematic for low-rank signals corrupted by rotational invariant noise matrices. See Section \ref{sec:quadraticModel}.
    \item Our analysis can be mainstreamed once we have identified a key integral that we refer to as the \emph{inhomogeneous spherical integral}. This exactly solvable integral is a generalization of the standard low-rank spherical integral appearing in random matrix theory (as it is related to the R-transform) \cite{potters2020first}, in spin-glasses \cite{Parisi_Potters,adatap,opper2016theory,fan2021replica,barbier2021marginals}, the theory of large-deviations for matrix-valued stochastic processes \cite{guionnet2005fourier,guionnet2002large} and matrix models in high-energy physics \cite{itzykson1980planar,kazakov2000solvable,guionnet2004first}. Given the breadth of applications of this integral, we foresee that the generalization we propose and analyze in Section \ref{appendix_spherical_generalized} may have applications well beyond the present setting, for the study of models where rotationally invariant matrices with non-independent matrices appear.
    \item Another important conclusion from our analysis is the fact that for signals $\bX^*$ whose law is  rotation-invariant (such as Gaussian or uniformly spherically distributed), the simple spectral PCA procedure of \cite{benaych2011eigenvalues} is Bayes-optimal: 
\end{itemize}

\begin{result}[Optimality of spectral PCA for rotation-invariant priors]
    Let $\bX^*$ be a standard Gaussian vector or uniformly sampled on the sphere of radius $\sqrt N$. Then its inference from $\bY$ can be optimally achieved from the naive spectral algorithm that constructs an estimator $C\boldsymbol \nu \boldsymbol \nu^\intercal$ of $\bP^*$ from the eigenvector $\boldsymbol \nu=\boldsymbol \nu(\bY)$ of $\bY$ with leading eigenvalue $\lambda_{\rm max}$ and that is then properly rescaled by a certain factor $C=C(\lambda,\rho)$, see \cite{benaych2011eigenvalues}.
\end{result}
    This is verified both by the replica method and an exact computation based on Gaussian integration and a saddle point method, see Section \ref{sec:PCAopt}. We remark that this statement is incorrect for other priors $P_X$.

\subsubsection*{Algorithmic results}

On the algorithmic side our contributions are the following:

\begin{itemize}
    \item We analytically show that the existing Approximate Message Passing algorithms \cite{fan2022approximate,zhong2021approximate}, whose iterates are based on the data matrix $\bY$, do not saturate the Bayes-optimal performance predicted by our replica theory. See Section \ref{sec:sub-opt}.
    \item We employ in Section \eqref{sec:AdaTAPtoward} the AdaTAP formalism of Opper et al \cite{adatap} to analyze the model from the algorithmic perspective. What the analysis shows is that, like in the replica method, one can reduce the model with interactions of order higher than two to a pure quadratic Ising model with an effective interaction matrix $\bJ(\bY)$ which is a non-trivial matrix polynomial of the data $\bY$. This explains the reason why the previously proposed AMP algorithms are sub-optimal: the data $\bY$ is \emph{not} the best choice of matrix to use in the AMP iterates, despite being the most natural one. The Bayes-optimal choice is instead $\bJ(\bY)$ obtained from our theory, which cannot be guessed a-priori. We informally state this fact as one of our main results:
\end{itemize}

\begin{result}[Bayesian-optimal processing of data and optimal AMP]
Consider the matrix estimation model under structured noise \eqref{channel}. Given the observed matrix of data $\bY$, the optimal choice of matrix to use in a Bayesian inference algorithm such as AMP is \emph{not} $\bY$ but instead a proper polynomial of it, i.e., $\bJ(\bY)=\sum_{k\le K}c_k\bY^k$, with coefficients $(c_k)_{k\in [K]}$ depending on $V$. For example, when the potential $V$ is given by \eqref{quartic_potential} we show in Sections \ref{sec:adatapFreeEn} and \ref{statCondADATAP} that the optimal choice is
\begin{align*}
    \bJ(\bY)=\mu\lambda \bY-\gamma \lambda^2\bY^2+\gamma\lambda\bY^3.
\end{align*}
Employing this matrix in the AMP iterates leads to a Bayesian-optimal inference algorithm whose complexity scales as the dimension $N$, see the result below.
In Section \ref{sec:preprocessing_sestic}, with the same techniques, we also derive the optimal pre-processing in the case of a {pure sestic potential} $V(x)=\xi x^6/6$: $J_6(x)=\xi \lambda x^5-\xi\lambda^2 x^4-\xi\lambda^2 x^2$.
\end{result}
\begin{itemize}
    \item After having defined the Bayesian-optimal AMP recursion, we provide a rigorous state evolution recursion to track its asymptotic performance. We highlight that, since the data matrix $\bY$ is replaced by the polynomial $\bJ(\bY)$, we cannot apply the state evolution result of \cite{fan2022approximate}. More specifically, the Onsager correction terms will have a different form than the ones of \cite{fan2022approximate}, and their derivation requires a novel analysis.
\end{itemize}

\begin{result}[State evolution of the Bayes-optimal AMP (BAMP)]
    Consider the Bayesian-optimal Approximate Message Passing (BAMP) algorithm defined by the recursion
    \begin{align}
           \bff^t = \bJ(\bY) \bu^t - \sum_{i\le t}{\sf c}_{t, i}\bu^i, \quad \bu^{t+1} = g_{t+1}(\bff^t), \quad t\ge 1.\label{BAMP}
    \end{align}
    When a proper choice of coefficients $\{{\sf c}_{t, j}\}_{j
\in [t]}$ is considered, for a large family of functions $(g_t)_{t\ge 1}$ and $\psi$, the following holds almost surely:
\begin{align*}
& \lim_{N \to \infty} \frac{1}{N} \sum_{i\le N} \psi(u_i^1, \ldots, u_i^{t+1}, f_i^1, \ldots, f_i^t, X^*_i) = \mathbb E \,\psi(U_1, \ldots, U_{t+1}, F_1, \ldots, F_t, X^*) .
\end{align*}
Equivalently the joint empirical distribution over the $N$ rows of the $N\times (2t+2)$ matrix $(\bu^1, \ldots, \bu^{t+1}, \bff^1, \ldots, \bff^t, \bX^*)$ converges in a certain sense to the $(2t+2)$-dimensional random vector $(U_1, \ldots, U_{t+1}, F_1, \ldots, F_t, X^*)$ when $N$ increases. Here $$U_{i+1}=g_{i+1}(F_t) \ \ \mbox{and} \ \ (F_1, \ldots, F_t)=(\mu_1,\ldots,\mu_t) X^*+(W_1,\ldots, W_t)$$ with $(W_i)_{i\le t}$ a multivariate Gaussian vector whose covariance as well as $(\mu_i)_{i\le t}$ can be computed via a deterministic state evolution recursion. 
\end{result}

The precise rigorous statement can be found in Section \ref{sec:AMP}. The idea of the argument is to construct an auxiliary AMP which tracks the quantities $(\bY^{j-1}\bu^t)_{t\ge 1, j\le K-1}$. By decomposing the iterates of this auxiliary AMP into a component aligned with previous iterates, a component in the direction of the signal and independent Gaussian noise, we obtain the form of the Onsager correction and the state evolution. From this result we can rigorously predict the performance of the novel AMP algorithm we propose. The optimality of the pre-processed matrix $\bJ(\bY)$ and associated AMP is then confirmed by the perfect matching of the fixed point of the state evolution recursion tracking the AMP mean-square error and our replica prediction for the MMSE. 

Some important remarks are in order. First, we emphasize that the BAMP algorithm \eqref{BAMP} we propose is \emph{not} the usual AMP of \cite{fan2022approximate} where the data matrix $\bY$ is just replaced by the pre-processed matrix $\bJ(\bY)$. Indeed, the correct Onsager coefficients $\{{\sf c}_{t, i}\}$ entering BAMP require a novel type of ``multi-stage'' state evolution recursion which is completely different from the one in \cite{fan2022approximate}, see Section \ref{sec:AMP}. The novel acronym we introduce emphasizes that crucial distinction.

Secondly, it is true that our replica prediction for the MMSE is non-rigorous. However, our state evolution analysis of BAMP is fully rigorous (just like the analysis of the AMP in \cite{fan2022approximate}). By comparing their asymptotic fixed point performance by state evolution in Section \ref{sec:numerics}, we show that BAMP improves over the AMP in \cite{fan2022approximate}. This improvement is thus a rigorous conclusion, while the conjecture is that, thanks to this improvement, BAMP saturates the Bayes-optimal performance.

{Finally, the ``multi-stage'' state evolution of BAMP suggests a choice of the denoisers in the AMP of \cite{fan2022approximate}, which differs e.g. from the greedy strategy of \cite{zhong2021approximate} picking the full posterior mean denoiser at every iteration. The numerical results of Section \ref{sec:numerics} also show that this denoiser selection -- motivated by BAMP -- meets the BAMP performance and, hence, the replica prediction of the Bayes-optimal error.}

{
\paragraph{Codes} A repository with the codes used in the present work
can be found \href{https://github.com/fcamilli95/Structured-PCA-}{here}.
}

\subsubsection*{Comments on the potential universality of our results} 

We comment the hypotheses under which our results are conjectured valid, and then extrapolate on the more general settings in which the results may still hold. 

We start with a remark concerning the insensitivity of our results to the ``statistical details'' of the noise eigenvalues. Let us precise the hypotheses on the distribution of the noise, in particular on its eigenvalues, under which our results are conjectured valid. As seen from \eqref{densityOD} the eigenvalues of the noise are strongly dependent due to the Vandermonde determinant. However, we conjecture that all our results still hold if one considers instead a simpler ensemble where the $N$ eigenvalues are drawn i.i.d. from $\rho(x)dx$, see \eqref{densityRho}. The reason is that all the analysis and results rely only on the weak convergence of the empirical density of eigenvalues of the ensemble under consideration towards $\rho$. Hence, as long as this is the case, our results must hold, even if we do not rigorously prove it. To formally show it, from now on we consider that the diagonal matrix $\bD$ of eigenvalues of the noise is \emph{deterministic} with the sole constraint that the empirical density of its diagonal entries converges towards $\rho(x)dx$. This of course includes as special cases the two aforementioned settings (i.i.d. and coupled by Vandermonde determinant). We therefore work in this paper under the following hypothesis.

\begin{hypothesis}[Distribution of the noise]\label{hyp-density}
The noise $\mathbb{R}^{N\times N}\ni\bZ=\bO^\intercal \bD\bO$ in model \eqref{channel} is a symmetric rotationally invariant matrix, namely, it is equal in law to $\bU^\intercal \bZ\bU$ for any orthogonal matrix $\bU\in \mathbb{O}(N)$ (the group of $N\times N$ orthogonal matrices). Equivalently, $\bO$ is drawn from the Haar (uniform) measure over $\mathbb{O}(N)$. Moreover, we only require for its (possibly deterministic) eigenvalues $(D_i)_{i\le N}$ that their empirical law $N^{-1}\sum_{i\le N}\delta_{D_i,x}$ is tending weakly as $N\to\infty$ to a probability measure with support bounded uniformly in $N$ and with density $\rho$ with respect to the Lebesgue measure. As mentioned earlier, for the purpose of having a uniform measure of SNR when tuning $(\mu,\gamma(\mu))$ we will consider cases where $\int d\rho(x) x^2=1$  despite this is not necessary for the analysis to hold.
\end{hypothesis}

A second remark concerns the rotational invariance of the noise. We believe that our results may extend beyond this hypothesis to cases where the noise eigenbasis may be invariant under more restrictive transformations (such as permutation invariant), or even ``almost deterministic''. This intuition comes from a very recent line of work concerning linear regression and phase retrieval with structured matrices of covariates. Indeed, the authors of \cite{dudeja2022universality2,dudeja2022universality,dudeja2022spectral} show that in this context, the class of rotationally invariant matrices leads to the same performance as a much broader class of almost deterministic matrices (with the same spectral density), also when AMP or its linearized version are used as inference algorithm. This is a different setting from the one we consider, since in our setup the structured matrix is the noise, but it nevertheless suggests that our predictions should remain true more generically. The confirmation of this universality is left for future work.

\subsubsection*{What is conjectured exact, and what is rigorous} 

We end this section with a remark concerning the level of rigor of our derivations. Most of our results are based on non-rigorous but well established methods from the statistical mechanics of mean-field disordered systems, in particular the replica method at the replica symmetric level, and the theory of Anderson-Thouless-Palmer equations. For a general background on these techniques we refer to \cite{mezard2009information,opper2001advanced,montanari2022short}. It is important to keep in mind that despite being non-rigorous, the results obtained from these techniques are conjectured \emph{exact} in the present setting of \emph{Bayesian-optimal inference} (or equivalently, statistical mechanical models living on their \emph{Nishimori line} \cite{nishimori01}), in the asymptotic large size limit $N\to\infty$. 

This widely admitted asymptotic exactness, first proved for the Sherrington-Kirkpatrick model \cite{guerra2003broken,talagrand2006parisi,panchenko2013sherrington}, spreads in numerous fields and in particular in the analysis of high-dimensional inference. In this context a plethora of rigorous results confirm the validity of replica predictions \cite{barbier2016proof,el2018estimation,BarbierMacris2019,barbier_strong_CMP,barbier2019overlap,krzakala_limits20,Barbier_IMAIAI21}. In particular, replica symmetric formulas for the free entropy, mutual information and minimum mean-square error have been systematically proved thanks to a combination of concentration techniques specifically adapted to the context of inference  \cite{barbier2019overlap,barbier_strong_CMP} together with rigorous versions of the cavity method \cite{Talagrand2011spina,panchenko2013sherrington,coja2018information}, (adaptive) interpolation techniques \cite{guerra2002thermodynamic,BarbierM17a,BarbierMacris2019} or Hamilton-Jacobi approaches \cite{mourrat2020hamilton,chen2021statistical,chen2021limiting}. From this fastly growing literature, we conjecture that it is only a matter of time before our replica-based predictions are proven. 

Concerning our algorithmic results on the novel approximate message passing we propose (BAMP), the results are completely rigorous; full proofs are provided as appendix. They are based on the theory of message passing algorithms and associated state evolution recursions \cite{Bayati_Montanari11}, in particular the most recent results for structured matrices as considered here \cite{fan2022approximate,zhong2021approximate}.

\section{The inhomogeneous spherical integral}\label{appendix_spherical_generalized}
In this section we derive the expression of a useful general integral that will play a crucial role along the whole analysis, and that we believe may have an interest on its own. For the reader interested in the information-theoretic and algorithmic analyses directly, this section can be skipped at first reading as only its main results \eqref{INasympto}, \eqref{varFormIN} and \eqref{SphIntRS_2} will be used in the rest.

Indices $\ell,\ell'\le n$ will always indicate the ``replica dimension'' (with $n$ which always remains finite), while $i,j,k\le N$ index the ``spin dimension'' (where $N$ will diverge). 

\subsection{Definition and variational characterization}\label{sec:sphInt}

Let $\bO\sim {\rm Haar}(\mathbb{O}(N))$ be drawn from the Haar measure over the orthogonal group of $N\times N$ matrices. Consider a fixed matrix $\bx\in \mathbb{R}^{N\times n}$ with rows $\bx_i\in\mathbb{R}^{n}$, $i\le N$, and columns $\bx_\ell\in\mathbb{R}^N$, $\ell\le n$. Assume it has the column-wise overlap structure 
\begin{align}
    \bx_\ell^\intercal \bx_{\ell'}= Nq_{\ell\ell'}, \qquad \ell,\ell'\le n.
\end{align}
We let $\bq=(q_{\ell\ell'})_{\ell,\ell'\le n}:=N^{-1}\bx^\intercal \bx$.  Every vector is considered a column vector, so, e.g., $(\bO\bx)_i$ is a $n$-dimensional column-vector corresponding to the transpose of the $i$th row of the $N\times n$ matrix $\bO\bx$, while $(\bO\bx)_i^\intercal$ is a row-vector.  

Let the matrices $\bC_{\ell\ell'}={\rm diag}((C_{i,\ell\ell'})_{i\le N})$, $\bC_i=(C_{i,\ell\ell'})_{\ell,\ell'\le n}$, and the ``external fields'' $\bh_\ell=(h_{i,\ell})_{i\le N}$, $\bh_i=(h_{i,\ell})_{\ell\le n}$ all having entries bounded uniformly in $N$. The sequence $(\bh_i\in\mathbb{R}^{n},\bC_i\in\mathbb{R}^{n\times n})_{i\le N}$ is assumed to have an empirical law tending to that of the random $(\bh\in\mathbb{R}^{n},\bC\in\mathbb{R}^{n\times n})$: for any continuous bounded function $f:\mathbb{R}^{n\times n}\times\mathbb{R}^{n} \mapsto \mathbb{R}^k$ with $k$ independent of $N$, 
\begin{align*}
    \frac1N \sum_{i\le N} f(\bC_i,\bh_i) \xrightarrow{N\to\infty} \EE f(\bC,\bh).
\end{align*}

We denote by $\mathbb{R}\ni I_N=I_N(\bq,(\bC_{\ell\ell'})_{\ell,\ell'\le n}, (\bh_\ell)_{\ell\le n})=I_N(\bq,(\bC_{i},\bh_i)_{i\le N})$ the generalized low-rank spherical integral, which is defined as
\begin{align}
    I_N&:=\frac1N \ln \mathbb{E}_{\bO} \exp\sum_{i\le N}\big( (\bO\bx)_i^\intercal \bC_{i} (\bO\bx)_i+(\bO\bx)_i^\intercal \bh_{i}  \big)\nonumber\\
    &\,\,=\frac1N \ln \mathbb{E}_{\bO} \exp\Big(\sum_{\ell,\ell'\le n} (\bO\bx_\ell)^\intercal \bC_{\ell\ell'} \bO\bx_{\ell'}+\sum_{\ell\le n}(\bO\bx_\ell)^\intercal \bh_{\ell}  \Big)\nonumber\\
    &\,\,=\frac1N \ln \mathbb{E}_{\bO} \exp\Big(\sum_{i,j,k\le N}\sum_{\ell,\ell'\le n}O_{ij}O_{ik} x_{j,\ell'}x_{k,\ell} C_{i,\ell\ell'}+\sum_{i,j\le N}\sum_{\ell\le n}O_{ij} x_{j,\ell} h_{i,\ell}\Big).\label{gene_sphInt}
\end{align}
Calling the columns $(\bx_\ell)_{\ell \le n}$ ``replicas'', the matrices $(\bC_i)_{i\le N}$, $(\bC_{\ell\ell'})_{\ell,\ell'\le n}$ are coupling them (after the replicas have been jointly rotated by the random $\bO$). Therefore we call them ``replica coupling matrices''.

As $N\to \infty$ with $n$ fixed this integral is given by
\begin{align}
    I_N\xrightarrow{N\to\infty} I_{\bC,\bh}(\bq),\label{INasympto}
\end{align}
with variational formula
\begin{align}
    I_{\bC,\bh}(\bq)&:=\frac12 {\rm extr}_{\tilde \bq} \big( \Tr\bq\tilde\bq+\EE\bh^\intercal (\tilde \bq-2\bC)^{-1}\bh-\EE\ln \det (\tilde \bq-2\bC)\big)\nonumber \\
    &\qquad \qquad-\frac{1}2(n+\ln\det \bq) .\label{varFormIN}
\end{align}
The extremum is over symmetric matrices such $\tilde \bq - 2\bC$ is positive definite for all $\bC$ living on its domain.

We remark that it may be the case that the extremum over $\tilde \bq$ is actually attained on the boundary of the optimization domain, in which case the optimization requires more care than what is done in \eqref{35} to solve it (as \eqref{35} assumes the extremum to lie inside the optimization domain). This is however not expected in the settings of the present paper. When this phenomenon happens, in the standard low-rank spherical integral this leads to a ``sticking phenomenon'' where the solution of the optimization is dependent on the maximum eigenvalue of the full-rank random matrix entering the integral's definition, see \cite{guionnet2005fourier}. 

\subsection{Special cases}

\subsubsection{Low-rank HCIZ integral}\label{sec:standardHCIZ}
The special case $\bh_i=\boldsymbol{0}$ and replica coupling matrices $\bC_i=\bC D_i$ for $i\le N$ corresponds to the standard rank-$n$ spherical (or HCIZ) integral: 
\begin{align*}
    I_N&=\frac1N \ln \mathbb{E}_{\bO} \exp\sum_{i,j,k\le N}\sum_{\ell,\ell'\le n}O_{ij}O_{ik} x_{j,\ell'}x_{k,\ell} C_{\ell\ell'} D_i \\
      &=\frac1N \ln \mathbb{E}_{\bO} \exp\sum_{\ell,\ell' \le n} C_{\ell\ell'} (\bO\bx_\ell)^\intercal\bD \bO\bx_{\ell'}\\
    &=\frac1N \ln \mathbb{E}_{\bO} \exp\Tr\,\bO^\intercal \bD\bO (\bx\bC\bx^\intercal),
\end{align*}
where $\bD={\rm diag}((D_i)_{i\le N})$ and $\bx\bC\bx^\intercal$ is an arbitrary rank-$n$ symmetric matrix (arbitrary given that $\bx$ and $\bC$ are so). Its asymptotic expression can also be obtained from the results of \cite{guionnet2005fourier} after diagonalizing $\bx\bC\bx^\intercal$ and depends only on the limit of the empirical distribution of $(D_i)$ and on the $n$ non-zero eigenvalues of $\bx\bC\bx^\intercal$.

\subsubsection{Low-rank spherical integral with external field and diagonal replica coupling}
Taking diagonal replica coupling matrices $\bC_i= I_n D_i/2$ gives (a generalization of) the spherical integral with external field found in [Prop. 2.7, \cite{fan2021replica}]:
\begin{align}
    I_N=\frac1N \ln \mathbb{E}_{\bO} \exp\sum_{\ell\le n}\Big(\frac12(\bO\bx_\ell)^\intercal  \bD \bO\bx_\ell + (\bO\bx_\ell)^\intercal \bh_{\ell} \Big).
\end{align}

\subsubsection{Low-rank spherical integral with non-diagonal replica coupling and replica symmetric overlap}\label{Low-rank_int_for_replicas}
Let the $N\times N$ diagonal matrices $\bA={\rm diag}((A_i)_{i\le N})$ and similarly for $\bB$. The empirical law of $(A_i,B_i)_{i\le N}$ tends to that of $(A,B)$. Of particular interest to us corresponds to taking $\ell\in\{0,\dots,n\}$, $\bh_i=\boldsymbol{0}$ and replica coupling matrices with only non-zero entries being 
\begin{align}
(\bC_i)_{\ell0}=(\bC_i)_{0\ell}=\frac{A_i}2 \ \ \mbox{for} \ \ 1\le \ell\le n, \qquad (\bC_i)_{\ell\ell}=\frac{B_i}2(1-\delta_{\ell,0}),\label{C_RS}    
\end{align}
or equivalently, 
\begin{align}
    \bC_{0\ell}=\bC_{\ell0}=\frac{\bA}2 \ \  \mbox{and} \ \ \bC_{\ell\ell}=\frac{\bB}2  \ \ \mbox{for} \ \ 1\le \ell\le n,\qquad \bC_{\ell\ell'}=\boldsymbol{0} \ \ \mbox{else}.
\end{align}
Note that this is \emph{not} a special case of the standard rank-$n$ spherical integral of the first example: here $\bC_i$ cannot be written as $\bC$ times a function of $i$; instead different entries of $\bC_i$ vary with $i$ differently. In this case the generalized spherical integral reads (the sum over $\ell$ below starts at $\ell=1$)
\begin{align}
    I_N&=\frac1N \ln \mathbb{E}_{\bO} \exp\sum_{\ell\le n} \Big((\bO \bx_{0})^\intercal \bA \bO \bx_{\ell}+\frac{1}2(\bO \bx_{\ell})^\intercal \bB\bO \bx_{\ell} \Big).\label{gene_sphInt_3}
\end{align}
So the $0$th replica plays here a special role (it corresponds to the planted signal). 

We consider a ``replica symmetric structure'' for the overlap matrix parametrized by the vector $(v_0,  v, m,q )\in\mathbb{R}^{4}$:
\begin{align}
    \bq=\begin{pmatrix}
    v_0   &m     &m&m&\dots&m\\
    m     &v     &q&q&\dots&q\\
    m     &q     &v&q&\dots&q\\
    m     &q     &q&v&\dots&q\\
    \vdots&\vdots&\vdots&\vdots&\ddots&\vdots\\
    m&q&q&q&\dots&v
    \end{pmatrix}\in \mathbb{R}^{(n+1)\times (n+1)},
\end{align}
and, coherently, we assume that the extremum over $\tilde \bq$ is attained for a matrix having the same structure with different constants $(\tilde v_0, \tilde v, \tilde m,\tilde q )$. Its determinant can be easily computed via Gauss' reduction: 
\begin{align}
    \ln\det\bq
    &=\ln v_0+n\ln(v-q)+\ln\Big(
    1+n\frac{v_0q-m^2}{v_0(v-q)}\Big).\label{lndetq}
\end{align}
We also need to compute 
\begin{align}
    \Tr \bq\tilde \bq= v_0\tilde v_0 + n (2 m\tilde m + v\tilde v+(n-1)q\tilde q).
\end{align}
Letting $\bC$ be defined as \eqref{C_RS} but with the random variables $A,B$ replacing $A_i,B_i$, the last missing term is obtained similarly as \eqref{lndetq}: under the replica symmetric stucture for $\tilde \bq$,
\begin{align*}
    &\EE\ln \det (\tilde \bq-2\bC)
    =\ln \tilde v_0+n\ln(\tilde v-B-\tilde q)+\EE\ln\Big(1+n\frac{\tilde v_0\tilde q-(\tilde m-A)^2}{\tilde v_0(\tilde v-B-\tilde q)}
    \Big).
\end{align*} 
Combining everything in the variational formula \eqref{varFormIN}, and taking into account that $\mathbf{q}$ here is a $(n+1)\times(n+1)$ matrix, we obtain the following expression for the generalized spherical integral with replica coupling \eqref{C_RS}, and under a replica symmetric structure for the overlap and conjugate matrices (thus the upperscript): 
\begin{align}
    &I_N\to I_{A,B}^{\rm RS}(\bq):= \frac12 {\rm extr}_{(\tilde v_0,\tilde v,\tilde m,\tilde q)} \Big\{ v_0\tilde v_0-\ln \tilde v_0 + n (2 m\tilde m + v\tilde v+(n-1)q\tilde q)\nonumber \\
    &\qquad\qquad-n\EE\ln(\tilde v-B-\tilde q)-\EE\ln\Big(1+n\frac{\tilde v_0\tilde q-(\tilde m-A)^2}{\tilde v_0(\tilde v-B-\tilde q)}\Big)\Big\} \nonumber\\
    &\qquad\qquad-\frac{1+\ln v_0}2-\frac n2\big(1+\ln(v-q)\big)-\frac12\ln\Big(
    1+n\frac{v_0q-m^2}{v_0(v-q)}\Big).\label{SphIntRS}
\end{align}
By definition \eqref{gene_sphInt_3} of $I_N$ this formula has to cancel when $n=0$. Thus
\begin{align}
    {\rm extr}_{\tilde v_0}\big\{ v_0\tilde v_0 - \ln \tilde v_0\big\}-1-\ln v_0=0.
\end{align}
The saddle point equation over $\tilde v_0$ then yields $\tilde v_0=1/v_0$, in which case this latter formula indeed cancels. So the simplified formula reads
\begin{align}
    &I_{A,B}^{\rm RS}(\bq)= \frac12 {\rm extr}_{(\tilde v,\tilde m,\tilde q)} \Big\{ n (2 m\tilde m + v\tilde v+(n-1)q\tilde q)-n\EE\ln(\tilde v-B-\tilde q)\nonumber \\
    &\qquad\qquad\qquad-\EE\ln\Big(1+n\frac{\tilde q-v_0(\tilde m-A)^2}{\tilde v-B-\tilde q}\Big)\Big\} -\frac n2\big(1+\ln(v-q)\big)\nn
    &\qquad\qquad\qquad-\frac12\ln\Big(
    1+n\frac{v_0q-m^2}{v_0(v-q)}\Big).\label{SphIntRS_2}
\end{align}

\subsection{Derivation of the variational formula}
Let $\tilde{\bx}_\ell:=\bO\bx_\ell$ the columns of $\tilde \bx=\bO\bx$. Under the law of $\bO$ at fixed $\bx$, these random vectors are uniform among all vectors having the overlap structure of $(\bx_\ell)$. Thus their law conditional on $\bx$ is just a function of the symmetric overlap $\bq=(q_{\ell\ell'})$:
\begin{align*}
    &P(\tilde \bx \mid \bx )=P(\tilde \bx \mid \bq )=\frac{ 1 }{\mathcal{Z}(\bq)}\prod_{\ell\ge \ell'}^{1,n} \delta(Nq_{\ell\ell'}- \tilde \bx_\ell^\intercal \tilde \bx_{\ell'})=\frac{ 1}{\mathcal{Z}(\bq)} \delta(N\bq- \tilde \bx^\intercal \tilde \bx) 
\end{align*}
with normalization 
\begin{align*}
    \mathcal{Z}(\bq)=\int  d\tilde \bx \, \delta(N\bq- \tilde \bx^\intercal \tilde \bx).
\end{align*}
Using the Fourier representation of the Delta function, the integral to compute reads (below $\tilde \bq$ is a $n\times n$ symmetric matrix with complex entries)
\begin{align*}
   &\exp(N I_N)=\frac1{ \mathcal{Z}(\bq)}\int  d\tilde \bx\, \delta(N\bq- \tilde \bx^\intercal \tilde \bx)  \exp\sum_{i\le N}\big(  \tilde \bx_i^\intercal \bC_{i} \tilde \bx_i+\tilde \bx_i^\intercal\bh_i  \big) \nonumber\\
   &\quad =\frac1{ \mathcal{Z}(\bq)}\int  d\tilde \bx  d\tilde \bq \exp\Big(\frac N2\Tr\bq\tilde\bq- \frac12\Tr\tilde \bx^\intercal \tilde \bx\tilde\bq+\sum_{i\le N}\big(  \tilde \bx_i^\intercal \bC_{i} \tilde \bx_i+\tilde \bx_i^\intercal\bh_i  \big)\Big) \nonumber\\
   &\quad=\frac1{ \mathcal{Z}(\bq)}  \int d\tilde \bq \exp\Big(\frac N2\Tr\bq\tilde\bq\Big) \prod_{i\le N}\int  d\tilde\bx_i \exp\Big(-\frac12  \tilde \bx_i^\intercal (\tilde \bq-2\bC_{i}) \tilde \bx_i+\tilde \bx_i^\intercal\bh_i  \Big).
   \end{align*}
   We will soon evaluate the $\tilde \bq$-integral by saddle-point approximation. We now assume that the dominating saddle-point belongs to a set $$D_\epsilon:=\{\tilde \bq \in \mathbb{R}^{n\times n} : \tilde \bq-2\bC\succ \epsilon I_n  \ \mbox{for all $\bC$ living on its domain}\},$$ for some arbitrarily small $\epsilon>0$ but independent of $N$. Thus restricting the integral to this domain yields a sub-leading correction $\exp o(N)$. For $\tilde \bq\in D_\epsilon$ a Gaussian integration over $\tilde \bx$ is possible: $\exp(N I_N)$ equals
  \begin{align*}
   &\frac{(2\pi)^{Nn/2}e^{o(N)}}{ \mathcal{Z}(\bq)}  \int_{D_\epsilon} d\tilde \bq \exp \frac N2\frac1{N}\sum_{i\le N}\big( \Tr\bq\tilde\bq+\bh_i^\intercal (\tilde \bq-2\bC_{i})^{-1}\bh_i-\ln \det(\tilde \bq-2\bC_{i})\big)\nonumber\\
   &\quad =\frac{(2\pi)^{Nn/2}e^{o(N)}}{ \mathcal{Z}(\bq)}  \int_{D_\epsilon} d\tilde \bq \exp\Big\{\frac N2 \EE\big( \Tr\bq\tilde\bq+\bh^\intercal (\tilde \bq-2\bC)^{-1}\bh-\ln \det(\tilde \bq-2\bC)\big)\Big\}.
\end{align*}
We used the convergence of the empirical law of the sequence $(\bC_i,\bh_i)_i$ to turn the above empirical mean into a statistical expectation over $(\bC,\bh)$, including the correction in the $\exp o(N)$; this is possible because over $D_\epsilon$ the summand is a bounded continuous function of $(\bC_i,\bh_i)$. As $N$ diverges at fixed $n$ we can estimate the integral by saddle-point and reach that the generalized spherical integral is
\begin{align}
  I_N\to \frac12 {\rm extr}_{\tilde \bq} &\big( \Tr\bq\tilde\bq+\EE\bh^\intercal (\tilde \bq-2\bC)^{-1}\bh-\EE\ln\det (\tilde \bq-2\bC)\big) \nonumber \\
  &-\frac12{\rm extr}_{\tilde \bq} \big( \Tr\bq\tilde\bq-\ln \det \tilde \bq\big)
\end{align}
where the term $-\ln\mathcal{Z}(\bq)/N$ from the normalization has been obtained by simply setting $\bC$ and $\bh$ to all-zeros in the first optimization problem. The extremum is over $n\times n$ symmetric matrices $\tilde \bq$ such $\tilde \bq - 2\bC$ is positive definite for all $\bC$ on its domain.

\emph{Assuming that the extremum is attained inside the optimization domain} we can perform the extremization using $\ln \det \bA=\Tr\ln \bA$. The extremum is solution of the matrix equation
\begin{align}
 \bq=\EE \bh^\intercal(\tilde \bq-2\bC)^{-2}\bh+\EE(\tilde \bq-2\bC)^{-1}.\label{35}
\end{align}
The second extremization leads instead to $\tilde \bq=\bq^{-1}$. Thus the result. 

\section{Information-theoretic analysis by the replica \\method}\label{sec:replica}

Let us start with a remark. Express the noise $\bZ=\bO^\intercal\mathbf{D}\bO$ in terms of its random Haar distributed basis $\bO$ and eigenvalues $\bD$, so that the observation model becomes
\begin{align}
    \mathbf{Y}=\frac{\lambda}{N}\mathbf{P}^*+\bO^\intercal\mathbf{D}\bO.\label{diagModel}
\end{align}
When the signal is rotationally invariant we can consider the noise diagonal right away by absorbing $\bO$ into $\mathbf{x},\mathbf{X}^*$. If the law $P_X$ is uniform on the sphere, then this joint rotation does not change the distribution of $\mathbf{x},\mathbf{X}^*$ which greatly simplifies the analysis. In this simpler case, the replica method is not needed as the computation of the free entropy can be carried out simply using a saddle point method. We provide this analysis in Section \ref{sec:PCAopt}. The rotational invariance of the Gaussian law implies that also that case could be treated similarly by direct computation. On the contrary, for other priors than spherical or Gaussian this is no longer possible and the replica method is needed.

In order to deal with such non-rotational invariant priors we are going to adapt an approach developed by Kabashima in \cite{kabashima2008inference,takahashi2020macroscopic} to study certain inference models where rotational invariant random matrices appear as quenched disorder. The main difference compared to the works is the fact that because they consider (generalized) linear regression, the structured matrix plays the role of covariates/data and therefore does not influence the form of the likelihood when writing the posterior. A novelty of the present setting is the fact that because the structured matrix is now the noise itself, the likelihood is a function of its statistics which in turn complicates the analysis.

The goal here is to compute the log-partition function \eqref{partFunc} using the replica trick
\begin{align}\label{replicatrick}
  \lim_{N\to\infty}\frac1N\EE\ln \mathcal{Z}=   \lim_{N\to\infty}\frac1N\lim_{n\to 0}\partial_n \ln \EE \mathcal{Z}^n=\lim_{n\to 0}\partial_n \lim_{N\to\infty}\frac1N\ln \EE \mathcal{Z}^n. 
\end{align}
The expectation is with respect to $\bY$ or equivalently the independent $\bO,\bx_0$ (recall $\bD$ is deterministic). The last equality \emph{assumes} the commutation of the two limits. Another key assumption of the method is that we are going to make the computation considering $n\in\mathbb{N}$ and then assume an analytic continuation to $n$ in a small neighborhood of $0$. Before doing all that we are going to first re-express our model in a form more convenient for analysis.

\subsection{An equivalent quadratic model}\label{sec:quadraticModel}

The Hamiltonian \eqref{Hamilt} of the model can be written in a more convenient way by introducing the following shorthand notations for order parameters. Despite at the moment only vector $\bx$ has been introduced, soon a family of vectors $(\bx_\ell)$ will be introduced when ``replicating'' the system. So we directly introduce the order parameters for these:
\begin{align}
    &v_\ell=v(\mathbf{x}_\ell):=\frac{1}{N}\Vert\mathbf{x}_\ell\Vert^2,\\
    &M_{(k)\ell}=M_{(k)}(\mathbf{x}_\ell,\bZ):=\frac{1}{N}\bx_\ell^\intercal \bZ^k\bx_\ell,\label{22}\\
    &\kappa_\ell=\kappa({\mathbf{x}}_\ell,\mathbf{x}_0,\bZ):=\frac{1}{N}\bx_\ell^\intercal \bZ\bx_0,\label{kappa}\\
    &q_{\ell\ell'}=q({\mathbf{x}}_\ell,{\mathbf{x}}_{\ell'}):=\frac{1}{N}\bx_\ell^\intercal \bx_{\ell'},\label{defq}
\end{align}
where the replica indices $0\leq\ell,\ell'\leq n$ with the identification $\bx_0:=\bX^*$. 

We now treat the quadratic and quartic part of the matrix potential separately. Let us denote 
$$\Delta:=\frac1N(\bP^*-\bP), \qquad M_\ell:=M_{(1)\ell}.$$ The quadratic part yields a contribution:
\begin{align}
    \frac{N}{4}\Tr[ (\bZ+\lambda \Delta)^2-\bZ^2]
&=\frac{1}{2}\Big[\lambda(\bx_0^\intercal\bZ\bx_0-\bx^\intercal\bZ\bx )+N\lambda^2\Big(\frac1{2N^2}(\|{\mathbf{x}}_0\|^4+\|{\mathbf{x}}\|^4)-q_{01}^2\Big)\Big]\nn
    &=-\frac{N\lambda}{2}M_{1}+\frac{N\lambda^2}{2}\Big(\frac12(v_0^2+v_1^2 )-q_{01}^2\Big)+o(N).
\end{align}
The subscript $1$ indicates that only one replica $\bx_1:=\bx$ is involved yet, and by convention it is replica number one. We used that by the law of large numbers, and thanks to the symmetry of the chosen matrix potential, we can assert that
$$M_{(2k+1)0}=o_N(1), \qquad M_{(2)0}= 1+o_N(1)$$ due to our choice of normalization, so in particular $M_0=M_{(1)0}=o_N(1)$. Again by the law of large numbers we have
$$v_0=\EE (X^{*}_1)^2 + o_N(1)=1+ o_N(1).$$

The quartic contribution is more complicated due to the non-commutativity of matrices:
\begin{align}
    &\frac{N}{8}\Tr[ (\mathbf{Z}+\lambda\Delta)^4-\mathbf{Z}^4]\nn
    &\quad=\frac{N}{8}\Tr[\lambda^4\Delta^4+4\lambda^3\bZ\Delta^3+4\lambda^2\bZ^2\Delta^2+4\lambda\bZ^3\Delta +2\lambda^2\bZ\Delta\bZ\Delta]\nn
    &\quad=\frac{N}{8}\Big[\lambda^4(2q_{01}^4+v_1^4+1-4q_{01}^2(v_1 ^2+1-v_1))\nn
    &\qquad\quad+4\lambda^3(M_{0}(1-q_{01}^2) -M_{1}(v_1 ^2-q_{01}^2)+2q_{01}(v_1 -1)\kappa_1)\nn
    &\qquad\quad+4\lambda^2 \Big(M_{(2)0}+v_1\frac1N \bx^\intercal \bZ^2\bx  -2q_{01}\frac1N\bx^\intercal\bZ^2\bx_0 \Big)+4\lambda \Big(M_{(3)0}-\frac1N\bx^\intercal \bZ^3\bx\Big)\nn
    &\qquad\quad+2\lambda^2 (M_{0}^2 + M_{1}^2-2\kappa_1^2)\Big].
\end{align}
Note that the only three terms which we did not write in a compact form using order parameters are linear and quadratic forms in $\bx$ that do \emph{not} appear elsewhere to a power greater than $1$. This is because introducing order parameters for these would add useless redundancy in the final equations (but it is necessary for the other order parameters due to powers of them appearing in the Hamiltonian). Let
\begin{align}\label{f}
    f_\ell&= f(q_{0\ell},v_\ell,M_\ell,\kappa_\ell):=\gamma\frac{\lambda^4}8\Big(2 q_{0\ell}^4+ v_\ell^4-4 q_{0\ell}^2( v_\ell^2+1- v_\ell)\Big)-\gamma\frac {\lambda^3}2 M_\ell( v_\ell ^2- q_{0\ell}^2)\nn
    &\qquad + \gamma\lambda^3q_{0\ell}(v_\ell - 1)\kappa_\ell  + \gamma\frac{\lambda^2} 4  M_\ell^2- \gamma\frac{\lambda^2}2  \kappa_\ell^2+\mu\frac{\lambda^2 }{2}\Big(\frac12v_\ell^2 -q_{0\ell}^2\Big)-\mu\frac{\lambda}{2}{M}_\ell .
\end{align}
Plugging the contributions we computed into \eqref{Hamilt} shows that the Hamiltonian is equivalently written as
\begin{align}\label{Hamiltonian}
    H_N(\mathbf{x};\mathbf{Z},{\mathbf{x}}_0)&= N f_1+  \gamma\frac\lambda2 \bx^\intercal(\lambda  v_1 \mathbf{Z}^2 -\mathbf{Z}^3)\bx - \gamma q_{01}\lambda^2 \bx^\intercal\mathbf{Z}^2\bx_0+C+o(N),
\end{align}
where we have put all irrelevant constants inside $C$. We will neglect the $o(N)$ contribution in the following as it yields a subleading correction to the free entropy. Also the constant $C$ is irrelevant, so we simply forget about it. Keep in mind that at the moment $f_1$ is still a function of $\bx$. This model is thus \emph{not} (yet) quadratic in $\bx$ due to terms such as $M_1(\bx,\bZ)^2$ appearing in $f_1$.

We now use delta functions to fix various order parameters. We are going to use repeatedly the Fourier representation of the delta function, namely
\begin{align}
  \delta(x)  = \frac1{2\pi} \int d\hat x \exp(i\hat x x). 
\end{align}
Because the integrals we will end-up with will always be at some point evaluated by saddle point, implying a deformation of the integration contour in the complex plane, tracking the imaginary unit $i$ in the delta functions will be irrelevant. Similarly, the normalization $1/(2\pi)$ will always contribute to sub-exponential corrections in the integrals at hand. Therefore, we will allow ourselves to formally write
\begin{align}
  \delta(x)  = \int d\hat x \exp( r\hat x x)
\end{align}
for a convenient constant $r$, keeping in mind these considerations (again, as we evaluate the final integrals by saddle point, the choice of $r$ ends-up being irrelevant).

We denote jointly $\btau:=(v_1,M_1,\kappa_1,q_{01})$ and $\hat \btau$ for their Fourier conjugates. Coming back to the the partition function for this equivalent model \eqref{Hamiltonian}, it can be re-expressed using delta functions as
\begin{align}
    &\int  dP_X(\bx) d \btau \exp\big(-H_N(\mathbf{x};\mathbf{Z},{\mathbf{x}}_0)\big)\nn
    &\qquad\times\delta(Nq_{01} - \bx^\intercal \bx_0)\delta(N v_{1} - \|\bx\|^2) \delta(N M_1-\bx^\intercal\bZ\bx)\delta(N \kappa_{1}-\bx^\intercal\bZ\bx_0)\nn
    &=\, \int  dP_X(\bx) d \btau d\hat\btau \exp\big(-H_N(\btau,\hat \btau,\bx;\bx_0,\bZ)\big),
    \end{align}
    where
    \begin{align}
H_N(\btau,\hat \btau,\bx;\bx_0,\bZ):=Nh(\btau,\hat \btau)  + \bx^\intercal \bJ_1(\btau,\hat \btau,\bZ)\bx+\bx^\intercal\bJ_0(\btau,\hat \btau, \bZ)\bx_0\label{28}
\end{align}
and
\begin{align}
   h(\btau,\hat \btau)&:=f_1-\hat{q}_{01}q_{01}-\frac{\hat v_1v_1}{2}-\frac{\hat M_1M_1}{2}- \hat \kappa_1   \kappa_1,\\
\bJ_1(\btau,\hat \btau,\bZ)&:=\frac{\hat v_1}{2}I_N+\frac{\hat M_1}2 \bZ+\gamma\frac{\lambda^2} 2  v_1 \mathbf{Z}^2 -\gamma\frac\lambda2\mathbf{Z}^3,\\
\bJ_0(\btau,\hat \btau,\bZ) &:=\hat q_{01}I_N+\hat \kappa_1 \bZ-\gamma q_{01}\lambda^2\mathbf{Z}^2.\label{34}
\end{align}
So what this shows is that by introducing new variables (order parameters and conjugate Fourier parameters), the original model turns out being equivalent to an extended system with Hamiltonian \eqref{28}. The key point of all this analysis is that by introducing the new variables $\btau,\hat\btau$ we have turned the interactions between the $(x_i)_{i\le N}$ into purely quadratic ones. This form is now more approriate to be solved using (generalizations of) known techniques. We emphasize that despite the algebraic manipulations leading from \eqref{Hamilt} to \eqref{28} are cumbersome, given a more complicated polynomial potential $V$ the very same strategy could be applied but would require the introduction of more order parameters. Yet, the equivalent model would still collapse into a quadratic one of the above form but with a more complicated function $h$ and matrices $\bJ_1,\bJ_0$ (still being polynomials of the noise $\bZ$ of order one less than the order of $V$). The reason is that the key mechanisms behind these simplifications when expanding the original Hamiltonian \eqref{Hamilt} are stemming from the low-rank structure of the spike.

\subsection{Replica symmetric free entropy using the inhomogeneous spherical integral}\label{replica_comp_section}

Having reduced the model to a quadratic one, we are now ready to replicate the system to compute the free entropy. The partition function $\mathcal Z$ is now computed using the equivalent model \eqref{28}. The expected replicated partition function is
\begin{align}
    \EE \mathcal{Z}(\bx_0,\bZ)^n=\int \prod_{\ell=0}^n dP_X(\bx_\ell) \prod_{\ell\le n} d\btau_\ell d\hat \btau_\ell \, \EE_{\bZ} \exp\Big(- \sum_{\ell \le n}H_N(\btau_\ell,\hat \btau_\ell,\bx_\ell;\bx_0,\bZ)\Big),\label{replicatedZ}
\end{align}
with replicas $(\bx_\ell,\btau_\ell,\hat \btau_\ell)_{\ell \le n}$ and shared quenched disorder $\bx_0,\bZ$. What we do next is to replace $\bZ$ by $\bO^\intercal \bD \bO$ and fix the overlap structure between replicas 
\begin{align}
\bx_\ell^\intercal \bx_{\ell'}=Nq_{\ell\ell'}  , \qquad \ell  ,\ell'\le n
\end{align}
by introducing further variables and their Fourier conjugates (this is already taken care of for the overlaps $\bx_\ell^\intercal \bx_0$ with the planted signal). The purpose will become clear soon. Redefining $\btau_\ell:=(v_\ell,M_\ell,\kappa_\ell)$ and similarly for $\hat \btau_\ell$, and defining the overlaps $\bq=(q_{\ell\ell'})_{0\le \ell<\ell'\le n}$ and similarly for $\hat \bq$, the log-partition function can be recast as
\begin{align}\label{replicated_free_entropy_withdeltas}
     \EE \mathcal{Z}^n&=\int  d\bq d\hat{\mathbf{q}}\prod_{\ell \le n} d\btau_\ell d\hat \btau_\ell  \,\exp N\Big(\sum_{\ell\le n}\Big(\frac{\hat v_\ell v_\ell}{2}+\frac{\hat M_\ell M_\ell}{2}+ \hat \kappa_\ell   \kappa_\ell-f_\ell\Big)+\sum_{0\le \ell<\ell'\le n}\hat q_{\ell\ell'}q_{\ell\ell'}\Big)\nn 
     &\qquad\times \int \prod_{\ell=0}^n  dP_X(\bx_\ell)\exp\Big(-\sum_{0\le \ell<\ell'\le n}\hat q_{\ell\ell'} \bx_\ell^\intercal \bx_{\ell'}-\frac12\sum_{\ell \le n} \hat v_\ell \|\bx_\ell\|^2 \Big)\nn
    &\qquad\times\EE_\bO  \exp\sum_{\ell\le n} \Big((\bO \bx_{0})^\intercal \bA_{\ell} \bO \bx_{\ell}+\frac{1}2(\bO \bx_{\ell})^\intercal \bB_{\ell}\bO \bx_{\ell} \Big)
\end{align}
where the $N\times N$ ``replica coupling matrices'' are
\begin{align}
    &\mathbf{A}_\ell:=-\hat \kappa_\ell\mathbf{D}+\gamma q_{0\ell}\lambda^2  \mathbf{D}^2,\\
    &\mathbf{B}_\ell:=-\hat M_\ell \mathbf{D}-\gamma\lambda^2  v_\ell \mathbf{D}^2+\gamma\lambda\mathbf{D}^3.        
\end{align}

We now assume a replica-symmetric ansatz which should lead to the correct solution due to the strong concentration-of-measure effects taking place in the Bayes-optimal setting as well as the Nishimori identities \cite{nishimori01,barbier_strong_CMP}. It means that we assume that the saddle point over the order parameters dominating the partition function as $N\to\infty$, which are finitely many, lies in the subset verifying the following (note the minus sign introduced for $-\hat q$ and $-\hat m$ for convenience): for all $\ell\neq \ell'=1,\ldots,n$
\begin{align}\label{OPs}
\mbox{Replica Symmetry Ansatz:} \
    \begin{cases}
    M_\ell=M,\qquad \hat M_\ell=\hat M,\\
     \kappa_\ell=\kappa, \qquad \hat\kappa_\ell=\hat \kappa,\\
     v_\ell=v ,\qquad \hat v_\ell=\hat v,\\
     q_{\ell\ell'}= q, \qquad \hat q_{\ell\ell'}=-\hat q, \\
     q_{0\ell}= m, \qquad \hat q_{0\ell}=-\hat m.
    \end{cases}
\end{align}

Using this ansatz, the matrices $(\bA_\ell,\bB_\ell)_{\ell \le n}$ become independent of $\ell$. We thus call their common value $\bA,\bB$. As a consequence the term $\EE_\bO(\,\cdot\,)$ at the third line in \eqref{replicated_free_entropy_withdeltas} is recognized to be what we call an \emph{inhomogeneous spherical integral} defined and analyzed in a devoted Section \ref{Low-rank_int_for_replicas}. From Section \ref{appendix_spherical_generalized} we know that the result of such integral depends only on the overlap structure; this is the reason why we fixed it earlier. We will thus replace it by $\exp N I^{\rm RS}_{A,B}(n,v,m,q)$ whose formula is \eqref{SphIntRS_2} and which is parametrized by the random variables (below $D\sim \rho$)
\begin{align}
    &A=-\hat \kappa D+\gamma m \lambda^2  D^2,\\
    &B=-\hat M D-\gamma\lambda^2  v D^2+\gamma\lambda D^3.    
\end{align}
Notice that at this point the only $\mathbf{x}$-integrals remaining (second line of \eqref{replicated_free_entropy_withdeltas}) are completely factorized over the spin indices $i$. Hence after taking the saddle point the log-replicated free entropy becomes in the limit $N\to\infty$
\begin{align*}
     &\frac1N\ln \EE \mathcal{Z}^n\to  {\rm extr}\Big\{n\Big(\frac{\hat v v}{2}+\frac{\hat M M}{2}+ \hat \kappa   \kappa-\hat m m+\frac{1-n}{2}\hat q q-f(m,v,M,\kappa)\Big) \nn
     & + I^{\rm RS}_{A,B}(n,v,m,q)+\ln\int\prod_{\ell=0}^n dP_X(x_\ell)e^{\hat q
    \sum_{\ell<\ell'\le n} x_\ell x_{\ell'}+\hat m \sum_{\ell \le n}x_0 x_\ell-\frac{\hat{v}}{2}\sum_{\ell\le n} x^2_\ell
    }\Big)\Big\}
\end{align*}
where the extremum is over all scalars in \eqref{OPs}. The last line can be treated by a Hubbard-Stratonovi\v{c} transform (i.e., Gaussian integral formula) to decouple the integral over the replica indices. Doing so it becomes
\begin{align*}    
    \mathbb{E} \Big(\int dP_X(x)\exp\Big(
    \sqrt{\hat{q}}Zx-\frac{\hat{q}+\hat{v}}{2}x^2+\hat mX_0x \Big)\Big)^n,
\end{align*}
with $Z\sim\mathcal{N}(0,1), X_0\sim P_X$. 

We now consider the limit of number of replicas going to $0$ assuming the analytic continuation of our formulas from integer $n$ to real. To expand the latter term we use $\ln \EE X^n=n\EE\ln X + O(n^2)$. The inhomogeneous spherical integral given by \eqref{SphIntRS_2} (with $v_0=1$) also has to be expanded in $n$. We get
\begin{align*}
   &I^{\rm RS}_{A,B}(n,v,m,q)=\frac n2 {\rm extr}_{(\tilde v,\tilde m,\tilde q)} \Big\{  2 m\tilde m + v\tilde v-q\tilde q-\EE\ln(\tilde v-B-\tilde q)\nonumber \\
    &\quad-\EE\frac{\tilde q-(\tilde m-A)^2}{\tilde v-B-\tilde q}\Big\} -\frac n2\big(1+\ln(v-q)\big)-\frac n2\frac{q-m^2}{v-q}+O(n^2) 
\end{align*}
with an expectation over $D\sim \rho$ entering $A,B$. Now we plug the previous expressions in the log-replicated partition function and expand up to $O(n)$ the resulting expression:
\begin{align*}
    &\frac1N\ln \EE \mathcal{Z}^n\to\text{extr}\Big\{n\Big(\frac{\hat v v}{2}+\frac{\hat M M}{2}+ \hat \kappa   \kappa-\hat m m+\frac{1-n}{2}\hat q q-f(m,v,M,\kappa)\Big)\nn
    & \quad + I^{\rm RS}_{A,B}(n,v,m,q) +n\EE\ln \int dP_X(x)\exp\Big(
    \sqrt{\hat{q}}Zx-\frac{\hat{q}+\hat{v}}{2}x^2+\hat mX_0x \Big)\Big\}+O(n^2).
\end{align*}
One can check that as it should $\lim_{N\to\infty}N^{-1}\ln\mathbb{E}Z^n$ vanishes when $n\to 0$. Taking the $n$-derivative (recall \eqref{replicatrick}) and then sending $n\to 0$ the final formula for the free entropy is obtained (and recalling that we dropped irrelevant constants along the computation):
\begin{align}\label{free_energy_RS}
    &\frac1N\EE\ln  \mathcal{Z}\to\text{extr}\Big\{\frac{\hat v v}{2}+\frac{\hat M M}{2}+ \hat \kappa   \kappa-\hat m m+\frac{\hat q q}{2}+  m\tilde m + \frac{v\tilde v}2-\frac{q\tilde q}2\nn
    &\qquad-\gamma\frac{\lambda^4}8\Big(2 m^4+ v^4-4 m^2( v^2+1- v)\Big)+\gamma\frac {\lambda^3}2 M( v^2- m^2)\nn
    &\qquad - \gamma\lambda^3m(v -1)\kappa  - \gamma\frac{\lambda^2} 4  M^2+ \gamma\frac{\lambda^2}2  \kappa^2-\mu\frac{\lambda^2 }{2}\Big(\frac12v^2 -m^2\Big)+\mu\frac{\lambda}{2}{M}\nn
    & \qquad +\EE\ln \int dP_X(x)\exp\Big(
    \sqrt{\hat{q}}Zx-\frac{\hat{q}+\hat{v}}{2}x^2+\hat mX_0x \Big) \nonumber \\
    &\qquad-\frac12\EE\ln(\tilde v-\tilde q+\hat M D+\gamma\lambda^2  v D^2-\gamma\lambda D^3)-\frac 12\ln(v-q)-\frac{q-m^2}{2(v-q)}\nn
    &\qquad+\frac12\EE\frac{(\tilde m+\hat \kappa D-\gamma m \lambda^2  D^2)^2-\tilde q}{\tilde v-\tilde q+\hat M D+\gamma\lambda^2  v D^2-\gamma\lambda D^3} \Big\} + \mbox{constant}.
\end{align}

The extremization is intended over the set of 13 variational parameters $v,\hat{v},\tilde v, m$, $\hat{m},\tilde m, q, \hat{q},\tilde q, M,\hat M,\kappa,\hat \kappa$. However, as we shall see later the saddle point equations will reduce only to two, because thanks to the Nishimori identities the saddle point values of many order parameters can be found right away. This is a specific and rather convenient feature of the Bayesian-optimal setting.

\subsection{Replica saddle point equations}\label{sec:RSFP}
Define the following random local measure
\begin{align}\label{local_measure}
   \langle\,\cdot\,\rangle_{\hat{m},\hat{q},\hat{v}}=\frac{\int dP_X(x)e^{\sqrt{\hat{q}}Zx+\hat{m}xX_0-\frac{\hat q+\hat v}{2}x^2}(\,\cdot\,)}{\int dP_X(x)e^{\sqrt{\hat{q}}Zx+\hat{m}xX_0-\frac{\hat q+\hat v}{2}x^{2}}} ,
\end{align}
the randomness being $Z\sim\mathcal{N}(0,1)$ and $X_0\sim P_X$,
and the random functions (random in $D\sim \rho$)
\begin{align}
    &H=(\tilde v-\tilde q+\hat M D+\gamma\lambda^2  v D^2-\gamma\lambda D^3)^{-1},\\
    &Q=\gamma m\lambda^2D^2-\hat \kappa D -\tilde m.
\end{align}
Below follow the saddle point equations obtained by equating to 0 the gradient w.r.t. the variational parameters of the variational free entropy in \eqref{free_energy_RS}. The parameter associated to each equation are reported in the round parenthesis:
\begin{align*}
    &(m)\quad \mu\lambda^2m+\gamma\lambda^4m(v^2+1-v-m^2)-\gamma\lambda^3Mm-\hat m+\tilde m+\frac{m}{v-q}\\
    &\qquad\qquad-\gamma\lambda^3(v-1)\kappa+\gamma\lambda^2\mathbb{E}QHD^2=0\\
    &(\hat m)\quad m=\mathbb{E}X_0\langle X\rangle_{\hat{m},\hat{q},\hat{v}}\\
    &(\tilde m)\quad m=\mathbb{E}QH\\
    &(q)\quad \hat q-\tilde q=\frac{q-m^2}{(v-q)^2}\\
    &(\hat q)\quad q=\mathbb{E}\langle X\rangle^2_{\hat{m},\hat{q},\hat{v}}\\
    &(\tilde q)\quad q=\mathbb{E}(Q^2-\tilde q)H^2\\
    &(v)\quad -\mu\lambda^2 v-\gamma\lambda^4(v^3-m^2(2v-1))+2\gamma\lambda^3Mv+\hat v+\tilde v-\frac{1}{v-q}-\frac{m^2-q}{(v-q)^2}\\
    &\qquad\qquad-\gamma\lambda^2\mathbb{E}HD^2 -2\gamma m\lambda^3\kappa+\gamma\lambda^2\mathbb{E}D^2(\tilde q-Q^2)H^2=0\\
    &(\hat v)\quad v=\mathbb{E}\langle X^2\rangle_{\hat{m},\hat{q},\hat{v}}\\
    &(\tilde v)\quad v=\mathbb{E}[H+H^2(Q^2-\tilde q)]\\
    &(M)\quad \mu\lambda+\gamma\lambda^3(v^2-m^2)-\gamma\lambda^2M+\hat M=0\\
    &(\hat M)\quad M=\mathbb{E}D[H+H^2(Q^2-\tilde q)]\\
    &(\kappa)\quad \hat\kappa=\gamma\lambda^3 m(v-1)-\gamma\lambda^2\kappa\\
    &(\hat\kappa)\quad \kappa=\mathbb{E}DQH
\end{align*}

As in any replica symmetric mean-field theory, the physical meaning of some order parameters makes it possible to fix their values to their expectation, obtainable using the Nishimori identities and, as a consequence, to drastically reduce this $13$-dimensional system. To begin with, recall that we fixed $v$ to be the squared norm of a sample from the posterior re-scaled by the number of components. Assuming concentration effects take place as they should in this optimal setting, and denoting the posterior mean by $\langle\,\cdot\,\rangle$, using the Nishimori identity we have that
\begin{align}
    v=\lim_{N\to\infty}\frac{1}{N}\mathbb{E}\langle \|\bx\|^2\rangle=\lim_{N\to\infty}\frac{1}{N}\mathbb{E} \|\bX^*\|^2=1.
\end{align}
We have $\hat{v}=0$ because the constraint $v=1$ is enforced by the prior without the need of a delta constraint. The $(\kappa)$-equation can then be used to directly eliminate $\hat \kappa$ by inserting $\hat\kappa=-\gamma\lambda^2\kappa$ into $Q$. The Nishimori identity also imposes
\begin{align}
    m=\mathbb{E}X_0\langle X\rangle_{\hat{m},\hat{q},0}=q=\mathbb{E}\langle X\rangle_{\hat{m},\hat{q},0}^2.
\end{align}
It is not difficult to realize that for this to be true one also needs necessarily $\hat{m}=\hat{q}$. So we have $8$ variables left.
The most tricky parameter is $M$, that we introduced to decouple the four body interactions in the Hamiltonian. Notice first that (recall definitions \eqref{22} and \eqref{defq})
\begin{align*}
        \frac1N\EE\langle \bx^\intercal \bZ\bx\rangle &=\frac1N\mathbb{E}\Big\langle \bx^\intercal\Big(\bY-\frac\lambda N\bP^*\Big)\bx\Big\rangle \nn
        &=\frac1N\mathbb{E} \bX^{*\intercal}\bY\bX^* -\lambda\EE\Big\langle\Big(\frac1N\bx^\intercal \bX^*\Big)^2\Big\rangle\nn
        &=\lambda \Big(1-\EE\Big\langle\Big(\frac1N\bx^\intercal \bX^*\Big)^2\Big\rangle\Big)+O\Big(\frac{1}{N}\Big).
\end{align*}
We used that by the Nishimori identity
\begin{align}
     \frac{\mathbb{E} \langle\bx^{\intercal}\bY\bx\rangle}N=\frac{\mathbb{E} \bX^{*\intercal}\bY\bX^*}N=\frac1N\EE\bX^{*\intercal} \Big(\frac\lambda N\bX^*\bX^{*\intercal}+\bZ\Big)\bX^*=(\EE(X_1^*)^2)^2\lambda=\lambda.\label{xyx}
\end{align}
Indeed, by diagonalizing the noise, $$\EE \bX^{*\intercal}\bZ\bX^*=\EE \sum_{i\le N} s_i^2D_i=\EE\|\bX^*\|^2\EE D_1=0,$$ where $\bs$ is a uniform spherical vector of same norm as $\bX^*$, and $\EE D_1=0$ by symmetry. By concentration happening on the Nishimori line \cite{barbier_strong_CMP} we have $$\EE\Big\langle\Big(\frac1N\bx^\intercal \bX^*\Big)^2\Big\rangle=\Big(\EE\Big\langle\frac1N\bx^\intercal \bX^*\Big\rangle\Big)^2+o_N(1)=m^2+o_N(1).$$
Hence
\begin{align}
\begin{split}
    M=\lim_{N\to\infty}\frac1N\EE\langle \bx^\intercal \bZ\bx\rangle=\lambda(1-m^2).
\end{split}
\end{align}
The $(M)$-equation together with the other identities implies $\hat{M}=-\mu\lambda$. To summarize the Nishimori identities and concentration properties enforce five constraints:
\begin{align}
    \begin{split}
        v=1,\quad \hat{v}=0,\quad m=q,\quad \hat{m}=\hat{q},\quad M=\lambda(1-m^2)
    \end{split}
\end{align}
and we have $6$ variables left.
Our updated definitions of $Q$ and $H$ are
\begin{align}
    &Q=\gamma m\lambda^2D^2+\gamma \lambda^2 \kappa D -\tilde m,\\
    &H=(\tilde v-\tilde q-\mu\lambda D+\gamma\lambda^2   D^2-\gamma\lambda D^3)^{-1}.
\end{align}
Using the Nishimori identities we see from the $(\tilde v)$ and $(\tilde q)$-equations that
\begin{align}
    m=\mathbb{E}H^2(Q^2-\tilde q)\quad\Rightarrow\quad \mathbb{E}H=1-m.
\end{align}
The latter has to be interpreted as an equation for the quantity $\tilde V:=\tilde v-\tilde q$ as a function of $m$. Furthermore, one can now express $\tilde{m}$ as a function of $\kappa $ and $m$. In fact from equation $(\tilde m)$, unfolding $Q$ and then solving for $\tilde{m}$, one gets
\begin{align}
    \tilde m=\frac{\gamma\lambda^2}{1-m}\mathbb{E}D(mD+\kappa)H-\frac{m}{1-m}.
\end{align}
Plugging this back into the $(m)$-equation we get $\hat m$, equation \eqref{hatmReplica}. We stress that inside $H$ there is still an $m$ dependency through $\tilde V$.

With all these simplifications we can close the equations on $(m,\kappa)$ only:
\begin{align}
    \label{m_replica_simplified}
    &(\hat m)\quad m=\mathbb{E}X_0\langle X\rangle_{\hat{m},\hat{m},0}\\
    &(\tilde\kappa)\quad \kappa=\mathbb{E}DQH, \label{systemRS_MMSE}   
\end{align}
where the random variables $Q=Q(m,\kappa,D)$ and $H=H(m,D)$ are
\begin{align}\label{Q_final}
    &Q=\gamma m\lambda^2D^2+\gamma \lambda^2 \kappa D - \frac{\gamma\lambda^2}{1-m}\mathbb{E}D(mD+\kappa)H+\frac{m}{1-m},\\
    \label{H_final}
    &H=(\tilde V-\mu\lambda D+\gamma\lambda^2  D^2-\gamma\lambda D^3)^{-1},
\end{align}
with $\tilde V=\tilde V(m)$ and $\hat{m}=\hat m(m,\kappa)$ being determined respectively by
\begin{align}\label{tildeV_eq_BO}
    &\mathbb{E}H=1-m,\\
    &\hat m=\gamma \lambda^2\EE H D\Big(\frac{m D+\kappa }{1-m}+D Q\Big)+\mu \lambda^2 m.\label{hatmReplica}
\end{align}
Then the replica prediction for the MMSE is 
\begin{align}
    \lim_{N\to\infty}\frac{1}{2N^2}\EE\| \bX^*\bX^{*\intercal}-\EE[\bX^*\bX^{*\intercal}\mid \bY]\|_{\rm F}^2=\frac12(1-m^2).\label{replicaMMSE}
\end{align}
From \eqref{hatmReplica} it is evident that when $\gamma=0$ and $\mu=1$ (to preserve unit variance of the noise), $\kappa$ and $\hat{m}$ decouple, $\hat{m}=\lambda^2 m$, and the equation \eqref{m_replica_simplified} reduces to the standard replica saddle point equation for the Wigner spike model.

There would be also an equation for $\tilde{q}$, that is decoupled though, meaning that $\tilde{q}$ is a simple function of $m$ and $\kappa$ in the end:
\begin{align}
    (q)\quad \tilde{q}=\hat{m}(m,\kappa)-\frac{m}{1-m}.
\end{align}

{
\subsection{The replica formula for the pure sestic potential}\label{sec:sestic_replica}
The same procedure can be followed to obtain a replica symmetric formula in the case of a pure sestic potential $V(x)=\xi x^6/6$. In this subsection we overview the main steps of the related computation.

The Hamiltonian takes the form
\begin{align}
    \begin{split}
        H_N&=\frac{N\xi}{12}\Big[\Tr(\mathbf{Z}+\lambda \Delta)^6-\Tr\mathbf{Z}^6\Big]=\\
        &=\frac{N\xi}{12}\Tr\Big[6\lambda\bZ^5 \Delta+6\lambda^2\bZ^4 \Delta^2+6\lambda^3\bZ^3 \Delta^3+6\lambda^4\bZ^2\Delta^4+6\lambda^5\bZ \Delta^5+\lambda^6\Delta^6\\\
        &\qquad+6\lambda^2\bZ^3\Delta\bZ\Delta+3\lambda^2\bZ^2\Delta\bZ^2\Delta +12\lambda^3\bZ^2\Delta^2\bZ\Delta+2\lambda^3\bZ\Delta\bZ\Delta\bZ\Delta\\
        &\qquad+6\lambda^4\bZ\Delta\bZ\Delta^3+3\lambda^4\bZ\Delta^2\bZ\Delta^2
        \Big]\,.
    \end{split}
\end{align}
We introduce two additional (candidate) order parameters:
\begin{align}
    &M_{k}=\frac{1}{N}\Tr\bZ^k\bx\bx^\intercal\,,\quad \kappa_k=\frac{1}{N}\Tr\bZ^k\bx^*\bx^\intercal\,,\quad k=1,2,3\,.
\end{align}As before $v_0:=\frac{1}{N}\Vert\bx^*\Vert^2$ is replaced by $1$ and $M^*_{1}=M^*_{3}=M^*_{5}=o_N(1)$ thanks to concentration. The other order parameters have the same meaning as in the quartic potential case.
For the sake of brevity, we do not report here the full simplification of the 12 terms appearing in the Hamiltonian, but only the final result: 
\begin{multline}
    H_N(\bx;\bx^*,\bZ)=-\frac{N\xi}{2}\Tr\big[\lambda\bZ^5-\lambda^2v\bZ^4\big]\bP-\xi\lambda^2q_{01}\Tr\bZ^4\bx^*\bx^\intercal\\+N\xi f_1(v,q_{01},(M_k,\kappa_k)_{k=1}^3)
\end{multline}
with
\begin{align}
    \begin{split}
        f_1&=-\frac{\lambda^3}{2}M_3(v_1^2-q_{01}^2)-\lambda^3q_{01}(1-v_1)\kappa_3 +\frac{\lambda^4}{2}M_2(v_1^3-2v_1q_{01}+q^2)\\
        &+\lambda^4\kappa_2q_{01}(v_1+q_{01}^2-v_1^2-1)\\
        &+\frac{\lambda^4}{2}(1+v_1q_{01}^2-2q_{01}^2)-\frac{\lambda^5}{2} M_1\big[(v_1^2-q_{01}^2)v_1^2-(1-q_{01}^2)q_{01}^2+2v_1q_{01}(1-v_1)\big]\\
        &-\lambda^5 \kappa_1\big[(1-q_{01}^2)q_{01}-(v_1^2-q_{01}^2)v_1q_{01}-q_{01}^3(1-v_1)-v_1q_{01}(1-v_1)\big]\\
        &+\frac{\lambda^6}{12}\big[(v_1^3-2v_1q_{01}^2+q^2_{01})^2+(1-2q_{01}^2+v_1q_{01}^2)^2-2q_{01}^2(v_1^2+1-v_1-q_{01}^2)^2\big]\\
        &+\frac{\lambda^2}{2}M_3M_1-\lambda^2\kappa_3\kappa_1+\frac{\lambda^2}{4}M_2^2-\frac{\lambda^2}{2}\kappa_2^2\\
        &-\lambda^3v_1M_2M_1-\lambda^3(1-v_1)\kappa_2\kappa_1+\lambda^3q_{01}(\kappa_2M_1+\kappa_1M_2-\kappa_1)-\frac{\lambda^3}{6}M_1^3+\frac{\lambda^3}{2}\kappa_1^2M_1\\
        &+\frac{\lambda^4}{2}(v_1^2-q_{01}^2)M_1^2-\frac{\lambda^4}{2}(v_1^2+1-2q_{01}^2)\kappa_1^2+\lambda^4q_{01}(1-v_1)\kappa_1M_1\\
        &+\frac{\lambda^4}{2}\kappa_1^2(v_1+q_{01}^2)+\frac{\lambda^4}{4}v_1^2M_1^2-\lambda^4v_1q_{01}\kappa_1M_1.
    \end{split}
\end{align}

Once we fix the new order parameters with some additional conjugates $(\hat{M}_k,\hat{\kappa}_k)_{k=1}^3$, using the inhomogenous spherical integral it is easy to cast a replica symmetric formula for the free entropy of this model:
\begin{align}
    \label{RSfreeent_sestic}
    \begin{split}
        \frac{1}{nN}\log\mathbb{E}\mathcal{Z}^n&\xrightarrow[]{}\text{extr}\Big\{
        \frac{v\hat v}{2}+\frac{q\hat q}{2}-m\hat m+\sum_{k=1}^3\Big(\frac{M_k\hat{ M_k}}{2}+\kappa_k\hat \kappa_k\Big)-\xi f(v,m,(M_k,\kappa_k)_{k=1}^3)\\
        &+\EE\ln \int dP_X(x)\exp\Big(
        \sqrt{\hat{q}}Zx-\frac{\hat{q}+\hat{v}}{2}x^2+\hat mX_0x \Big)+m\tilde m+\frac{v\tilde v-q\tilde q}{2}\\
        &-\frac{1}{2}\EE\log\Big(\tilde v-\tilde q-\xi\lambda D^5+\xi\lambda^2vD^4+\hat{M}_3D^3+\hat{M}_2D^2+\hat{M}_1D\Big)\\
        &-\frac{1}{2}\EE\frac{\tilde q-\big(\xi\lambda^2mD^4-\hat\kappa_3D^3-\hat\kappa_2D^2-\hat\kappa_1D-\tilde m\big)^2}{\tilde v-\tilde q-\xi\lambda D^5+\xi\lambda^2vD^4+\hat{M}_3D^3+\hat{M}_2D^2+\hat{M}_1D}\\
        &-\frac{1}{2}\log(v-q)-\frac{1}{2}\frac{q-m^2}{v-q}
        \Big\}
    \end{split}
\end{align}
where extremization is intended w.r.t. the set of parameters: $v$, $\hat v$, $\tilde v$, $m$, $\hat m$, $\tilde m$, $q$, $\hat q$, $\tilde q$, $\kappa_1,\kappa_2.\kappa_3$, $\hat \kappa_1,\hat\kappa_2,\hat\kappa_3$, $M_1,M_2,M_3,\hat M_1,\hat M_2,\hat M_3$, for a total of 21.

\subsection{Replica saddle point equations for the pure sestic potential}
As we did for $M_1=\lambda(1-m)^2$, $q=m$, $\hat q=\hat m$ and $v=1$, $\hat{v}=0$, we are able to evaluate $M_2,M_3$ in terms of the other parameters too. 

First of all, using Nishimori identities (see \eqref{Nishif_12}-\eqref{Nishif_3} later) it is possible to show that $\EE\Tr\langle\bY\bP\rangle=\lambda$, $\EE\Tr\langle\bY^2\bP\rangle=\lambda^2+1$, $\EE\Tr\langle\bY^3\bP\rangle=\lambda^3+2\lambda$. In the limit we can write formally:
\begin{align}
\begin{split}\label{eq:M_2id}
    M_2&\simeq \EE\langle\Tr\bZ^2\bP\rangle=\EE\langle\bY^2\bP\rangle-2\lambda\EE\langle\bY\bP^*\bP\rangle+\lambda^2\EE\langle\bP^{*2}\bP\rangle\\
    &=1+\lambda^2-2\lambda\EE\langle\Tr(\bZ+\lambda\bP^*)\bP^*\bP\rangle+\lambda^2m^2=1+\lambda^2(1-m^2)-2\lambda m\kappa_1.
\end{split}
\end{align}
Analogously:
\begin{align}\label{eq:M_3id}
    \begin{split}
        M_3&\simeq\EE\langle(\bY-\lambda\bP^*)^3\bP\rangle=\lambda^3(1-m^2)-2\lambda^2 m\kappa_1+\lambda(2-2m\kappa_2-\kappa_1^2)\,.
    \end{split}
\end{align}
Hence the only non trivial parameters we have to look for are $m,\kappa_1,\kappa_2,\kappa_3$.
With the previous identities we can simplify further for the $\hat{M}$'s. Indeed, starting from $\hat{M}_1$ and imposing the Nishimori identities a derivative yields:
\begin{align}
    \begin{split}
        \xi^{-1}\hat{M}_1&=2\frac{\partial f}{\partial M_1}=-\lambda^5(1-m^2)^2+\lambda^2M_3-2\lambda^3M_2+2\lambda^3m\kappa_2-\lambda^3M_1^2+\lambda^3\kappa_1^2\\
        &+2\lambda^4(1-m^2)M_1+\lambda^4M_1-2\lambda^4m\kappa_1\,.
    \end{split}
\end{align}
Using $M_1=\lambda(1-m^2)$ we readily get
\begin{align}
     \begin{split}
        \xi^{-1}\hat{M}_1&=2 \frac{\partial f}{\partial M_1}=\lambda^2M_3-2\lambda^3M_2+2\lambda^3m\kappa_2+\lambda^5(1-m^2)-2\lambda^4m\kappa_1+\lambda^3\kappa_1^2.
    \end{split}
\end{align}
Now, thanks to the identities \eqref{eq:M_2id} and \eqref{eq:M_3id}, used in this order, we obtain the surprisingly simple result:
\begin{align}
    \begin{split}
        \xi^{-1}\hat{M}_1&=2 \frac{\partial f}{\partial M_1}=\lambda^2M_3-2\lambda^3+2\lambda^3m\kappa_2-\lambda^5(1-m^2)+2\lambda^4m\kappa_1+\lambda^3\kappa_1^2=0\,.
    \end{split}
\end{align}
Continuing on $\hat{M}_2$, using identities in the same order:
\begin{align}
    \begin{split}
        \xi^{-1}\hat{M}_2&=2 \frac{\partial f}{\partial M_2}= \lambda^4(1-m^2)+\lambda^2 M_2-2\lambda^3 M_1+2\lambda^3\kappa_1m\\
        &=-\lambda^4(1-m^2)+\lambda^2 M_2+2\lambda^3\kappa_1m=\lambda^2\,.
    \end{split}
\end{align}
Concerning $M_3$ instead:
\begin{align}
    \begin{split}
        \xi^{-1}\hat{M_3}=2 \frac{\partial f}{\partial M_3}=-\lambda^3(1-m^2)+\lambda^2 M_1=0\,.
    \end{split}
\end{align}

Again by \eqref{eq:M_2id} and \eqref{eq:M_3id} we can also fix the values of the $\hat{\kappa}$'s. Starting from $\hat{\kappa}_1$:
\begin{align}
    \begin{split}
        \xi^{-1}\hat{\kappa}_1&=\frac{\partial f}{\partial \kappa_1}=-\lambda^2\kappa_3+\lambda^3mM_2-\lambda^3m+\lambda^3\kappa_1M_1-2\lambda^4(1-m^2)\kappa_1\\
        &+\lambda^4\kappa_1(1+m^2)-\lambda^4mM_1\\
        &=-\lambda^2\kappa_3+\lambda^3mM_2-\lambda^3m + 2\lambda^4 m^2 \kappa_1-\lambda^5m(1-m^2)=-\lambda^2\kappa_3\,.
    \end{split}
\end{align}
Continuing for $\kappa_2$:
\begin{align}
    \begin{split}
        \xi^{-1}\hat{\kappa}_2&=\frac{\partial f}{\partial \kappa_2}=-\lambda^4m(1-m^2)-\lambda^2\kappa_2+\lambda^3mM_1=-\lambda^2\kappa_2\,.
    \end{split}
\end{align}
And finally:
\begin{align}
    \begin{split}
        \xi^{-1}\hat{\kappa}_3&=\frac{\partial f}{\partial \kappa_3}= -\lambda^2\kappa_1 \,.
    \end{split}
\end{align}

Define now the quantities
\begin{align}
    &H=(\tilde v-\tilde q-\xi\lambda D^5+\xi\lambda^2 D^4+\xi\lambda^2 D^2)^{-1}\\
    &Q=\xi\lambda^2mD^4+\xi\lambda^2\kappa_1D^3+\xi\lambda^2\kappa_2D^2+\xi\lambda^3\kappa_3D-\tilde m\,.
\end{align}
The equations for $\tilde v$ and $\tilde q$ appear respectively as
\begin{align}
    &v=1=\EE[H+H^2(Q^2-\tilde q)]\\
    &q=m=\EE H^2(Q^2-\tilde q)
\end{align}
which implies again an equation for $\tilde V:=\tilde v-\tilde q$:
\begin{align}
    \EE H=1-m\,.
\end{align}
Similarly, the equation for $\tilde m$
\begin{align}
    m=\EE QH
\end{align}
can be inverted to find $\tilde m$:
\begin{align}
    \tilde m=\frac{\xi\lambda^2}{1-m}\EE D(mD^3+\kappa_1 D^2+\kappa_2 D+\kappa_3)H-\frac{m}{1-m}\,.
\end{align}

The equations for the $\kappa$'s are obtained with a simple deviative w.r.t. $\hat{\kappa}$'s:
\begin{align}
    &\kappa_k=\EE D^k QH\,,\quad k=1,2,3\,.
\end{align}

Now we just miss the equation for $\hat{m}$ that can be obtained deriving w.r.t. $m$:
\begin{align}
    \begin{split}
        &\tilde m-\hat m+\frac{m}{1-m}+\xi\lambda^2\EE D^4 QH=\xi\Big[
        \lambda^3 mM_3-\lambda^4 M_2m-\lambda^4\kappa_2(1-m^2)+2\lambda^4\kappa_2m^2\\
        &-\lambda^4m+2\lambda^5 M_1m(1-m^2)-\lambda^6m(1-m^2)^2+\lambda^3(M_2\kappa_1+M_1\kappa_2-\kappa_1)-\lambda^4mM_1^2\\
        &+3\lambda^4m\kappa_1^2-\lambda^4\kappa_1M_1
        \Big]=
        \xi\Big[
        \lambda^3 mM_3-\lambda^4 M_2m+2\lambda^4\kappa_2m^2-\lambda^4m+\lambda^3(M_2\kappa_1-\kappa_1)\\
        &+3\lambda^4m\kappa_1^2-\lambda^5\kappa_1(1-m^2)
        \Big]=\\
        &=\xi\Big[
        \lambda^3 mM_3-\lambda^4 m(1+\lambda^2(1-m^2)-2\lambda\kappa_1m)+2\lambda^4\kappa_2m^2-\lambda^4m+\lambda^4m\kappa_1^2
        \Big]\\
        &=\xi\Big[
        \lambda^3 m\big(\lambda^3(1-m^2)-2\lambda^2 m\kappa_1+\lambda(2-2m\kappa_2-\kappa_1^2)\big)\\
        &-\lambda^4 m(1+\lambda^2(1-m^2)-2\lambda\kappa_1m)+2\lambda^4\kappa_2m^2-\lambda^4m+\lambda^4m\kappa_1^2
        \Big]\\
        &=\xi\Big[
        \lambda^4 m (-2m\kappa_2-\kappa_1^2) +2\lambda^4\kappa_2m^2+\lambda^4m\kappa_1^2
        \Big]=0.
    \end{split}
\end{align}
Hence the system of saddle point equations reduces to 
\begin{align}
    &m=\EE X_0\langle X\rangle_{\hat{q}=\hat{m},\hat{m},\hat{v}=0},\\
    &\kappa_k=\EE D^kQH\,,\quad k=1,2,3\\
    &\EE H=1-m\,.
\end{align}
where 
\begin{align}\label{eq:hatm_sestic}
    \hat{m}=\xi\lambda^2\EE H\Big[\frac{mD^4+\kappa_1D^3+\kappa_2D^2+\kappa_3D}{1-m}+D^4Q\Big]\,.
\end{align}
The first four have to be initialized and iterated in parallel. At each iteration instead one has to impose \eqref{eq:hatm_sestic} and to solve $\EE H=1-m$ exactly by dichotomy, obtaining $\tilde{V}=\tilde{V}(m)$.

}

\subsection{Spectral PCA is optimal for rotation-invariant signals}\label{sec:PCAopt}

{Let us start by pointing out that PCA has the same SNR threshold to obtain non-zero overlap for any signal prior. This readily follows from the analysis of \cite{benaych2011eigenvalues}: there, it is proved that both the spectral threshold and the overlap do not depend on the prior of the rank-1 perturbation, as long as its tails are sufficiently well-behaved and the noise is rotationally invariant, as assumed in our work.}
In this section we show that spectral PCA \cite{benaych2011eigenvalues} is optimal for inferring $\bX^*$ such that $\bX^*$ equals in law $\bO\bX^*$ for any orthogonal matrix $\bO$. This is the case for Gaussian and spherically uniformly distributed $\bX^*$.

To do so, we first show that the previous computations can be straightforwardly modified to accommodate the case of spherical prior. Let us assume that the signal $\mathbf{X}^*$ is uniformly distributed on a sphere of radius $\sqrt{N}$. We denote the uniform measure on this sphere by $\omega$. Thanks to the invariance property of the measure on the sphere under rotations we know that $\mathbf{x}$ equals in law $\bO\mathbf{x}$ for  $\mathbf{x}\sim\omega$ and any orthogonal matrix $\bO$. Therefore, we can directly diagonalize the noise without loss of generality and work with the equivalent model
\begin{align}
    \bY=\frac\lambda N\bP^* + \bD.
\end{align}
In this way we can get rid of $\bO$ and as a consequence replicating the system and the inhomogeneous spherical integral becomes useless. Only Gaussian integrations and a saddle point estimation are needed. 

The partition function is \eqref{28}--\eqref{34} but with the diagonal matrix $\bD$ replacing $\bZ$  (the constraint $\|\bx\|^2=N$ is taken care of by the Hamiltonian):
\begin{align}
\int  d\bx d \btau d\hat\btau \exp\big(-Nh(\btau,\hat \btau)  - \bx^\intercal \bJ_1(\btau,\hat \btau,\bD)\bx-\bx^\intercal\bJ_0(\btau,\hat \btau, \bD)\bx_0\big).\label{70}
    \end{align}
Because now $\bJ_1$ and $\bJ_0$ are diagonal matrices, the $\bx$-integral in the partition function is just a Gaussian integral: it is (up to an irrelevant multiplicative constant)
\begin{align}
\int   d \btau d\hat\btau \exp N\Big(-h(\btau,\hat \btau)  -\frac1{2N}\sum_{i\le N}\ln J_{1,i}+\frac1{4N}\sum_{i\le N}x_{0,i}^2\frac{J_{0,i}^2}{J_{1,i}}\Big)
    \end{align}
with $v_1=1$ (appearing in $h$). Because $\bx_0$ is a uniform spherical vector combined with the convergence of the empirical law of $(D_i)$ we have $$-\frac1{2N}\sum_{i\le N}\ln J_{1,i}+\frac1{4N}\sum_{i\le N}x_{0,i}^2\frac{J_{0,i}^2}{J_{1,i}}=-\frac12\EE\ln J_{1,1}+\frac14\EE \frac{J_{0,1}^2}{J_{1,1}}+o_N(1).$$
Thus saddle point estimation of \eqref{70} yields
\begin{align}
    &\frac1N\ln \mathcal{Z}\to \mbox{const}+{\rm extr}\Big\{-f(m,1,M,\kappa) + \hat m m +\frac{\hat v v}{2}+\frac{\hat M M}2+\hat \kappa \kappa\nn
    &\qquad-\frac12\EE\ln \big({\hat v}+\hat M D+\gamma \lambda^2 D^2-\gamma\lambda D^3 \big)+\frac12\EE \frac{(\hat m + \hat \kappa D -\gamma\lambda^2 m D^2 )^2}{{\hat v}+\hat M D+\gamma \lambda^2 D^2-\gamma\lambda D^3} \Big\},
\end{align}
where recall that $f$ is defined by \eqref{f}. Note that this strategy does not require the replica method, and it could also be applied in the case of Gaussian prior $P_X=\mathcal{N}(0,1)$, due to its rotational invariance.

At this point, the saddle point equations can be written and simplified similarly as in the previous section. After doing so and from the numerical solution of the saddle point equations, one can deduce that: $(i)$ in the case of spherical and Gaussian priors the MMSE is the same; and $(ii)$ this MMSE matches the performance of the spectral PCA algorithm studied in \cite{benaych2011eigenvalues}. Additionally, $(iii)$ the MMSE obtained from this exact approach matches the replica prediction of the previous section in the case of Gaussian prior (a special case of factorized $P_X$ tackled by our replica theory). This further confirms the validity and consistency of our methodology. Therefore we conclude that spectral PCA is Bayes-optimal in the special case of rotationally invariant priors and noise.

Let us provide a further argument in support of Bayes-optimality of PCA in the present setting. In this argument we consider the noise eigenvalues as quenched random variables, and we are going to average over them. We first notice that the MMSE estimator is diagonal in the basis of the matrix of data $\bY$. Indeed, letting $\bY$ be diagonalized as $\bY=\bU^\intercal\bS\bU$ then using the posterior \eqref{posterior},
\begin{align}
    \EE[\bX^*\bX^{*\intercal}\mid \bY]& = \frac{C_V}{P_Y(\bY)} \int dP_X(\mathbf{x})\exp\Big(-\frac{N}{2}\Tr V\Big(\bS-\frac{\lambda}{N}(\bU\bx)(\bU\bx)^\intercal\Big)\Big) \bx\bx^{\intercal} \nn
    &= \frac{C_V}{P_Y(\bY)} \bU^\intercal\Big(\int dP_X(\mathbf{x})\exp\Big(-\frac{N}{2}\Tr V\Big(\bS-\frac{\lambda}{N}\bx\bx^\intercal\Big)\Big) \bx\bx^{\intercal}\Big)\bU
\end{align}
where we changed $\bU\bx$ to $\bx$, which leaves the prior invariant by rotational invariance. We would then like to see that the matrix
\begin{equation*}
    \bL = \frac{C_V}{P_Y(\bY)}\int dP_X(\mathbf{x})\exp\Big(-\frac{N}{2}\Tr V\Big(\bS-\frac{\lambda}{N}\bx\bx^\intercal\Big)\Big) \bx\bx^{\intercal}
\end{equation*}
is a diagonal. Indeed, because $\bS=\mbox{diag}(s_1,\dots,s_N)$ is diagonal, $\Tr V(\bS-(\lambda/N)\bx\bx^\intercal)$ can be easily seen (see, e.g., the steps leading to \eqref{polyform}) to be a polynomial of degree $k$ of the $k$ variables 
\begin{equation*}
    \Big(\sum_{i\leq N} x_i^2,\,\sum_{i\leq N} s_i x_i^2,\dots,\,\sum_{i\leq N} s^{k-1}_i x_i^2\Big).
\end{equation*}
Then, for every $1\leq j \leq N$, the integrand that defines $\bL$ takes the same value for $\bx$ and the point $\bx'$ which results from changing the sign of the $j$-th coordinate of $\bx$. We thus have that $\bL$ is a diagonal matrix. 

For $1\leq k \leq N$, let $\bu_k$ be the eigenvector of the $k$-largest eigenvalue of $\bY$. Then we can express $\bL(\bY)$ as $\mbox{diag}(\gamma_1(\bY),\dots,\gamma_N(\bY))$, where by definition we have that 
\begin{align}
    \EE[\bX^*\bX^{*\intercal}\mid \bY]=\sum_{k\le N} \gamma_{k} \bu_k\bu_K^\intercal,
\end{align}
i.e., $\gamma_k = \bu_k^\intercal\EE[\bX^*\bX^{*\intercal}\mid \bY] \bu_k$ with the ordering $\gamma_1\ge \gamma_2\ge\cdots\ge \gamma_N$. This therefore means that the ``matrix magnetization'' may be written according to
\begin{equation*}
    \frac{1}{N^2}\EE \Tr(\EE[\bX^*\bX^{*\intercal}\mid \bY]\bX^* \bX^{*\intercal}) = \frac{1}{N^2}\sum_{k\leq N} \EE[ (\bu_k^\intercal \bX^*)^2\gamma_k].
\end{equation*}

We would like now to compute the asymptotic magnetization of the Bayes estimator. For this we will use Nishimori identities and a bound over the projections of $\bX^*$ onto the eigenvectors of $\bY$ that we verify numerically. More specifically, we will assume that there is some constant $K > 0$ such that for all $k\geq2$ it holds that
\begin{equation}\label{eq:bound_proj}
    (\bu_k^\intercal \bX^*)^2 \leq K.
\end{equation}
As mentioned before, inequality \eqref{eq:bound_proj}, which is an explicit rate of convergence for the limit in \cite[Theorem 2]{benaych2011eigenvalues}, has been verified through many numerical experiments for different noise potentials and SNRs. In every case, a bound of this type is observed, although for experiments close to the corresponding phase transition, the constant $K$ takes larger values and the quantity bounded exhibits a larger variance (this type of behavior is expected to hold very close to the transition point).

Now, notice that by Nishimori identities the following holds
\begin{equation}\label{eq:nishi_proj}
    \EE \gamma_k = \EE (\bu_k^\intercal\bX^*)^2.
\end{equation}
Also, by \cite[Theorem 2]{benaych2011eigenvalues} we have that (below $R$ is the R-transform associated with the noise spectral density $\rho$)
\begin{equation*}
    \begin{split}
        \frac{1}{N^2}\EE[\gamma_1(\bu_1^\intercal\bX^*)^2] & = \frac{1}{N}\Big(1-\frac{R'(1/\lambda)}{\lambda^2}\Big)\EE\gamma_1+\frac{1}{N}\EE\Big[\gamma_1\Big(\frac{(\bu_1^\intercal\bX^*)^2}{N}-1+\frac{R'(1/\lambda)}{\lambda^2}\Big)\Big],
    \end{split}
\end{equation*}
where the second term on the r.h.s. is a vanishing function of $N$. If we use \eqref{eq:nishi_proj} and \cite[Theorem 2]{benaych2011eigenvalues} a second time, we get that
\begin{equation*}
    \begin{split}
        \lim_{N\to\infty}\frac{1}{N^2}\EE[\gamma_1(\bu_1^\intercal\bX^*)^2] & = \Big(1-\frac{R'(1/\lambda)}{\lambda^2}\Big)^2.
    \end{split}
\end{equation*}

On the other hand, by inequality \eqref{eq:bound_proj} and the Nishimori identities \eqref{eq:nishi_proj} we get
\begin{equation*}
    \begin{split}
        \frac{1}{N^2}\sum_{2\leq k \leq N} \EE[\gamma_k(\bu_k^\intercal \bX^*)^2] & \leq \frac{K}{N^2}\sum_{2\leq k \leq N} \EE \gamma_k = \frac{K}{N^2}\sum_{2\leq k \leq N} \EE(\bu_k^\intercal \bX^*)^2.
    \end{split}
\end{equation*}
that by \cite[Theorem 2]{benaych2011eigenvalues}, vanishes in the limit. We then conclude that
\begin{equation*}
    \frac{1}{N^2}\EE \Tr(\EE[\bX^*\bX^{*\intercal}\mid \bY]\bX^* \bX^{*\intercal}) = \Big(1-\frac{R'(1/\lambda)}{\lambda^2}\Big)^2 + o_N(1).
\end{equation*}
This in turn implies that
\begin{equation*}
    \lim_{N\to\infty} \frac{1}{2N^2}\EE\|\EE[\bX^*\bX^{*\intercal}\mid \bY]-\bX^*\bX^{*\intercal}\|^2 = 1 - \Big(1-\frac{R'(1/\lambda)}{\lambda^2}\Big)^2,
\end{equation*}
which is the MSE of the optimally scaled PCA estimator \cite{benaych2011eigenvalues}.

\section{Sub-optimality of the previously proposed AMP}\label{sec:sub-opt}

Consider the following AMP iteration for $t\ge 1$:
\begin{equation}\label{eq:AMPZF}
    \bff^t = \bY \bu^t - \sum_{i=1}^t{\sf b}_{t, i}\bu^i, \quad \bu^{t+1} = h_{t+1}(\bff^t).
\end{equation}
Here, $\bff^t=(f^t_1, \ldots, f^t_N), \bu^{t+1}=(u^{t+1}_1, \ldots, u^{t+1}_N)\in \mathbb R^N$ and the \emph{denoiser} function $h_{t+1}:\mathbb R\to\mathbb R$ is continuously differentiable, Lipschitz and applied  component-wise, namely $h_{t+1}(\bff^t)=(h_{t+1}(f^t_1), \ldots, h_{t+1}(f^t_N))$. The time-dependent AMP estimate of the spike $\bP^*$ is $(\bu^t)^\intercal \bu^t$.

The Onsager coefficients $\{{\sf b}_{t, i}\}_{i\in [t], t\ge 1}$ are carefully chosen so that, conditioned on the signal, the empirical distribution of the components of iterate $\bff^t$ is Gaussian. The form of these Onsager coefficients was derived by \cite{opper2016theory} using non-rigorous dynamic functional theory techniques, and a rigorous state evolution result was recently proved in \cite{fan2022approximate}. More formally, assume that $\bX^*\stackrel{\mathclap{W_2}}{\longrightarrow} X^*$. Then, the state evolution result of  \cite{fan2022approximate} gives that
\begin{equation}\label{eq:SEZF}
    (\bff^1, \ldots, \bff^t)\stackrel{\mathclap{W_2}}{\longrightarrow} (F_1, \ldots, F_t) := \bmu_t X^* + \bW_t,
\end{equation}
where $\bmu_t = (\mu_1, \ldots, \mu_t)$ and $\bW_t=(W_1, \ldots, W_t)$ is a multivariate Gaussian with zero mean and covariance $\bSigma_t=(\sigma_{ij})_{i,j\le t}$ independent of $X^*$. Furthermore, the mean vectors $\{\bmu_t\}_{t\ge 1}$ and the covariance matrices $\{\bSigma_t\}_{t\ge 1}$ are tracked by a deterministic state evolution recursion. We refer to \cite{fan2022approximate} for more details on this AMP and associated state evolution. Such details won't be crucial for our argument, as we are going to focus directly on the fixed point performance, and not on the dynamics.

For this section, we restrict the analysis to $(i)$ Rademacher prior $P_X=\frac12(\delta_1+\delta_{-1})$, and 
$(ii)$ a ``large enough'' signal-to-noise ratio.
We emphasize that our methodology extends to more generic factorized priors. However, since our goal is to prove sub-optimality of AMP, this setting suffices. Moreover, we will further restrict our proof of sub-optimality to $(iii)$ the ``one-step memory'' version of the AMP in \cite{fan2022approximate}. This means that the denoiser $h_{t+1}$ in \eqref{eq:AMPZF} is allowed to depend only on the past iterate $\bff^t$. A more general ``multi-step memory AMP'' was proposed in \cite{zhong2021approximate}, where the denoiser $h_{t+1}$ can depend on \emph{all} the past iterates $\bff^1, \ldots, \bff^t$. {We remark that the analysis of \cite{opper2016theory} suggests that the fixed points of both these versions are the same (see Sec. 4.2 there); the longer memory of the latter AMP being only useful to improve its convergence properties. Note, however, that the setting of the aforementioned reference is different from ours as we have the presence of a spike, not present in \cite{opper2016theory}. We thus extrapolate the conclusions of \cite{opper2016theory} for the setting without a spike, in order to conjecture that a multi-step denoiser would not improve the fixed point performance of the AMP of \cite{fan2022approximate} for spiked matrix inference with structured noise. This is further validated by our numerical experiments of Sec.~\ref{sec:7.2}.} Therefore, despite our analysis below holds under hypotheses $(i)$--$(iii)$, we conclude more generically that \emph{the existing AMP algorithms for structured PCA in \cite{fan2022approximate,zhong2021approximate} are sub-optimal, and this is the case for most SNR values and prior/signal's distributions that are not rotationally invariant}\footnote{We do not discard the possibility that for very peculiar choices of SNR regimes and/or priors these generically sub-optimal AMPs end-up being optimal, but that would be for highly specific setting-dependent reasons. One case where the AMPs of \cite{fan2022approximate}, and also the spectral PCA algorithm \cite{benaych2011eigenvalues}, are actually optimal is when the prior is rotationally invariant (spherical or Gaussian prior), see Section \ref{sec:PCAopt}.}. From the findings in the following sections, the reason for the sub-optimality of these AMPs will become clear. {Essentially, the data $\bY$ is \emph{not} the best choice of matrix to use in the AMP iterates, despite being the most natural one, and existing AMP algorithms only exploit the noise structure as a mean to converge rather than a way to increase the inference performance, see Sec.~\ref{sec:4.3} for more details.}

\subsection{Analysis of the one-step AMP fixed point performance}

In this section we analyse the AMP algorithm \eqref{eq:AMPZF} for structured PCA proposed in \cite{fan2022approximate}, with a posterior mean denoiser with a single-step memory term:
\begin{align} \label{eq:ZFden}  
h_{t+1}(f^t_i) = \mathbb E[X \mid f^t_i].
\end{align}
In [\cite{fan2022approximate}, Section 3] it is shown that the fixed point of this AMP algorithm is, for $\lambda$ sufficiently large, described by the following system:
\begin{align}
    1-\Delta_*={\rm mmse}\Big(\frac{\lambda^2 \Delta^2_*}{\Sigma_*}\Big), \qquad \Sigma_*= \Delta_* R'\Big(\frac{\lambda \Delta_*(1-\Delta_*)}{\Sigma_*}\Big).\label{ZF_FP}
\end{align}
Here, $R'(\cdot)$ denotes the derivative of the $R$ transform of the (limiting) distribution of the noise eigenvalues $D$. For details about the $R$-transform, the interested reader is referred to \cite{novak2014three}. The above is related to the asymptotic overlap of the AMP estimator through 
\begin{align}
    \lim_{t\to \infty}\lim_{N\to\infty} \Big|\frac{1}{N} \bX^{*\intercal}\hat \bx^t\Big| =\lim_{t\to \infty}\lim_{N\to\infty} \frac{1}{N} \|\hat \bx^t\|^2=\Delta_*
\end{align}
and thus the AMP mean-square error is
\begin{align}
    \lim_{t\to \infty}\lim_{N\to\infty} \frac{1}{2N^2} \EE\|\hat \bx^t(\hat \bx^t)^\intercal-\bX^*\bX^{*\intercal}\|^2=\frac12(1-\Delta_*^2).\label{MSEAMP}
\end{align}
In the case of Rademacher prior the explicit form of the posterior-mean denoiser is
\begin{align}
 h_{t+1}(f^t) = \tanh\Big(\frac{f^t \mu_t}{\sigma_{tt}^2}\Big)
\end{align}
where $(\mu_t,\sigma_{tt})$ are the mean and variance of the (empirically) ``Gaussian observation'' $f^t$ computed from the state evolution of \cite{fan2022approximate}. The associated mmse function is (below $Z\sim \mathcal{N}(0,1)$ is a standard Gaussian random variable and $X^*\sim\frac12(\delta_{-1}+\delta_1)$)
{
\begin{align}
    {\rm mmse}(x)&=1-\mathbb{E}\Big[X^*\frac{\int dP_X(y)\, y\,e^{(Z\sqrt{x}+xX^*)y-\frac{x}{2}y^2} }{\int dP_X(y')e^{(Z\sqrt{x}+xX^*)y'-\frac{x}{2}y^{'2}}}\Big]\\
    &=1-\EE \tanh(x +\sqrt{x} Z).\label{mmseRad}
\end{align}}
We now consider the limit $(\lambda,\Delta_*,\Sigma_*)\to(\infty,1,1)$ which indeed is a fixed point of \eqref{ZF_FP} as we verify at the end of this section. Moreover it is unique, see [\cite{fan2022approximate}, Theorem 3.1]. It implies $x:=\lambda^2 \Delta^2_*/\Sigma_*\to\lambda^2\to \infty$. We have in this limit 
\begin{align}
   {\rm mmse}(x)
    &=1-\int dt \frac{e^{-\frac1{2x}(t-x)^2}}{\sqrt{2\pi x}}\tanh(t) \nn
    &=\sqrt{\frac{\pi}{2}}\frac{e^{-\frac x2}}{\sqrt{x}}(1+O(1/x))\nn
    &=\exp\Big(-\frac x2(1+o_x(1))\Big).\label{scalingMMSElargeSNR}
\end{align}
%
%
%
We plug this in the first equation of \eqref{ZF_FP} which gives at leading order
\begin{align}
 \Delta_*
 &=1-\exp\Big(-\frac{\lambda^2}{2}(1+o_\lambda(1))\Big).\label{scalingAMP}
\end{align}

It just remains to check that $(\lambda,\Delta_*,\Sigma_*)=(\infty,1,1)$ is indeed the unique fixed point of \eqref{ZF_FP} in the large SNR regime. From our analysis we already know that this fixed point is consistent with the first equation of \eqref{ZF_FP}. So we simply need to verify the second one, namely,
\begin{align}
R'\big(\lambda (1-\Delta_*)(1+o_\lambda(1))\big)\to 1
\end{align}
as $\lambda\to\infty$. From \eqref{scalingAMP} we have in this limit $\lambda (1-\Delta_*)\to 0$ exponentially fast in $\lambda$, and it can be readily verified that $R'(0)=1$, as the noise distribution $D$ has unit second moment. This ends the argument.

\subsection{Analysis of the replica Bayes-optimal fixed point}

We now analyse in the same large SNR regime the replica fixed point equations that we recall below for convenience: let us rename $\tilde V:=\tilde v-\tilde q$ as they always appear together. We consider that all quantities below are at their saddle point values maximizing the replica free entropy \eqref{free_energy_RS}.

Let us recall the outcome of the Section \ref{sec:RSFP} on the saddle point equations. Consider the random variables (random through their dependence in $D$)
\begin{align}
    &Q=\gamma m\lambda^2D^2+\gamma \lambda^2 \kappa D - \frac{\gamma\lambda^2}{1-m}\mathbb{E}D(mD+\kappa)H+\frac{m}{1-m},\\
    &H=(\tilde V-\mu\lambda D+\gamma\lambda^2  D^2-\gamma\lambda D^3)^{-1}.\label{defH}
\end{align}
For a given value of the parameter $m$, the saddle point equations require $\tilde V=\tilde V(m)$ to be the solution of the implicit equation
\begin{align}
    \EE H &=1-m.\label{EH=1-m}
\end{align}
Using this implicit solution, $H$ is a function $H(m)$ and $Q=Q(m,\kappa)$. Let $Z\sim\mathcal{N}(0,1)$ and $X^*\sim P_X$.  The saddle point equations over the order parameters $(m,\kappa)$ read
\begin{align}
    m&=1-{\rm mmse}(\hat m) ,\label{RS:m}\\
    \kappa&= \EE DQH\label{RS:kappa},
\end{align}
where ${\rm mmse}(\hat m)$ is the same function \eqref{mmseRad} as before and
\begin{align}
    \hat m&=\hat m(m,\kappa)=\gamma \lambda^2\EE H \Big(\frac{m D^2+\kappa D }{1-m}+D^2 Q\Big)+\mu \lambda^2 m.
\end{align}

Recall that the replica prediction for the MMSE is \eqref{replicaMMSE}. In the regime $\lambda\to\infty$ we thus necessarily have $m \to 1^-$. Since the solution $(m,\kappa)$ of the replica saddle point equations yields the MMSE \eqref{replicaMMSE} which must be at least as good as the AMP MSE \eqref{MSEAMP} then $m\ge \Delta_*$. Thus from \eqref{scalingAMP} we deduce  
\begin{align}
 1-m = O\Big(\exp\Big(-\frac{\lambda^2}{2}(1+o_\lambda(1))\Big)\Big).\label{scaling1-m}   
\end{align}
The support of the density of $D$ is bounded, therefore from \eqref{defH} it is then clear that for \eqref{EH=1-m} to be verified under the scaling \eqref{scaling1-m} in the large $\lambda$ limit, the solution $\tilde V$ of \eqref{EH=1-m} must verify 
\begin{align}
\frac{\lambda^2}{\tilde V}= o_\lambda(1).   \label{89} 
\end{align}
Thus from \eqref{EH=1-m} we obtain
\begin{align}
    (1-m)\tilde V=\EE\Big(1+\frac{\gamma \lambda D^2(\lambda-D)-\mu \lambda D}{\tilde V}\Big)^{-1}=1+o_\lambda(1)\label{(1-m)tildeV}
\end{align}
from which we deduce using \eqref{scaling1-m} that
\begin{align}
\tilde V= \Theta\Big(\frac1{1-m}\Big)=\Omega\Big(\exp\Big(\frac{\lambda^2}{2}(1+o_\lambda(1))\Big)\Big).\label{scalingtildeV}      
\end{align}
This also implies that in the limit of large SNR, $H$ becomes deterministic:
\begin{align}
    H=\tilde V^{-1}+O\Big(\frac{\lambda^2}{\tilde V^2}\Big).\label{HistildeV}
\end{align}
This equality means that $H$ can be written as $\tilde V^{-1}$ plus a possibly random term dependent of $D$, which can be bounded by a non-random constant of order $O(\lambda^2/\tilde V^2)$. Similarly for $Q$: using that $\kappa$ is bounded (recall that it is the limit of the expectation of \eqref{kappa}), \eqref{HistildeV} and \eqref{(1-m)tildeV}, we get the following deterministic scaling in the large SNR regime:
\begin{align}
  Q=\frac{m}{1-m}  +O(\lambda^2).
\end{align}
Using all these scalings together with the fact that $\EE D=0$ and $\kappa$ is bounded (actually it can now be seen from the $(\hat \kappa)$-equation of Section \ref{sec:RSFP} that $\kappa=o_\lambda(1)$) we reach, using $\EE D^2=1$ and \eqref{89}, \eqref{(1-m)tildeV},
\begin{align}
    \hat m&=\gamma \lambda^2\EE H \Big(\frac{m D^2+\kappa D }{1-m}+ D^2 Q\Big)+\mu \lambda^2 m\nn
    &=\gamma \lambda^2 \Big(\tilde V^{-1}+O\Big(\frac{\lambda^2}{\tilde V^2}\Big)\Big) \Big(\frac{2m}{1-m}+O(\lambda^2)\Big)+\mu \lambda^2 m\nn
     &=\gamma \lambda^2  \Big(\frac{2m}{\tilde V(1-m)}+O\Big(\frac{\lambda^2}{\tilde V}\Big)+O\Big(\frac{\lambda^2}{\tilde V}\times \frac{1}{\tilde V(1-m)}\Big)+O\Big(\frac{\lambda^4}{\tilde V^2} \Big)\Big)+\mu \lambda^2 m\nn
     &=\gamma \lambda^2 (2m+o_\lambda(1))+\mu \lambda^2 m\nn
     &=\lambda^2(2\gamma  +\mu)(1+o_\lambda(1))
\end{align}
where also used $m=1+o_\lambda(1)$, see \eqref{scaling1-m}. Recall $m=1-{\rm mmse}(\hat m)$ as well as the scaling \eqref{scalingMMSElargeSNR}. So we have
\begin{align}
    m=1-\exp\Big(- \frac{\lambda^2}2(2\gamma +\mu)(1+o_\lambda(1))\Big).
\end{align}
By comparing with \eqref{scalingAMP} we see that $m\neq \Delta_*$. Moreover, since $m$ is the Bayes-optimal overlap, it has to be the case that $m\ge \Delta_*$, namely, $2\gamma +\mu\ge 1$. From \eqref{gamma(mu)} it can be verified that $2\gamma +\mu> 1$ strictly for $\mu<1$. Equality holds for the pure Wigner case $(\mu=1,\gamma=0)$, as expected. This ends the proof that the MMSE \eqref{replicaMMSE} is asympotically in $\lambda$ strictly exponentially smaller than the MSE of AMP with one-term memory \eqref{MSEAMP} whenever $\mu <1, \gamma >0$.

\subsection{What is actually doing this sub-optimal AMP?\\ Mismatched estimation with Gaussian likelihood}\label{sec:4.3}
In the same spirit as \cite{price_of_ignorance22}, we study here a mismatched estimation where the statistician assumes the noise to be Gaussian, thus a wrong likelihood, whereas the noise is drawn from the quartic ensemble with potential \eqref{quartic_potential}.
In the same way as we did for the quartic potential, the mismatched posterior associated to \eqref{channel} is written as
\begin{align}\label{Glike_posterior}
        d\bar{P}_{X|Y}(\mathbf{x}\mid \mathbf{Y})=\frac{1}{\bar{\mathcal{Z}}(\mathbf{Y})} dP_X(\mathbf{x})\exp\Big(
        \frac{\lambda}{2}\Tr\mathbf{Y}\mathbf{x}\bx^\intercal-\frac{\lambda^2}{4N}\Vert\mathbf{x}\Vert^4
        \Big)
\end{align}
where we have re-absorbed $\mathbf{x}$-independent terms in the normalization. The corresponding log-partition function is
\begin{align}\label{mismatched_free_ent}
    \mathbb{E}\ln \bar{\mathcal{Z}}(\mathbf{Y}).
\end{align}
Notice that we have barred some quantities to distinguish them from their Bayes-optimal analogues. We further stress that, with Gaussian likelihood, the spin-glass model that arises already contains only two body interactions.

We aim at \emph{approximating} \eqref{mismatched_free_ent}. Indeed, we are going to perform a replica symmetric computation, which has no a-priori reasons to be exact as we are not anymore in the Bayesian-optimal setting \cite{barbier_strong_CMP} (nor the mismatched posterior is log-concave \cite{Barbier_IMAIAI21}, see also \cite{Camilli_mismatch} as a counter-example). We denote jointly $\boldsymbol{\tau}=(v_1,q_{01})$ and $\hat{\boldsymbol{\tau}}$ their Fourier conjugates. The partition function can then be expressed using deltas to fix the $\boldsymbol{\tau}$ parameters and expanding $\mathbf{Y}$ as in \eqref{channel}. Up to irrelevant constants it reads
\begin{align}
    \bar{\mathcal{Z}}(\mathbf{Y})=\int dP_X(\mathbf{x})d\boldsymbol{\tau}d\hat{\boldsymbol{\tau}}\exp\big(-\bar{H}_N(\boldsymbol{\tau},\hat{\boldsymbol{\tau}},\mathbf{x};\mathbf{x}_0,\mathbf{Z})\big)
\end{align}
where
\begin{align}
    \begin{split}
        \bar{H}_N(\boldsymbol{\tau},\hat{\boldsymbol{\tau}},\mathbf{x};\mathbf{x}_0,\mathbf{Z}):=N\bar{h}(\boldsymbol{\tau},\hat{\boldsymbol{\tau}})+\mathbf{x}^\intercal \bar{\mathbf{J}}(\boldsymbol{\tau},\hat{\boldsymbol{\tau}},\mathbf{Z})\mathbf{x}+\hat{q}_{01}\mathbf{x}^\intercal\mathbf{x}_0
    \end{split}
\end{align}
and
\begin{align}
    &\bar{h}(\boldsymbol{\tau},\hat{\boldsymbol{\tau}}):=\frac{\lambda^2}{4}v_1^2-\frac{\lambda^2}{2}q_{01}^2-q_{01}\hat{q}_{01}-\frac{v_1\hat{v}_1}{2},\\
    &\bar{\mathbf{J}}(\boldsymbol{\tau},\hat{\boldsymbol{\tau}},\mathbf{Z}):=\frac{\hat{v}_1}{2}I_N-\frac{\lambda}{2}\mathbf{Z}.
\end{align}
While replicating we will need as before to fix the entire overlap structure (and not only $q_{01}$), i.e., $(N\mathbf{q})_{\ell\ell'}=Nq_{\ell\ell'}=\mathbf{x}_\ell^\intercal\mathbf{x}_{\ell'}$, the diagonal elements being denoted as $v_\ell$. As usual, we also introduce the corresponding Fourier conjugates $\hat{\mathbf{q}}$. The expected  replicated partition function then reads as
\begin{align}
    \mathbb{E}\bar{\mathcal{Z}}^n&=\int d\mathbf{q}d\hat{\mathbf{q}}\exp N\Big(
    \sum_{\ell\leq n}\Big(\frac{\lambda^2}{2}q_{0\ell}^2-\frac{\lambda^2}{4}v_\ell^2+\frac{v_\ell\hat{v}_\ell}{2} \Big)+\sum_{0\leq\ell<\ell'\leq n}\hat{q}_{\ell\ell'}q_{\ell\ell'}
    \Big)\nn
    &\qquad\times \int\prod_{\ell=0}^ndP_X(\mathbf{x}_\ell)\exp\Big(
    -\sum_{0\leq\ell<\ell'\leq n}\hat{q}_{\ell\ell'}\mathbf{x}^\intercal_\ell\mathbf{x}_{\ell'}-\frac{1}{2}\sum_{\ell\leq n}\hat{v}_\ell\Vert\mathbf{x}_\ell\Vert^2
    \Big)\nn
    &\qquad\qquad\times\mathbb{E}_\mathbf{O}\exp\Big(
    \frac{\lambda}{2}\Tr\mathbf{O}\mathbf{D}\mathbf{O}^\intercal\sum_{\ell\leq n}\mathbf{x}_\ell\mathbf{x}_\ell^\intercal
    \Big).
\end{align}
In the last line we recognize a rank-$n$ (standard) spherical integral, see Section \ref{sec:standardHCIZ} and \cite{guionnet2005fourier}. Recall the spectrum is deterministic with empirical law tending weakly to $\rho$. Hence we can use the results from Section \ref{Low-rank_int_for_replicas}, with the difference that $\mathbf{C}=I_n\frac{\lambda}{2}D$ is virtually a scalar random variable, and thus w.l.o.g. we can also assume $\mathbf{q}$, and thus $\tilde{\mathbf{q}}$ to be diagonal in \eqref{varFormIN}. If we aim for a replica symmetric ansatz
\begin{align}
    \text{Replica Symmetric Ansatz: }\begin{cases}
    v_\ell=v ,\quad \hat{v}_\ell=\hat{v}\\
    q_{0\ell}=m,\quad \hat{q}_{0\ell}=-\hat{m}\\
    q_{\ell\ell'}=q,\quad\hat{q}_{\ell\ell'}=-\hat{q}\;(\ell\neq\ell')
    \end{cases}
\end{align}
then $\mathbf{q}$ has a non degenerate eigenvalue $v+(n-1)q$ and $n-1$ degenerate eigenvalues $v-q$. Within this ansatz we can thus replace the mentioned spherical integral with
\begin{multline}
    \mathbb{E}_\mathbf{O}\exp\Big(
    \frac{\lambda}{2}\Tr\mathbf{O}\mathbf{D}\mathbf{O}^\intercal\sum_{\ell\leq n}\mathbf{x}_\ell\mathbf{x}_\ell^\intercal
    \Big)=\exp N\Big(
    (n-1)I_D(v-q)+I_D(v-q+nq)
    \Big)\\
    =\exp Nn\Big(
    I_D(v-q)+I'_D(v-q)q +O(n)
    \Big)
\end{multline}
as done in \cite{adatap}, where $I_D(\cdot)$ are rank-one spherical integrals. The rest can be treated exactly as in Section \ref{replica_comp_section}, yielding
\begin{align}
    &\lim_{N\to\infty}\frac{1}{N}\mathbb{E}\ln\bar{\mathcal{Z}}(\mathbf{Y})=\text{extr}\Big\{
    \frac{\lambda^2}{2}m^2-\frac{\lambda^2}{4}v^2+\frac{\hat{v}v}{2}-\hat{m}m+\frac{\hat{q}q}{2}+I_D(v-q)\nn
    &\qquad+qI'_D(v-q)+\EE\ln \int dP_X(x)\exp\Big(
    \sqrt{\hat{q}}Zx-\frac{\hat{q}+\hat{v}}{2}x^2+\hat mX_0x \Big)
    \Big\}
\end{align}
where extremization is intended over $m,\hat{m},q,\hat{q},v,\hat{v}$. With the same notation for the local measure \eqref{local_measure}, the fixed point equations read
\begin{align}\label{eqs_Glike_1}
    &(m)\qquad \hat{m}=\lambda^2m\\
    \label{eqs_Glike_2}
    &(\hat{m})\qquad m=\mathbb{E}X_0\langle X\rangle_{\hat{m},\hat{q},\hat{v}}\\
    \label{eqs_Glike_3}
    &(q)\qquad \hat{q}=2qI_D''(v-q)\\
    \label{eqs_Glike_4}
    &(\hat{q})\qquad q=\mathbb{E}\langle X\rangle^2_{\hat{m},\hat{q},\hat{v}}\\
    \label{eqs_Glike_5}
    &(v)\qquad \hat{v}=\lambda^2 v-2I_D'(v-q)-2qI_D''(v-q)\\
    \label{eqs_Glike_6}
    &(\hat{v})\qquad v=\mathbb{E}\langle X^2\rangle_{\hat{m},\hat{q},\hat{v}}.
\end{align}
The computation above follows the same lines as that in \cite{adatap}, with the only difference being the presence of a planted signal. In case of Gaussian likelihood, the term arising from the spike though is easily tractable, as well as the term containing the fourth norm of the estimator (see \eqref{Glike_posterior}). This suggests that the AMP algorithm designed in \cite{fan2022approximate}, whose aim was to make the results in \cite{adatap,opper2016theory} rigorous, has to match the performance predicted by our replica computation, measured by the MSE
\begin{align}
    \lim_{N\to\infty}\frac{1}{2N^2}\EE\| \bX^*\bX^{*\intercal}-\bar{\EE}[\bX^*\bX^{*\intercal}\mid \bY]\|_{\rm F}^2=\frac12(1-2m^2+q^2).\label{Glike_replicaMSE}
\end{align}in the large $N$ limit, where the $\bar{\EE}$ denotes the expectation w.r.t. \eqref{Glike_posterior}, and $m$ and $q$ solve \eqref{eqs_Glike_1}--\eqref{eqs_Glike_6}.

An alternative to \eqref{eqs_Glike_1}--\eqref{eqs_Glike_6}, which turns out to be more practical from the numerical point of view, can be obtained by keeping $\mathbf{q}$ as it is, without diagonalizing it. In this case one needs the entire formula \eqref{varFormIN}, with $\tilde{\mathbf{q}}$ having the same RS structure as $\mathbf{q}$, in a similar fashion as that of Section \ref{Low-rank_int_for_replicas}. The spherical integral then takes the form (up to constants)
\begin{align}
        &\mathbb{E}_\mathbf{O}\exp\Big(
    \frac{\lambda}{2}\Tr\mathbf{O}\mathbf{D}\mathbf{O}^\intercal\sum_{\ell\leq n}\mathbf{x}_\ell\mathbf{x}_\ell^\intercal
    \Big)\propto \exp \Big(Nn\,\text{extr}\Big\{\frac{v\tilde{v}-q\tilde{q}}{2}-\frac{1}{2}\mathbb{E}\ln(\tilde{v}-\tilde{q}-\lambda D)\nn
    &\qquad-\frac{\tilde{q}}{2}\mathbb{E}(\tilde{v}-\tilde{q}-\lambda D)^{-1}-\frac{1}{2}\ln (v-q)-\frac{q}{2(v-q)}+O(n)\Big\}\Big)
\end{align}
where extremization is w.r.t. the tilded variables only, for now. Consequently, the free entropy rewrites as follows
\begin{align}
    &\lim_{N\to\infty}\frac{1}{N}\mathbb{E}\ln\bar{\mathcal{Z}}(\mathbf{Y})=\text{extr}\Big\{
    \frac{\lambda^2}{2}m^2-\frac{\lambda^2}{4}v^2+\frac{(\hat{v}+\tilde{v})v}{2}-\hat{m}m+\frac{(\hat{q}-\tilde{q})q}{2}\nn
    &\qquad-\frac{1}{2}\mathbb{E}\ln(\tilde{v}-\tilde{q}-\lambda D)-\frac{\tilde{q}}{2}\mathbb{E}(\tilde{v}-\tilde{q}-\lambda D)^{-1}-\frac{1}{2}\ln (v-q)-\frac{q}{2(v-q)}\nn
    &\qquad+\EE\ln \int dP_X(x)\exp\Big(
    \sqrt{\hat{q}}Zx-\frac{\hat{q}+\hat{v}}{2}x^2+\hat mX_0x \Big)
    \Big\}.
\end{align}
Here instead, extremization is intended over the tilded and hatted variables, together with $m,q,v$.

The fixed point equations are
\begin{align}
    \label{m_equation}
    &(m)\qquad \hat{m}=\lambda^2 m\\
    &(\hat{m})\qquad m=\mathbb{E}X_0\langle X\rangle_{\hat{m},\hat{q},\hat{v}}\\
    &(q)\qquad \hat{q}-\tilde{q}=\frac{q}{(v-q)^2}\\
    &(\hat{q})\qquad q=\mathbb{E}\langle X\rangle^2_{\hat{m},\hat{q},\hat{v}}\\
    &(\tilde{q})\qquad q=-\tilde{q}\mathbb{E}(\tilde{v}-\tilde{q}-\lambda D)^{-2}\\
    &(v)\qquad \hat{v}+\tilde{v}-\lambda^2v-\frac{1}{v-q}+\frac{q}{(v-q)^2}=0\\
    &(\hat{v})\qquad v=\mathbb{E}\langle X^2\rangle_{\hat{m},\hat{q},\hat{v}}\\
    \label{tildev_equation}
    &(\tilde{v})\qquad v-\mathbb{E}(\tilde{v}-\tilde{q}-\lambda D)^{-1}+\tilde{q}\mathbb{E}(\tilde{v}-\tilde{q}-\lambda D)^{-2}=0.
\end{align}
Plugging $(\tilde{q})$ into $(\tilde{v})$ we readily see that
\begin{align}\label{tildev_tildeq_equation}
    v-q=\mathbb{E}(\tilde{V}-\lambda D)^{-1}
\end{align}
that works as an equation for $\tilde V:=\tilde{v}-\tilde{q}$ as a function of $v,q$. Analogously, we can plug $(q)$ into $(v)$ obtaining
\begin{align}\label{hatv_hatq_equation}
    \hat{v}+\hat{q}=\lambda^2v+\frac{1}{v-q}-\tilde{V}
\end{align}
that determines $\hat{v}+\hat{q}$ as a function of $v$ and $q$, thanks to the above equation for $\tilde{V}$.
Finally, from $(\tilde{q})$ and $(q)$ we have respectively
\begin{align}
    \label{tildeq_equation_final}
    &\tilde{q}=-\frac{q}{\mathbb{E}(\tilde{V}-\lambda D)^{-2}}\\
    \label{hatq_equation_final}
    &\hat{q}=\frac{q}{(v-q)^2}+\tilde{q}.
\end{align}
Notice that, being in a mismatched setting, there cannot be any simplifications due to the Nishimori identities. 

It is not difficult to verify a posteriori that the systems \eqref{eqs_Glike_1}--\eqref{eqs_Glike_6} and \eqref{m_equation}--\eqref{tildev_equation} are equivalent. The extremization over the tilded variables has indeed the purpose of reproducing $I_D$ and its derivatives. From \eqref{tildev_tildeq_equation} one can infer
\begin{align}
    \tilde{V}=R_{\lambda \mathbf{D}}(v-q)+\frac{1}{v-q}
\end{align}
where $R_{\lambda \mathbf{D}}$ denotes the R-transform of $\lambda \mathbf{D}$, and deriving both sides w.r.t. $v$ one also has
\begin{align}
    \tilde{V}'=-\frac{1}{\mathbb{E}(\tilde{V}-\lambda D)^{-2}}=R'_{\lambda \mathbf{D}}(v-q)-\frac{1}{(v-q)^2}.
\end{align}
Therefore, from \eqref{tildeq_equation_final}
\begin{align}
    \tilde{q}=qR'_{\lambda \mathbf{D}}(v-q)-\frac{q}{(v-q)^2}\quad\Rightarrow\quad \hat{q}=qR'_{\lambda \mathbf{D}}(v-q),
\end{align}
and from \eqref{hatv_hatq_equation}
\begin{align}
    \hat{v}+\hat{q}=\lambda^2 v-R_{\lambda\mathbf{D}}(v-q),
\end{align}
both in perfect agreement with \eqref{eqs_Glike_3} and \eqref{eqs_Glike_5}, as long as $R_{\lambda\mathbf{D}}=2I_D'$ \cite{guionnet2005fourier}.

\begin{figure}[t!]
\begin{center}
                \includegraphics[width=0.7\linewidth]{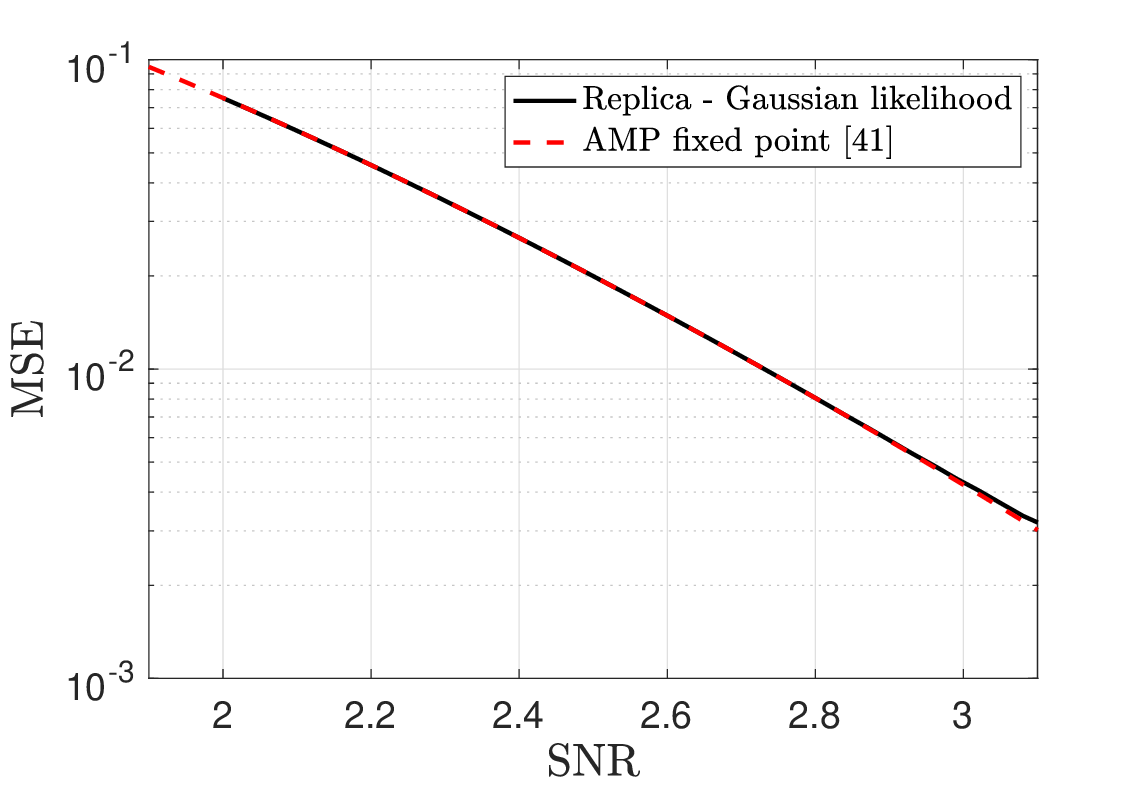}
                 \caption{Comparison between the fixed point of the AMP algorithm in \cite{fan2022approximate} and that obtained via the replica computation (cf. \eqref{m_equation}--\eqref{tildev_equation}), for i.i.d. Rademacher distributed $(X_i^*)_i$. The agreement between these two fixed points  is excellent when the SNR is between 2 and 3.}
                 \label{fig:Glike}
\end{center}
\end{figure}

The system of fixed point equations \eqref{m_equation}--\eqref{tildev_equation} can be solved numerically as follows: \emph{(i)} initialize $m=m_0$, $q=q_0$, $v=v_0$ (the latter being identically $1$ if we use a Rademacher prior); \emph{(ii)}
solve \eqref{tildev_tildeq_equation} for $\tilde{V}$; \emph{(iii)} compute $\hat{q}$, $\tilde{q}$, $\hat{m}$ and $\hat{v}+\hat{q}$ from \eqref{hatq_equation_final}, \eqref{tildeq_equation_final}, \eqref{m_equation} and \eqref{hatv_hatq_equation} respectively; \emph{(iv)} update the values of $m,q,v$ through $(\hat{m})$, $(\hat{q})$ and $(\hat{v})$ obtaining $m^1,q^1$ and $v^1$; \emph{(v)} repeat the steps \emph{(i)}--\emph{(iv)} starting from $m=m^1$, $q=q^1$ and $v=v^1$, thus obtaining $m^2$, $q^2$ and $v^2$, and so forth. 

The numerics arising from this procedure though turns out to be delicate for extreme values of the overlap, namely when $v-q$ is really small, which in turn happens when $\lambda$ is large (typically $>3$ for Rademacher prior). The equation that seems to generate numerical instability is \eqref{hatq_equation_final}, and in particular the two contributions there appearing. With reference to the Rademacher prior, and the related \figurename\,\ref{fig:Glike}, when $\lambda >3$ the overlap gets close to $\sim 0.999$. At this value $1/(v-q)^2\sim 10^6$. $\tilde{q}$, that is also contributing to \eqref{hatq_equation_final}, on the contrary becomes really negative, and is such that $\hat{q}$ is typically $\sim 10$ near $\lambda\sim 3$. The subtraction of these two big numbers apparently dooms the iterations for larger SNRs. This was not the case in the Bayes-optimal setting, thanks to the simplifications introduced by the Nishimori identities. Indeed, from \eqref{Q_final}, \eqref{H_final}, \eqref{tildeV_eq_BO} and \eqref{hatmReplica} we see that $1-m$ appears at most at the first power in denominators. The only issue there was that $\tilde{V}$ can grow exponentially fast, and this can be solved by allowing for a wide range of search of the solution of \eqref{tildeV_eq_BO}.

The fixed point of the MSE arising from \eqref{m_equation}--\eqref{tildev_equation} is compared with the fixed point \eqref{ZF_FP}, which corresponds to the MSE of the AMP proposed in \cite{fan2022approximate}. The match between these two computations is excellent, as long as the SNR is not too large, because of the aforementioned numerical issues in iterating  \eqref{m_equation}--\eqref{tildev_equation}. The plot of Figure \ref{fig:Glike} is a compelling numerical confirmation of the arguments put forward in this section. The conclusion is the following: the AMP algorithm of \cite{fan2022approximate} is solving a \emph{replica symmetric approximation} to the TAP equations associated with the mismatched posterior distribution \eqref{Glike_posterior}. {This analysis, in turns, shows that despite the existing AMP algorithm \cite{fan2022approximate} is aware of the noise structure/statistics, it turns out that it nevertheless makes an implicit assumption of i.i.d. Gaussian noise, and the noise structure is ``only'' exploited to enforce convergence despite this mismatched assumption, rather than as a source of improvement in statistical accuracy. In contrast, our AMP algorithm proposed in the next sections exploits noise structure for both convergence and statistical accuracy.}

\section{Towards an optimal AMP: AdaTAP formalism}\label{sec:AdaTAPtoward}

We have previously shown that the AMP found in the literature for structured PCA \cite{fan2022approximate} is sub-optimal. In this section we understand the fundamental reason behind this issue by generalizing the Adaptive Thouless-Anderson-Palmer (AdaTAP) formalism of \cite{adatap,opper2016theory}. Using our new insights we will then be able in the next section to cure the issue and derive a Bayes-optimal AMP. Like in the replica method and in particular Section \ref{sec:quadraticModel}, a key ingredient will be to reduce the model to a quadratic one of the Ising type.

{Let us mention that, to the best of our knowledge, it is the first time that AdaTAP is used for a planted model with a spike; usually the interaction matrix of the Ising-type model for which AdaTAP was designed is rotationally invariant. Our motivation is that in our setting, the spike being low-rank, this should not affect much the ``macroscopic'' properties of the data matrix compared to the null model, i.e., with pure noise, for which AdaTAP was developed. By macroscopic quantities we mean here density of eigenvalues, its cumulants, or the fact that the eigenbasis of the data remains ``almost'' uniformly (Haar) distributed in the set of orthogonal matrices in spite of the presence of the rank-one spike. Of course these are not proper arguments for its validity in the current setting. The real confirmation of the validity of AdaTAP for our setting will be the a-posteriori perfect match between the BAMP derived from AdaTAP (in particular its state evolution fixed point) and the replica prediction for the MMSE, obtained in a completely different manner.}

\subsection{The AdaTAP single-instance free entropy}\label{sec:adatapFreeEn} 

Recall that the posterior distribution is given by \eqref{posterior}. Denoting $\bp:=\bx\bx^\intercal/N$ and $v:=\|\bx\|^2/N$ the trace of the matrix potential \eqref{quartic_potential} can be expanded as follows:
\begin{align*}
 &\Tr V(\bY-\lambda \bp)= C+\frac \mu 2\Tr\big\{\lambda^2v^2 - 2 \lambda{\bY\bp}\big\} \nonumber\\ 
 &\qquad+\frac\gamma 4\Tr\big\{\lambda^4v^4-4\lambda^3v^2 \bY\bp +4\lambda^2 v\bY^2\bp  - 4{\lambda}\bY^3\bp  + 2\lambda^2 \bY\bp\bY\bp\big\}
\end{align*}
where $C$ is independent of $\bx$. Define the matrix polynomial:
\begin{align}
    \bR(v,\bY):=-(\mu\lambda +\gamma\lambda^3v^2)\bY+\gamma\lambda^2v\bY^2-\gamma\lambda\bY^3.
\end{align}
Then 
\begin{align}
&-\frac N2\Tr V(\bY-\lambda \bp)\propto\nn
&\qquad-\frac N4\lambda^2v^2\Big(\mu+\frac{\gamma\lambda^2v^2}2\Big)-\frac{\bx^\intercal \bR(v,\bY) \bx}2-\frac{N}4\gamma\lambda^2\Big(\frac{\bx^\intercal \bY\bx}N\Big)^2.      \label{polyform}
\end{align}  
The partition function of the model defined by \eqref{partFunc} can then be written in the form
\begin{align}\label{Zadatap}
\mathcal{Z}&\propto \int dP_X(\bx)dvdf \delta(Nv-\|\bx\|^2)\delta(Nf - \bx^\intercal \bY\bx) \nonumber\\
&\qquad\times\exp\Big(-\frac N4\lambda^2v^2\Big(\mu+\frac{\gamma\lambda^2v^2}2\Big)-\frac12 \bx^\intercal \bR(v,\bY) \bx-\frac N4{\gamma\lambda^2}f^2\Big)\nonumber\\
&=\int dvd\hat vdfd\hat f \exp\Big(N\hat vv +N\hat f f -\frac N4\lambda^2v^2\Big(\mu+\frac{\gamma\lambda^2v^2}2\Big)-\frac N4{\gamma\lambda^2}f^2\Big)\nonumber\\
&\qquad\times\int dP_X(\bx)\exp\Big(- \hat v \|\bx\|^2-\hat f \bx^\intercal \bY\bx -\frac12 \bx^\intercal \bR(v,\bY) \bx\Big)\nonumber \\
&=\int dvd\hat vdfd\hat f \exp\Big(N\hat vv +N\hat f f -\frac N4\lambda^2v^2\Big(\mu+\frac{\gamma\lambda^2v^2}2\Big)-\frac N4{\gamma\lambda^2}f^2\Big)\nonumber\\
&\qquad\times\int dP_X(\bx)\exp\Big(\frac12{\bx^\intercal \bJ(v,\hat v,\hat f,\bY)\bx}\Big),
\end{align}
where the overall symmetric interaction matrix of this ``Ising model'' is
\begin{align}
 \bJ(v,\hat v,\hat f,\bY):=-\bR(v,\bY) -2\hat v \boldsymbol{I}_N -2\hat f \bY.
\end{align}

Now, defining the free entropy at fixed $(v,\hat v,\hat f)$
\begin{align}
\Phi_N(v,\hat v,\hat f,\bY):=\ln \int dP_X(\bx)\exp\Big(\frac12{\bx^\intercal \bJ(v,\hat v,\hat f,\bY)\bx}\Big)    ,
\end{align}
because the prior is factorized and we have an Ising-type of model, we can directly use the AdaTAP result \cite{adatap}: it tells us that
%
%
%
\begin{align}
&\Phi_N(v,\hat v,\hat f,\bY)= -{\rm extr}_{\bmm,\boldsymbol\tau,\bV}\Big\{\frac12\bmm^\intercal \bJ(v,\hat v,\hat f,\bY)\bmm\nonumber\\
&\quad+\frac12\ln \det\big(\boldsymbol\Omega-\bJ(v,\hat v,\hat f,\bY)\big)-\frac12 \bV^\intercal\bmm^2+\frac12\sum_{i\le N}\ln (\tau_i-m_i^2)\nonumber \\
& \quad-\sum_{i\le N}\ln\int dP_X(x)\exp\Big(\frac12V_ix^2+\big((\bJ(v,\hat v,\hat f,\bY)\bmm)_i -V_im_i\big) x\Big)\Big\}+o_N(1).\label{adatap_generic_simple}
\end{align}
The extremization is over $(\bmm,\boldsymbol\tau,\bV)\in\mathbb{R}^N\times(\mathbb{R}_{\ge 0}^N)^2$,  $\bmm^2=(m_i^2)_{i\le N}$, and the diagonal matrix 
\begin{align}
\boldsymbol\Omega:={\rm diag}(\bV+(\btau-\bmm^2)^{-1}).\label{Omega}    
\end{align}
Let the bracket notation  $\langle \,\cdot\, \rangle$ be used as expectation with respect to the posterior \eqref{posterior}, while $\langle \,\cdot\,\rangle_{\backslash i}$ is the mean with respect to the Gibbs measure of the ``cavity graph'' where $(J_{ij})_{j}$ are set to $0$. Define also the cavity fields $$h_i:=(\bJ\bx)_i.$$ The 
various variables at their extremum values are (asymptotically exact approximations to) the marginals means, second moments and variances of the cavity fields $$ m_i=\langle x_i\rangle,\quad \tau_i=\langle x_i^2\rangle, \quad V_i=\langle h_i^2\rangle_{\backslash i}-\langle h_i\rangle_{\backslash i}^2.$$
From the AdaTAP free entropy at fixed $(v,\hat v,\hat f)$ we can compute the total log-partition function by saddle-point and get 
\begin{align}\label{AdatapF}
&\frac{1}N\ln \mathcal{Z}(\bY)\propto o_N(1) \nn
&\quad+{\rm extr}\Big\{\hat vv +\hat f f -\frac 14\lambda^2v^2\Big(\mu+\frac{\gamma\lambda^2v^2}2\Big)-\frac 14{\gamma\lambda^2}f^2+\Phi_N(v,\hat v,\hat f,\bY)\Big\}
\end{align}
where the extremization is over $(v,\hat v,f,\hat f)$.

\subsection{Saddle point: reduction to an Ising model, AdaTAP equations and optimal pre-processing of the data}\label{statCondADATAP} 
By extremization of the AdaTAP single-instance free entropy \eqref{AdatapF} we derive the AdaTAP equations. We start with the intensive parameters. The extremization with respect to $f$ is trivial and gives $$\hat f=\frac12\gamma\lambda^2f.$$ So the leading order of the AdaTAP free entropy simplifies to
\begin{align}
{\rm extr}_{v,\hat v,f}\Big\{\hat vv +\frac14\gamma\lambda^2f^2 -\frac 14\lambda^2v^2\Big(\mu+\frac{\gamma\lambda^2v^2}2\Big)+\Phi_N\Big(v,\hat v,\frac12\gamma\lambda^2f,\bY\Big)\Big\}.
\end{align}
The remaining saddle point equations can simply be written down. But this is not necessary as
%
the solution of the three remaining intensive order parameters at the saddle point is simply deduced from their physical meaning, concentration properties, and the Nishimori identity: in the large size limit, $$v\to \lim_{N\to \infty}\frac1N\EE\langle \|\bx\|^2\rangle=\lim_{N\to \infty}\frac1N\EE \|\bX^*\|^2=1,$$ 
as well as (recall \eqref{xyx})
$$f\to\lim_{N\to \infty} \frac1N\EE\langle \bx^\intercal \bY\bx\rangle=\lambda.$$
Moreover we know that $$\hat v\to 0$$ because the prior is already enforcing the constraint that $v=\|\bx\|^2/N\to 1$ in \eqref{Zadatap} without the need to introducing a further, redundant, delta constraint; note that for Rademacher or spherical prior this is simply true as no delta function is needed. Therefore the AdaTAP free entropy becomes
\begin{align}
\frac18\gamma\lambda^4 -\frac 14\mu\lambda^2+\Phi_N\Big(1,0,\frac12\gamma\lambda^3,\bY\Big)+o_N(1).\label{finalAdaTAP}
\end{align}

From this AdaTAP free entropy we see that the values of the marginal means and variances correspond to the solution of the variational problem \eqref{adatap_generic_simple} with interaction matrix 
\begin{align}\label{optimalJ}
 \bJ\Big(1,0,\frac12\gamma\lambda^3,\bY\Big)=\mu\lambda \bY-\gamma\lambda^2\bY^2+\gamma\lambda\bY^3 =:\bJ(\bY) .
\end{align}
So we end-up with the following effective partition function of an Ising-like model:
\begin{align}
 \int dP_X(\bx)\exp\Big(\frac12{\bx^\intercal \bJ(\bY)\bx}\Big)  . 
\end{align}
This shows that the original model is equivalent to an Ising model with interaction matrix $\bJ(\bY)$, which can thus be interpreted as a \emph{Bayes-optimal pre-processing of the data}. This will be verified in Section \ref{sec:AMP}, as the use of $\bJ(\bY)$ instead of $\bY$ will turn AMP into an optimal algorithm. Ising models like this are precisely studied in \cite{adatap} and we can therefore again exploit directly the AdaTAP formalism.
Let 
\begin{align}
    \eta_i(\bJ,\bmm,V_i)&:=\frac{\int dP_X(x) \,x\, e^{\frac12V_ix^2+((\bJ\bmm)_i -V_im_i) x}}{\int dP_X(x)e^{\frac12V_ix^2+((\bJ\bmm)_i -V_im_i) x}},\label{TAPm}\\
    g_i(\bJ,\bmm,V_i)&:=\frac{\int dP_X(x) \,x^2\, e^{\frac12V_ix^2+((\bJ\bmm)_i -V_im_i) x}}{\int dP_X(x)e^{\frac12V_ix^2+((\bJ\bmm)_i -V_im_i) x}}.
\end{align}
The associated AdaTAP equations over 
$(\bmm,\btau,\bV)$, namely the saddle point equations associated with the AdaTAP free entropy \eqref{adatap_generic_simple} with $\bJ(v,\hat v,\hat f,\bY)$ replaced by $\bJ=\bJ(\bY)$, read
\begin{align}
    m_i&=\eta_i(\bJ,\bmm,V_i),\label{tap1}\\
    \tau_i&=g_i(\bJ,\bmm,V_i),\\
    \tau_i-m_i^2&=\big([{\rm diag}(\bV+(\btau-\bmm^2)^{-1})- \bJ]^{-1}\big)_{ii},\label{tap3}
\end{align}
where the last equation is understood as an implicit equation for $\bV$.

\subsection{Optimal pre-processing for the order 6 potential}\label{sec:preprocessing_sestic}
{Let us now consider the pure sestic ensemble with matrix potential $V(x)=\xi x^6/6$, and $\xi=27/80$.}
With the same notations as in the previous section, the trace of the matrix potential now reads
\begin{align}
    &\frac{6}{\xi}\Tr V(\mathbf{Y}-\lambda \mathbf{p})=C+ \Tr\mathbf{p}\Big[
    -6\lambda \mathbf{Y}^5+6\lambda^2v\bY^4-6\lambda^3v^2\bY^3+6\lambda^4v^3\bY^2-6\lambda^5v^4\bY
    \Big]\nonumber\\
    &\qquad+6\lambda^2 f_3f_1+3\lambda^2f_2^2-12\lambda^3vf_2f_1-2\lambda^3f_1^3+9\lambda^4v^2f_1^2+\lambda^6v^6
\end{align}
where we have introduced the parameters $f_j=\Tr\bY^j\mathbf{p}$, $j=1,2,3$. Hence the Hamiltonian in the posterior measure \eqref{posterior} can be written as
\begin{align}
        &-\frac{N}{2}\Tr V(\bY-\lambda \mathbf{p})\propto-\frac{\xi\mathbf{x}^\intercal \mathbf{R}_6(v,\bY)\mathbf{x}}{2}\nonumber\\
        &\qquad-\frac{N\xi}{2}\Big(\lambda^2 f_3f_1+\frac{1}{2}\lambda^2f_2^2-2\lambda^3vf_2f_1-\frac{1}{3} \lambda^3 f_1^3 +\frac{3}{2}\lambda^4v^2f_1^2+\frac{\lambda^6v^6}{6}\Big)
\end{align}
with $\mathbf{R}_6(v,\bY)=-\lambda \mathbf{Y}^5+\lambda^2v\bY^4-\lambda^3v^2\bY^3+\lambda^4v^3\bY^2-\lambda^5v^4\bY$.

Now, as for the quartic potential, we need to fix the order parameters $(f_j)_{j=1,2,3}$ and $v$ introducing the Fourier conjugates $(\hat{f}_j)_{j=1,2,3}$ and $\hat{v}$, which produces additional two-body interaction terms. The partition function then reads
\begin{align}\label{Zadatapsestic}
&\mathcal{Z}=\int dvd\hat v \prod_{j=1}^3df_j d\hat f_j \exp\Big[N\hat vv +N\sum_{j=1}^3\hat f_j f_j -\frac{N\xi}{2}\Big(\lambda^2 f_3f_1+\frac{1}{2}\lambda^2f_2^2-2\lambda^3vf_2f_1\nonumber\\
&-\frac{1}{3} \lambda^3 f_1^3 +\frac{3}{2}\lambda^4v^2f_1^2+\frac{\lambda^6v^6}{6}\Big)
\Big]
\int dP_X(\bx)\exp\Big(\frac12{\bx^\intercal \bJ(v,\hat v,(\hat f_j)_{j=1,2,3},\bY)\bx}\Big)\,,
\end{align}
where
\begin{multline}
    \bJ(v,\hat v,(\hat f_j)_{j=1,2,3},\bY)=\xi\lambda\bY^5-\xi\lambda^2v\bY^4+(\xi\lambda^3v^2-2\hat{f}_3)\bY^3-(\xi\lambda^4v^3+2\hat{f}_2)\bY^2\\
    +(\xi\lambda^5v^4-2\hat{f}_1)\bY-2\hat{v}\boldsymbol{I}_N\,.
\end{multline}
Using the Nishimori identities we are already able to fix the values of $v$ and $\hat{v}$ to $1$ and $0$ respectively.
Furthermore, in order to have an explicit $\bJ(\bY)$, we also need to fix the remaining $\hat{f}_j$'s. Without repeating all the procedure, we notice that their values are determined by the argument of the first exponential in \eqref{Zadatapsestic}. In particular, it suffices to equate to zero its gradient w.r.t. $(f_j)_{j=1,2,3}$ to obtain the system of equations:
\begin{align}
    \hat{f}_1&=\frac{\xi\lambda^2}{2}f_3-\xi\lambda^3 f_2-\frac{\xi\lambda^3}{2}f_1^2+\frac{3\xi\lambda^4}{2}f_1\,,\\
    \hat{f}_2&=\frac{\xi\lambda^2}{2}f_2-\xi\lambda^3f_1\,,\\
    \hat{f}_3&=\frac{\xi\lambda^2}{2}f_1\,,
\end{align}
where we have already set $v=1\,,\,\, \hat{v}=0$.
The values of $(f_j)_{j=1,2,3}$ can be fixed again by the Nishimori identities; indeed, in the thermodynamic limit one has
\begin{align}\label{Nishif_12}
    f_1\to \lambda\,,\quad f_2\to \lim_{N\to\infty}\frac{1}{N}\mathbb{E}\langle\mathbf{x}^\intercal\bY^2\bx\rangle=1+\lambda^2
\end{align}
and also
\begin{align}\label{Nishif_3}
    f_3\to\lim_{N\to\infty}\frac{1}{N}\mathbb{E}\langle\mathbf{x}^\intercal\bY^3\bx\rangle=\lambda^3+2\lambda +\lambda\lim_{N\to\infty}\mathbb{E}\Big[\frac{1}{N}\Tr\bZ\bX^*\bX^{*\intercal}\Big]^2=\lambda^3+2\lambda\,.
\end{align}
The expectation of $(\Tr\bZ\bX^*\bX^{*\intercal}/N)^2$ vanishes because the variable $\Tr\bZ\bX^*\bX^{*\intercal}/N$ concentrates around $0$. Gathering all these results and plugging them into $\bJ$ we finally get the pre-processing matrix that should lead to an AMP algorithm with Bayes-optimal performance:
\begin{align}
    \label{pre-processed-sestic}
    \bJ(\mathbf{Y})&=\bJ\Big(1,0,\hat f_1=\xi\frac{\lambda^5}{2},\hat f_2
    =\xi\frac{\lambda^2-\lambda^4}{2},
    \hat f_3=\xi\frac{\lambda^3}{2},\bY\Big)\nonumber\\
    &=\xi\lambda\bY^5-\xi\lambda^2\bY^4-\xi\lambda^2\bY^2\,.
\end{align}

\subsection{Simplifying the AdaTAP equations by self-averaging of the Onsager reaction term}\label{sec:selfaveraging}

The variances $V_i$ are expected to be self-averaging with respect to the interaction matrix, i.e., in the large size limit $V_i=\bar V:=\lim_{N\to \infty}\EE_\bJ V_i$. The computation we are going to carry out now could be performed in various ways leading to different but equivalent expressions. For pedagogical reasons we take a path that remains as close as possible to the approach of \cite{adatap}. Following this reference we compute the expectation of the AdaTAP equation for $\bV$. In this section, all quantities $\bV$, $\bmm$ and $\btau$ are fixed to a solution of the AdaTAP equations \eqref{tap1}--\eqref{tap3}. 

We start from the convenient identity
\begin{align}
     \big([\boldsymbol\Omega- \bJ]^{-1}\big)_{ii}=\partial_{\Omega_{ii}}\ln \det (\boldsymbol\Omega- \bJ).\label{111}
\end{align}
We are going to average the right-hand side. As for a Gaussian model there is no spin glass phase and strong concentrations take place, the quenched and annealed averages match \cite{adatap}: we can thus simply compute the logarithm of the average of the determinant. A Gaussian identity then gives
\begin{align}
\EE \det (\boldsymbol\Omega- \bJ)^{-1/2}= \int \frac{d\bz}{(2\pi)^{N/2}}\exp\Big(-\frac12 \bz^\intercal\boldsymbol\Omega\bz\Big) \EE\exp\Big(\frac12\bz^\intercal \bJ\bz\Big) .   \label{112}
\end{align}
We denote $\bJ=\sum_{k\le 3}c_k\bY^k$ where $\bc=(\mu\lambda,-\gamma\lambda^2,\gamma\lambda)$. The term we need to compute therefore reads
\begin{align}
   &\EE\exp\Big(\frac12\bz^\intercal \bJ\bz\Big) 
   =\EE \exp\frac12(\bz^\intercal (c_1 \bY+c_2\bY^2+c_3\bY^3)\bz )\label{starting}
\end{align}
Define the order parameters
\begin{align}
p:=\frac{1}N\bz^\intercal \bX^*,\quad v:=\frac{1}{N}\|\bz\|^2, \quad p_D:=\frac1N(\bO \bz)^\intercal\bD\bO\bX^*.
\end{align}
We also have $\|\bX^*\|^2/N=1+o_N(1)$.
Our goal is to identify the generalized spherical integral \eqref{gene_sphInt}. Replacing $\bY$ by $\lambda \bp^* +\bO^\intercal\bD\bO$ (with $\bp^*:=\bX^*\bX^{*\intercal}/N$) we expand the various terms. The first term is then simply
\begin{align}
c_1\bz^\intercal \big(\lambda\bp^*+\bO^\intercal\bD\bO\big)\bz=c_1\big(\lambda Np^2 +(\bO\bz)^\intercal \bD\bO\bz\big).
\end{align}
The second term is 
\begin{align}
    &c_2\bz^\intercal \big(\lambda^2(\|\bX^*\|^2/N) \bp^* +\lambda\bp^*\bO^\intercal\bD\bO + \lambda\bO^\intercal\bD\bO \bp^* +\bO^\intercal\bD^2\bO\big)\bz\nonumber\\
    &=c_2\big(N\lambda^2  p^2+2N\lambda p p_D  +(\bO\bz)^\intercal \bD^2 \bO\bz\big) +o(N).
\end{align}
Finally the last term is a bit more cumbersome:
\begin{align}
    &c_3\bz^\intercal \big(\lambda^3(\|\bX^*\|^4/N^2) \bp^* +\lambda^2\bp^*\bO^\intercal\bD\bO\bp^* + \lambda^2(\|\bX^*\|^2/N)\bO^\intercal\bD\bO \bp^* +\lambda\bO^\intercal\bD^2\bO\bp^*\nonumber\\
    &\quad+\lambda^2(\|\bX^*\|^2/N) \bp^*\bO^\intercal\bD\bO +\lambda\bp^*\bO^\intercal\bD^2\bO + \lambda\bO^\intercal\bD\bO \bp^*\bO^\intercal\bD\bO +\bO^\intercal\bD^3\bO\big)\bz\nonumber\\
    &=c_3 \big(N  \lambda^3 p^2+\lambda^2p^2 (\bO\bX^*)^\intercal \bD\bO\bX^*+2\lambda p (\bO\bz)^\intercal\bD^2\bO\bX^* \nonumber\\
    &\quad+2N\lambda^2 pp_D+\lambda N p_D^2+(\bO\bz)^\intercal\bD^3\bO\bz\big)+o(N).
\end{align}
Combining all we reach
\begin{align}
   &\EE\exp\Big(\frac12\bz^\intercal \bJ\bz\Big) =\int d\btau d\hat\btau \,\exp\Big(NK+\frac12\hat v\|\bz\|^2-\frac N2\hat v v+o(N)\Big)\nonumber\\
  &\qquad\times\EE_\bO\exp\Big((\bO\bz)^\intercal \bC_{\bz ,\bz} \bO\bz+ (\bO\bX^*)^\intercal \bC_{*,*}\bO\bX^*+ (\bO\bX^*)^\intercal\bC_{\bz,*}\bO\bz\Big)\label{114}
\end{align}
with $d\btau:=(dp,dv,dp_D)$ and $d\hat \btau:=(d\hat p,d\hat v,d\hat p_D)$, and (all coupling matrices below are $N\times N$ and symmetric)
\begin{align*}
     K&:=\frac12\big(\mu\lambda^2 p^2-\gamma\lambda^2(\lambda^2p^2+2\lambda p p_D) +\gamma\lambda(\lambda^3p^2+2\lambda^2pp_D+\lambda p_D^2)+\hat p p  +\hat p_D p_D\big),\\
    \bC_{*,*}&:= \frac12\gamma\lambda^3p^2\bD,\\
    \bC_{\bz,\bz}&:=\frac12\big(\mu\lambda \bD-\gamma\lambda^2  \bD^2+\gamma \lambda \bD^3\big),\\
    \bC_{\bz,*}&:=\frac12\big(-\hat pI_N -\hat p_D \bD+2\gamma\lambda^2 p\bD^2\big).
\end{align*}
Note the asymmetry for the variable $\hat v$ compared to the other hat-variables, which has not been injected in the definition of the coupling matrices as the others, but instead leads to a term appearing explicitly in \eqref{114} (both choices are equivalently valid ones). The term averaged over $\bO$ is an inhomogeneous spherical integral as studied in Section~\ref{appendix_spherical_generalized}. In particular, we are in the case of Section \ref{Low-rank_int_for_replicas} with $\ell\in\{0,1\}$ with the exception that $\bX^*$ also (playing the role of the $0$th replica) has a non-zero self-coupling. So this trivial modification of the computation of Section \ref{Low-rank_int_for_replicas} yields 
\begin{align*}
  \EE e^{\frac12\bz^\intercal \bJ\bz} &=\int d\btau d\hat\btau \,\exp \Big(NK+\frac12\hat v\|\bz\|^2-\frac N2\hat v v +NI_\bC(p,v,\hat p, \hat p_D) +o(N)\Big)
\end{align*}
where the $2\times 2$ random coupling matrix $\bC$ has entries
\begin{align}
    2C_{00} &=\gamma\lambda^3 p^2 D, \\
    2C_{11}&=\mu\lambda D-\gamma\lambda^2  D^2+\gamma \lambda D^3,\\
    2C_{01}=2C_{10}&=\frac12(-\hat p -\hat p_D D+2\gamma\lambda^2 pD^2),
\end{align}
with $D\sim \rho$ drawn from the noise asymptotic spectral density, and  
\begin{align}
    &I_\bC(p,v,\hat p, \hat p_D)=\frac12 {\rm extr}_{(\tilde v_0,\tilde v,\tilde p)} \Big\{\tilde v_0  + 2 \tilde p p + \tilde v v \nonumber\\
    &\qquad-\EE\ln\big((\tilde v_0-2C_{00})(\tilde v-2C_{11})-(\tilde p-2C_{01})^2\big)\Big\}-\frac12\ln( v-p^2)-1.\label{IC_125}
\end{align}
One can check that $I_\bC$ is null when $C_{00}=C_{11}=C_{01}$ as it should. Therefore equation \eqref{112} becomes at leading exponential order
\begin{align*}
&\ln\EE \det (\boldsymbol\Omega- \bJ)^{-1/2}\\
&\qquad= \ln\int \frac{d\bz}{(2\pi)^{N/2}} d\btau d\hat\btau \exp\Big(-\frac12 \bz^\intercal(\boldsymbol\Omega-\hat v I_N)\bz+NK -\frac N2\hat v v+NI_\bC+o(N)\Big) \nonumber  \\
&\qquad= \ln\int d\btau d\hat\btau \exp\Big(NK-\frac N2\hat v v+ NI_\bC-\frac 12\ln\det (\boldsymbol\Omega-\hat v I_N)+o(N) \Big)\\
&\qquad={\rm extr}\Big\{NK-\frac N2\hat v v+ NI_\bC-\frac 12\ln\det (\boldsymbol\Omega-\hat v I_N) \Big\}+o(N),
\end{align*}
where we used Gaussian integration followed by a saddle point estimation. By the aforementioned strong concentration properties of the Gaussian model, this is also equal to $-\frac12\ln \EE\det (\boldsymbol\Omega- \bJ)\approx -\frac12\EE\ln \det (\boldsymbol\Omega- \bJ)$ so we reach at leading order
\begin{align}
\EE\ln \det (\boldsymbol\Omega- \bJ)&\approx {\rm extr}\big\{-2NK+N\hat v v -2NI_\bC+\ln\det (\boldsymbol\Omega -\hat v I_N) \big\}\nn
&={\rm extr}_{(\hat v,v)}\Big\{N\hat v v +\sum_{i\le N}\ln(\Omega_{ii} -\hat v)-2N \tilde G(v)\Big\}\label{ToExtr}
\end{align}
where the extremization is over all variables and 
\begin{align}
\tilde G(v):={\rm extr}_{(p,p_D,\hat p, \hat p_D)}\big\{I_\bC(p,v,\hat p, \hat p_D)+K(p,p_D,\hat p,\hat p_D)\big\}.  \label{tildeG}  
\end{align}
This is the analogue of the G-function appearing, e.g., in \cite{adatap}. The extremization over $\hat v$ in \eqref{ToExtr} yields that at the saddle point, $$v = \frac1N \sum_{i\le N}\frac1{\Omega_{ii}-\hat v}.$$ Moreover, combining the TAP equation \eqref{tap3} with \eqref{111} and \eqref{ToExtr} we have
\begin{align}
\EE(\tau_i-m_i^2)=\partial_{\Omega_{ii}}\EE\ln \det (\boldsymbol\Omega- \bJ)=\frac1{\Omega_{ii}-\hat v}\label{usefulId}    
\end{align}
where $\hat v$ is evaluated at its saddle point value. Therefore, summing over $i$ the last identity and recalling the definition of $\Omega_{ii}$ we reach  
\begin{align}
\bar \chi:=\frac1N \EE\sum_{i\le N}(\tau_i-m_i^2)=v=\frac 1N\EE\sum_{i\le N}\frac1{V_i+(\tau_i-m_i^2)^{-1}-\hat v}.    
\end{align}
Under the concentration assumption $V_i=\bar V$ for all $i\le N$, this identity implies 
\begin{align}
    V_i=\hat v.
\end{align}
Additionally the saddle point equation for $v$ extracted from \eqref{ToExtr} yields
\begin{align}
\hat v=2 \partial_v\tilde G(v)|_{v=\bar \chi} \quad \Rightarrow \quad V_i= \bar V :=   2 \partial_v\tilde G(v)|_{v=\bar \chi}.\label{barVtildeG}
\end{align}
The variable $\bar \chi$ is instance-independent and can be deduced from our replica theory: it is equal to twice the MMSE \eqref{replicaMMSE}, namely,
\begin{align}
\bar \chi = 1-m^2    
\end{align}
where $m$ is solution to the replica fixed point equations \eqref{m_replica_simplified}--\eqref{hatmReplica}. Computing $\bar V$ from \eqref{barVtildeG} is then easy, as taking a derivative w.r.t. $v$ of $\tilde G(v)$ is straightforward: all the quantities appearing on the right-hand side of \eqref{tildeG} are at the saddle point, so it simply amounts to a partial derivative of \eqref{IC_125}. It gives 
\begin{align}
 \bar V =    \tilde v-\frac{1}{1-m^2-p^2}
\end{align}
where $\tilde v=\tilde v(p,v)$ takes its saddle point value from \eqref{IC_125} while $p=p(v)$ from \eqref{tildeG} with $v=\bar \chi$ fixed.

Thanks to these simplifications the AdaTAP equation reads in the large size limit
\begin{align}
    m_i&=\eta_i(\bJ,\bmm,\bar V).
\end{align}
Or, when written in a fashion closer to the form of AMP algorithms, the AdaTAP equations read
\begin{align}
  \bff= \bJ\bmm -\bar V \bmm,\qquad \bmm=\eta_{\bar V}(\bff),
\end{align}
where the ``denoiser'', which is applied component-wise above, is
\begin{align}
    \eta_{\bar V}(f):=\frac{\int dP_X(x) \,x\, e^{\frac12\bar V x^2+ fx}}{\int dP_X(x)e^{\frac12\bar Vx^2+f x}}.\label{denoiserV}
\end{align}

\section{Approximate message passing, optimally}\label{sec:AMP}

We will now describe an AMP algorithm that matches the replica prediction for the minimum mean-square error. We therefore conjecture it to be Bayes-optimal and refer to it as BAMP. The main difference between this new AMP and the previously proposed one for structured PCA is that it is constructed from iterates based on the pre-processed matrix $\bJ(\bY)$ rather than $\bY$ as in \cite{fan2022approximate}. Consequently, the Onsager reaction terms will have to be adapted. {Finally, inspired by the structure of BAMP, we present a choice of denoisers in the AMP of \cite{fan2022approximate} which alternates between linear functions and posterior means given all the previous iterates (AMP with Alternating Posteriors, AMP-AP). The numerical results of the following section will show that AMP-AP matches the BAMP performance and, therefore, the replica prediction.}

\subsection{BAMP: Bayes-optimal AMP}\label{subsec:BoptAMP}

The AdaTAP approach described in Section \ref{sec:AdaTAPtoward} suggests that, in order to achieve Bayes-optimal performance, one should consider the BAMP iteration which is of the form
\begin{equation}\label{eq:AMPnew}
    \bff^t = \bJ(\bY) \bu^t - \sum_{i=1}^t{\sf c}_{t, i}\bu^i, \quad \bu^{t+1} = g_{t+1}(\bff^t), \quad t\ge 1.
\end{equation}
As in the AMP iteration \eqref{eq:AMPZF}, the \emph{denoiser} function $g_{t+1}:\mathbb R\to\mathbb R$ is continuously differentiable, Lipschitz and applied  component-wise. Crucially, the Onsager coefficients $\{{\sf c}_{t, i}\}_{i\in [t], t\ge 1}$ need to ensure that, conditioned on the signal, the empirical distribution of the iterate $\bff^t$ is Gaussian, namely, the convergence result in \eqref{eq:SEZF} holds for some mean vector $\bmu_t$ and covariance matrix $\bSigma_t$.

We highlight that the matrix $\bY$ in \eqref{eq:AMPZF} is replaced by the matrix $\bJ(\bY)$ in \eqref{eq:AMPnew}. This means that the state evolution result of \cite{fan2022approximate} cannot be applied and the Onsager coefficients $\{{\sf c}_{t, i}\}_{i\in [t], t\ge 1}$ will have a different form with respect to $\{{\sf b}_{t, i}\}_{i\in [t], t\ge 1}$.

In what follows, we will consider the general case in which $\bJ(\bY)$ is an arbitrary polynomial of degree $K$ in $\bY$, namely, $$\bJ(\bY) = \sum_{i\le K} c_i \bY^i.$$ To compute $\{{\sf c}_{t, i}\}_{i\in [t], t\ge 1}$ and obtain a state evolution result for the iteration \eqref{eq:AMPnew}, the key idea is to map the first $T$ iterations of \eqref{eq:AMPnew} to the first $K \times T$ iterations of an \emph{auxiliary} AMP with iterates $(\tilde \bz^t, \tilde \bu^t)_{t\in [KT]}$ and denoisers $\{\tilde h_{t+1}\}_{t\in [K T]}$, whose state evolution can be deduced from \cite{fan2022approximate}. The denoisers $\{\tilde h_{t+1}\}_{t\in [KT]}$ of this auxiliary AMP are chosen so that, for $t\in [T]$  and $\ell\in [K]$, 
\begin{align}
    \lim_{N\to\infty}\frac{1}{N}&\|\tilde\bu^{K(t-1)+\ell} - \bY^{\ell-1}\bu^t\|_2^2=0.\label{eq:mapping1}
\end{align}
More specifically, for $t\in [T]$  and $\ell\in \{2, \ldots, K\}$, the denoiser $\tilde h_{K(t-1)+\ell}$ giving $\tilde\bu^{K(t-1)+\ell}$ is a linear combinations of the past iterates $\tilde \bu^1, \ldots, \tilde \bu^{K(t-1)+\ell-1}$ and of $\tilde \bz^{K(t-1)+\ell-1}$; furthermore, the coefficients of these linear combinations are chosen to ensure that $\tilde\bu^{K(t-1)+\ell}\approx \bY^{\ell-1}\bu^t$. Hence, from $\tilde \bz^{Kt}$ and $(\tilde \bu^{K(t-1)+\ell})_{\ell\in \{2, \ldots, K\}}$, one obtains $(\bY^{\ell}\bu^t)_{\ell\in [K]}$ (up to an $o_N(1)$ error). As a result, $\bJ(\bY) \bu^t$ can be expressed as a linear combination of $(\tilde \bu^1, \ldots, \tilde \bu^{Kt}, \tilde \bz^{Kt})$, which in turn is a linear combination of \emph{(i)} the past iterates $\{\bu^i\}_{i\in [t]}$, \emph{(ii)}
the signal $\bX^*$, plus \emph{(iii)} independent Gaussian noise. By inspecting the coefficients of this linear combination, one deduces \emph{(i)} the values of 
the Onsager coefficients $\{{\sf c}_{t, i}\}_{i\in [t], t\ge 1}$ (as the coefficients multiplying the past iterates $\{\bu^i\}_{i\in [t]}$), \emph{(ii)} the mean $\mu_t$ (as the coefficient multiplying the signal $\bX^*$), and \emph{(iii)} the covariance matrix $\bSigma_t$ (as the covariance matrix of the remaining noise terms). Finally, by making $\tilde h_{Kt+1}$ depend on $g_{t+1}$, we enforce that $\tilde\bu^{Kt+1}\approx \bu^{t+1}$. We highlight that the auxiliary AMP is employed purely as a proof technique. Its formal description is deferred to Appendix \ref{appsubsec:aux}, and its state evolution follows in Appendix \ref{appsubsec:auxAMPSE}. 

For simplicity, we assume to have access to an initialization $\bu^1\in \mathbb R^N$, which is independent of the noise $\bZ$ and has a strictly positive correlation with $\bX^*$, i.e., 
\begin{equation}\label{eq:AMPinit}
(\bX^*, \bu^1)\stackrel{\mathclap{W_2}}{\longrightarrow} (X^*, U_1), \quad \mathbb E[X^* \,U_1]:=\epsilon>0, \quad \mathbb E[U_1^2]= 1.
\end{equation}
The requirement \eqref{eq:AMPinit} is rather standard in the analysis of AMP algorithms. However, as having access to such an initialization is often impractical, a recent line of work has designed AMP iterations which are initialized with the eigenvector of the data matrix $\bY$ associated to the largest eigenvalue, see \cite{montanari2017estimation,mondelli2021pca,zhong2021approximate}. By following the approach detailed in \cite{mondelli2021pca}, one can design a Bayes-optimal AMP with spectral initialization. As this would be out of the scope of the current contribution -- whose goal is to obtain an algorithm with a Bayes-optimal \emph{fixed point} -- we will not pursue this extension here. 

\subsection{Onsager coefficients and state evolution recursion}
\label{subsec:OnsSE}

We now detail the calculation of the Onsager coefficients $\{{\sf c}_{t, i}\}_{i\in [t], t\ge 1}$ and of the state evolution parameters $\bmu_t, \bSigma_t$ associated to the AMP algorithm \eqref{eq:AMPnew}. We obtain these quantities from the state evolution recursion of the auxiliary AMP which, up to a $o_N(1)$ error, tracks $(\bY^{\ell-1}\bu^t)_{\ell\in [K]}$ and, as such, has a number of iterations $K$ times larger. To express the latter, we define a number of auxiliary quantities: the vector $\tilde \bmu_{Kt}\in \mathbb R^{Kt}$, the matrices $\tilde\bDelta_{Kt},\tilde\bPhi_{Kt},\tilde \bSigma_{Kt},\tilde \bB_{Kt} \in \mathbb R^{Kt\times Kt}$, and the coefficients $\{\alpha_{i, j}\}_{j\in [i], i\in [Kt]}$, $\{\beta_{i, j}\}_{j\in [\lfloor(i-1)/K\rfloor+1], i\in [Kt]}$, $\{\gamma_i\}_{i\in [Kt]}$, $\{\theta_{i, j}\}_{i\in [t], j\in [Kt]}$. The quantities $\tilde \bmu_{Kt}$, $\tilde\bDelta_{Kt}$, $\tilde\bPhi_{Kt}$, $\tilde \bSigma_{Kt}, \tilde \bB_{Kt}$ are directly connected to the state evolution of the auxiliary AMP (see the remark at the end of Appendix \ref{appsubsec:auxAMPSE}). Furthermore, the coefficients $\{\alpha_{i, j}\}_{j\in [i], i\in [Kt]}$, $\{\beta_{i, j}\}_{j\in [\lfloor(i-1)/K\rfloor+1], i\in [Kt]}$, $\{\gamma_i\}_{i\in [Kt]}$, $\{\theta_{i, j}\}_{i\in [t], j\in [Kt]}$ allow for a useful (approximate) decomposition of the vectors $(\bY^{\ell}\bu^t)_{\ell\in [K-1]}$, see the remark at the end of this section.

We start with the initialization 
\begin{equation}
    \tilde U_1 := U_1,
\end{equation}
where $U_1$ satisfies \eqref{eq:AMPinit}, and we set
\begin{equation}
\begin{split}\label{eq:SEnewinit}
    \tilde\mu_1 &:= \lambda \epsilon, \quad (\tilde\bDelta_1)_{1, 1} := 1, \quad (\tilde\bPhi_1)_{1, 1} := 0,\quad (\tilde\bB_1)_{1, 1} :=\bar\kappa_1, \quad (\tilde\bSigma_1)_{1, 1} :=\bar\kappa_2,\\
    \alpha_{1, 1}&:=0, \quad \beta_{1, 1} := 1, \quad \gamma_1 := 0.
\end{split}
\end{equation}
Here and in what follows, we denote by $\{\bar \kappa_k\}_{k\ge 1}$ the sequence of free cumulants associated to $D$. The free cumulants can be recursively computed from the moments, see e.g. \cite[Section 2.5]{novak2014three}.  

For $t\ge 1$, let us define
\begin{align}
        \tilde U_{K(t-1)+1+\ell} &:= \tilde Z_{K(t-1)+\ell} + \tilde \mu_{K(t-1)+\ell} X^* + \hspace{-1em}\sum_{j=1}^{K(t-1)+\ell} \hspace{-.5em}(\tilde \bB_{K(t-1)+\ell})_{K(t-1)+\ell, j}\tilde U_j,\, \ell\in [K-1],\label{eq:3ip0}\\ 
        \tilde U_{Kt+1} &:= g_{t+1}\Big(\mu_{t}X^* + \sum_{j=1}^{Kt}\theta_{t, j}\tilde Z_{j}\Big),\label{eq:3ip1}\\ 
    (\tilde Z_1, \ldots, \tilde Z_{Kt})&\sim \normal (0, \tilde\bSigma_{Kt}) \mbox{ and independent of } X^*, U_1. \label{eq:Zdef}
\end{align}
We note that the function $g_{t+1}$ in \eqref{eq:3ip1} is the AMP denoiser in \eqref{eq:AMPnew}. 
Let us also define
\begin{align}
    \tilde\mu_{K(t-1)+1+\ell} &= \lambda \mathbb E[\tilde U_{K(t-1)+1+\ell}X^*],\label{eq:tildemuup}\\
    (\tilde\bDelta_{K(t-1)+1+\ell})_{K(t-1)+1+\ell, j} &=(\tilde\bDelta_{K(t-1)+1+\ell})_{j, K(t-1)+1+\ell} = \mathbb E[\tilde U_{K(t-1)+1+\ell}\tilde U_{j}], \label{eq:Deltaup}\\
    &\hspace{14em} j\in [K(t-1)+1+\ell],\notag\\
    (\tilde\bPhi_{K(t-1)+1+\ell})_{K(t-1)+1+\ell, j} &= \mathbb E[\partial_{\tilde Z_j}\tilde U_{K(t-1)+1+\ell}], \quad j\in [K(t-1)+\ell],\label{eq:Phiup}\\
    \tilde \bB_{K(t-1)+1+\ell} &= \sum_{j=0}^{K(t-1)+\ell}\bar\kappa_{j+1}\tilde \bPhi_{K(t-1)+1+\ell}^j,\label{eq:Bup}\\
    \tilde \bSigma_{K(t-1)+1+\ell} = \sum_{j=0}^{2(K(t-1)+\ell)}&\bar\kappa_{j+2}\sum_{k=0}^j(\tilde \bPhi_{K(t-1)+1+\ell})^k\tilde\bDelta_{K(t-1)+1+\ell}(\tilde\bPhi_{K(t-1)+1+\ell}^\intercal)^{j-k}.\label{eq:tildeBup}
\end{align}
Now, we obtain $\tilde\bmu_{K(t-1)+1}, \tilde\bDelta_{K(t-1)+1}, \tilde\bPhi_{K(t-1)+1}, \tilde \bB_{K(t-1)+1}, \tilde \bSigma_{K(t-1)+1}$ by setting $\ell=0$ in \eqref{eq:tildemuup}--\eqref{eq:tildeBup} (and by using the initialization \eqref{eq:SEnewinit} for $t=1$). This allows us to define $\tilde U_{K(t-1)+2}$ by setting $\ell=1$ in \eqref{eq:3ip0}. Next, we obtain  $\tilde\bmu_{K(t-1)+2}$, $\tilde\bDelta_{K(t-1)+2}$, $\tilde\bPhi_{K(t-1)+2}$, $\tilde \bB_{K(t-1)+2}$, $\tilde \bSigma_{K(t-1)+2}$ by setting $\ell=1$ in \eqref{eq:tildemuup}--\eqref{eq:tildeBup}. This allows us to define $\tilde U_{K(t-1)+2}$ by setting $\ell=2$ in \eqref{eq:3ip0}. We iterate this procedure until we have obtained 
($\tilde\bmu_{K(t-1)+\ell}$, $\tilde\bDelta_{K(t-1)+\ell}$, $\tilde\bPhi_{K(t-1)+\ell}$, $\tilde \bB_{K(t-1)+\ell}$, $\tilde \bSigma_{K(t-1)+\ell}$)$_{\ell\in [K]}$ and $(\tilde U_{K(t-1)+1+\ell})_{\ell\in [K-1]}$.
We note that, for any $i\ge 1$, $\tilde \bB_{i}$ and $\tilde \bSigma_{i}$ are the top left sub-matrices of $\tilde \bB_{i+1}$ and $\tilde \bSigma_{i+1}$, respectively.

At this point, for $\ell\in [K-1]$, we compute the quantities $\{\alpha_{K(t-1)+1+\ell, j}\}_{j\in [K(t-1)+\ell]}$, $\{\beta_{K(t-1)+1+\ell, j}\}_{j\in [t]}$, $\gamma_{K(t-1)+1+\ell}$  as
\begin{align}
    \alpha_{K(t-1)+1+\ell, j} &= \delta_{K(t-1)+\ell, j} + \sum_{\substack{i=1 \\ i\not\equiv 1 ({\rm mod}\, K)}}^{K(t-1)+\ell}\alpha_{i, j}\,(\tilde\bB_{K(t-1)+\ell})_{K(t-1)+\ell, i}, \quad j\in [K(t-1)+\ell],\label{eq:alphaup}\\
    \beta_{K(t-1)+1+\ell, j} &= (\tilde\bB_{K(t-1)+\ell})_{K(t-1)+\ell, K(j-1)+1} +\hspace{-.5em} \sum_{\substack{i=1 \\ i\not\equiv 1 ({\rm mod}\, K)}}^{K(t-1)+\ell}\hspace{-.75em}\beta_{i, j}\,(\tilde\bB_{K(t-1)+\ell})_{K(t-1)+\ell, i}, \,\, j\in [t],\label{eq:betaup}\\
    \gamma_{K(t-1)+1+\ell} &= \tilde\mu_{K(t-1)+\ell} +\sum_{\substack{i=1 \\ i\not\equiv 1 ({\rm mod}\, K)}}^{K(t-1)+\ell}(\tilde\bB_{K(t-1)+\ell})_{K(t-1)+\ell, i}\gamma_{i}.\label{eq:gammaup}
\end{align}
In \eqref{eq:alphaup}, $\delta_{i, j}$ denotes the Kronecker symbol ($\delta_{i, j}=1$ if $i=j$ and $0$ otherwise), and $\alpha_{i, j}$ is assumed to be $0$ if $j\ge i$; in \eqref{eq:betaup}, $\beta_{i, j}$ is assumed to be $0$ if $j> \lceil(i-1)/K\rceil$.

Recall that $\{c_i\}_{i=1}^K$ are the coefficients of the polynomial $\bJ(\bY)$ (in $\bY$), i.e., $\bJ(\bY)=\sum_{i=1}^K c_i \bY^i$. 
Finally, we are ready to express $\mu_{t}$, $\{\theta_{t, j}\}_{j\in [Kt]}$:
\begin{align}
    \mu_{t} &= \sum_{i=1}^K c_i\Big(\tilde\mu_{K(t-1)+i} + \sum_{k=1}^{K(t-1)+i}\gamma_k\,(\tilde\bB_{K(t-1)+i})_{K(t-1)+i, k}\Big),\label{eq:muup}\\
    \theta_{t, j} &= \sum_{i=1}^K c_i\Big(\delta_{K(t-1)+i, j} + \sum_{k=1}^{K(t-1)+i}\alpha_{k, j}\,(\tilde\bB_{K(t-1)+i})_{K(t-1)+i, k}\Big),\quad j\in [Kt]\label{eq:thetaup}.
\end{align}
As before, $\alpha_{i, j}$ is assumed to be $0$ if $j\ge i$. This allows us to define $\tilde U_{Kt+1}$ via \eqref{eq:3ip1} and, after setting $\beta_{Kt+1, t+1}=1$, $\beta_{Kt+1, j}=0$ for all $j\in [t]$, $\alpha_{Kt+1, j}=0$ for all $j\in [Kt+1]$ and $\gamma_{Kt+1}=0$, the definition of the state evolution recursion is complete.

From the state evolution recursion defined above, we can derive the Onsager coefficients $\{{\sf c}_{t, j}\}_{j
\in [t]}$ as
\begin{equation}\label{eq:Ons}
    {\sf c}_{t, j} =\sum_{i=1}^K c_i\sum_{k=1}^{K(t-1)+i}\beta_{k, j}\,(\tilde\bB_{K(t-1)+i})_{K(t-1)+i, k},\quad j\in [t].
\end{equation}

At this point, we are ready to present our result concerning the characterization of the iterates of the AMP algorithm \eqref{eq:AMPnew}, with Onsager coefficients given by \eqref{eq:Ons}, in the high-dimensional limit $N\to\infty$: we prove that the convergence \eqref{eq:SEZF} holds, where $\mu_t$ is given by \eqref{eq:muup} and $W_t=\sum_{j=1}^{Kt}\theta_{t, j}\tilde Z_{j}$, with $\{\theta_{t, j},\tilde Z_j\}_{j\in [Kt]}$ described by the recursion above. Equivalently \cite[Corollary 7.21]{feng2022unifying}, the convergence can be expressed in terms of pseudo-Lipschitz test functions. A function $ \psi\colon\mathbb R^m\to\mathbb R $ is \emph{pseudo-Lipschitz of order $2$}, denoted by $ \psi\in\mathrm{PL}(2) $, if there exists a constant $ C>0 $ such that 
\begin{align}
    \|\psi(\bx) - \psi(\by)\|_2 &\le C\Big(1 + \|\bx\|_2 + \|\by\|_2\Big) \|\bx - \by\|_2, \notag
\end{align}
for all $ \bx,\by\in\mathbb R^m$. 

\begin{theorem}[State evolution of the BAMP]\label{th:SE}
Let $\bY$ be given by \eqref{channel} and which verifies Hypothesis \ref{hyp-density}, and let $\bJ(\bY)=\sum_{i=1}^K c_i \bY^i$. Consider the AMP algorithm \eqref{eq:AMPnew}, with initialization \eqref{eq:AMPinit}, Onsager coefficients $\{{\sf c}_{t, j}\}_{j
\in [t]}$ given by \eqref{eq:Ons} and where, for $t\ge 1$, $g_{t+1}$ is continuously differentiable and Lipschitz. Then, the following limit holds almost surely for any $\mathrm{PL}(2)$ function $\psi: \mathbb R^{2t+2}  \to \mathbb R$, for $t \geq 1$ as $N\to \infty$:
\begin{align}
&  \frac{1}{N} \sum_{i\le N} \psi(u_i^1, \ldots, u_i^{t+1}, f_i^1, \ldots, f_i^t, X^*_i) \to \mathbb \EE\, \psi(U_1, \ldots, U_{t+1}, F_1, \ldots, F_t, X^*) . \label{eq:psiX}
\end{align}
Equivalently, as $N\to\infty$, the joint empirical distribution of $(\bu^1, \ldots, \bu^{t+1}, \bff^1, \ldots, \bff^t, \bX^*)$ converges almost surely in Wasserstein-2 distance to $(U_1, \ldots, U_{t+1}, F_1, \ldots, F_t, X^*)$. Here, for $i\in [t]$, $U_{i+1}=g_{i+1}(F_t)$ and $(F_1, \ldots, F_t)=\bmu_t X^*+(W_1, \ldots, W_t)$, with $W_t=\sum_{j=1}^{Kt}\theta_{t, j}\tilde Z_{j}$ and where $\bmu_t$ can be computed via \eqref{eq:muup}, $\{\theta_{t, j}\}_{j\in [Kt]}$ via \eqref{eq:thetaup} and $\{Z_j\}_{j\in [Kt]}$ is given by \eqref{eq:Zdef}.
\end{theorem}

The proof of Theorem \ref{th:SE} is deferred to Appendix \ref{appsubsec:proof}. A few remarks are now in order. 
First, we highlight that \eqref{eq:psiX} directly implies a high-dimensional characterization of the performance of the AMP \eqref{eq:AMPnew}. In fact, by taking the pseudo-Lipschitz functions $\psi(U_{t+1}, X^*)=(U_{t+1}- X^*)^2$, $\psi(U_{t+1}, X^*)=U_{t+1}\cdot X^*$ and $\psi(U_{t+1}, X^*)=(U_{t+1})^2$, we obtain the limit mean-square error and overlap of the AMP iterates as
\begin{equation}\label{eq:overlapMSE}
    \begin{split}
\lim_{N\to\infty} \frac{1}{2 N^2} \E\| \bX^* (\bX^*)^\intercal-\bu^{t}(\bu^{t})^\intercal \|_{\rm F}^2&=\frac12\big(1-2\Big(\mathbb E[U_{t}\cdot X^*]\Big)^2+(\mathbb E[(U_{t})^2])^2\big),\\
    \lim_{N\to\infty}\frac{|\langle \bX^*, \bu^{t}\rangle|}{\|\bu^{t}\|\cdot \|\bX^*\|} &= \frac{|\mathbb E[U_{t}\cdot X^*]|}{\sqrt{\mathbb E[(U_{t})^2]}}.
\end{split}
    \end{equation}

Next, note that Theorem \ref{th:SE} holds for any family of denoisers $\{g_{t+1}\}_{t\ge 1}$, subject to some mild regularity requirement. A natural choice is to pick the posterior mean
\begin{equation}\label{eq:postBAMP}
    g_{t+1}(f) = \mathbb E[U_*\mid F_t= f].
\end{equation}
Such a choice requires estimating the state evolution parameters $\mu_t$, $\{\theta_{t, j}\}_{j\in [Kt]}$ and $\tilde \bSigma_{Kt}$. These parameters, as well as the Onsager coefficients \eqref{eq:Ons}, can be estimated consistently from the data. To do so, first we obtain $\tilde \bDelta_{Kt}$ and $\tilde \bPhi_{Kt}$ by replacing expectations with empirical averages in \eqref{eq:Deltaup} and \eqref{eq:Phiup}, respectively. Next, we compute $\tilde \bB_{Kt}$ and $\tilde \bSigma_{Kt}$ by plugging in such estimates in \eqref{eq:Bup} and \eqref{eq:tildeBup}, respectively. Having done that, we obtain $\{\alpha_{K(t-1)+1+\ell, j}\}_{j\in [K(t-1)+\ell],\ell\in [K-1]}$, $\{\beta_{K(t-1)+1+\ell, j}\}_{j\in [t],\ell\in [K-1]}$, $\{\gamma_{K(t-1)+1+\ell}\}_{\ell\in [K-1]}$ via \eqref{eq:alphaup}--\eqref{eq:gammaup}. Finally, $\mu_t$, $\{\theta_{t, j}\}_{j\in [Kt]}$ and $\{{\sf c}_{t, j}\}_{j
\in [t]}$ can be computed from \eqref{eq:muup}, \eqref{eq:thetaup} and \eqref{eq:Ons}, respectively. 

As a final remark, we provide an interpretation of the coefficients $\{\alpha_{i, j}\}$, $\{\beta_{i, j}\}$, $\{\gamma_i\}$. As a by-product of the argument proving Theorem \ref{th:SE}, we will show that, for $\ell\in [K-1]$,  (cf. \eqref{eq:auxAMP1}--\eqref{eq:auxAMPnew2})
\begin{align}\label{eq:charac}
    \lim_{N\to\infty}&\frac{\|\bY^\ell\bu^{t}\hspace{-.2em}-\hspace{-.2em}\sum_{j=1}^{K(t-1)+\ell}\hspace{-.2em}\alpha_{K(t-1)+1+\ell, j}\tilde\bz^j\hspace{-.2em}-\hspace{-.2em}\sum_{j=1}^{t}\hspace{-.2em}\beta_{K(t-1)+1+\ell, j}\bu^j\hspace{-.2em}-\hspace{-.2em}\gamma_{K(t-1)+1+\ell}\bX^*\|^2}{N}=0.
\end{align}
This formalizes the fact that $\bY^\ell\bu^{t}$ can be approximately expressed as a linear combination of \emph{(i)} the past iterates $\{\bu^j\}_{j\in [t]}$, \emph{(ii)} the signal $\bX^*$, plus \emph{(iii)} independent Gaussian noise (represented by the $\tilde\bz^j$'s). The quantities $\{\alpha_{i, j}\}$, $\{\beta_{i, j}\}$, $\{\gamma_i\}$ represent the coefficients of this linear combination. The characterization \eqref{eq:charac} allows to subtract from $\bJ(\bY)\bu^k$ just the right Onsager terms, so that this difference equals a component in the direction of the signal (whose size is captured by $\mu_t$) plus independent Gaussian noise (given by the linear combination of the $\tilde\bz^j$'s via the coefficients $\{\theta_{i, j}\}$). 

\subsection{AMP-AP: AMP with Alternating Posteriors}

{As discussed above, the derivation of the Onsager coefficients for BAMP involves approximating vectors of the form $\{\bY^{\ell}\bu^t\}_{\ell\le K-1}$. This fact suggests an alternative choice for the denoisers of the AMP in \cite{fan2022approximate}. For each batch of $K$ iterations, we pick linear denoisers in the first $K-1$ of them, as this allows to construct the vectors $\{\bY^{\ell}\bu^t\}_{\ell\le K-1}$; then, in the $K$-th iteration, we pick the posterior mean using \emph{all} the past iterates, as this -- in principle -- allows to assemble the vectors $\{\bY^{\ell}\bu^t\}_{\ell\le K-1}$ to obtain the quantity $\bJ(\bY)\bu^t$. We refer to this algorithm as AMP with Alternating Posteriors (AMP-AP). In formulas, AMP-AP is given by a slight generalization of \eqref{eq:AMPZF}, where the entry-wise denoiser function depends on all past iterates, as follows:
\begin{equation}\label{eq:AMPAPden}
    \begin{split}
        h_{t+1}(f_i^1, f_i^2, \ldots, f_i^t) &= f_i^t, \qquad \mbox{for } t \not\equiv 1 \,\,\,(\mbox{mod }K),\\
        h_{t+1}(f_i^1, f_i^2, \ldots, f_i^t) &= \mathbb E[X \mid f_i^1, f_i^2, \ldots, f_i^t], \qquad \mbox{for } t \equiv 1 \,\,\, (\mbox{mod }K).
    \end{split}
\end{equation}
We note that AMP-AP does not require the coefficients of the polynomial $\bJ(\bY)$. In fact, it leaves to the posterior mean denoiser the task of learning them from the data. As such, it provides an efficient alternative to our proposed BAMP.}

\section{Numerics}\label{sec:numerics}

For all experiments in this section, random instances of $\bY$ are generated according to the model \eqref{channel}. 
The noise matrices $\bZ=\bO^\intercal \bD\bO$ are generated by first drawing $N$ i.i.d. eigenvalues $(D_i)_{i\le N}$ according to the density \eqref{densityRho}, or \eqref{densityRho_6} (with $\mu=\gamma=0$), and then multiplying from left and right the diagonal matrix of eigenvalues $\bD$ by a random Haar distributed orthogonal matrix $\bO$ sampled independently for each realization. As mentioned at the end of Section \ref{sec:results}, the results are expected to be the same if we were to draw $\bZ$ according to the harder to sample\footnote{This can be done using the Dyson Brownian motion, see \cite{potters2020first}.} measure \eqref{Z-ensemble}. 

\subsection{Spectral properties of the pre-processed matrix $\bJ(\bY)$}

\begin{figure}[t!]
\begin{center}
                \includegraphics[width=0.48\linewidth,trim={0.6cm  0 1cm 0},clip]{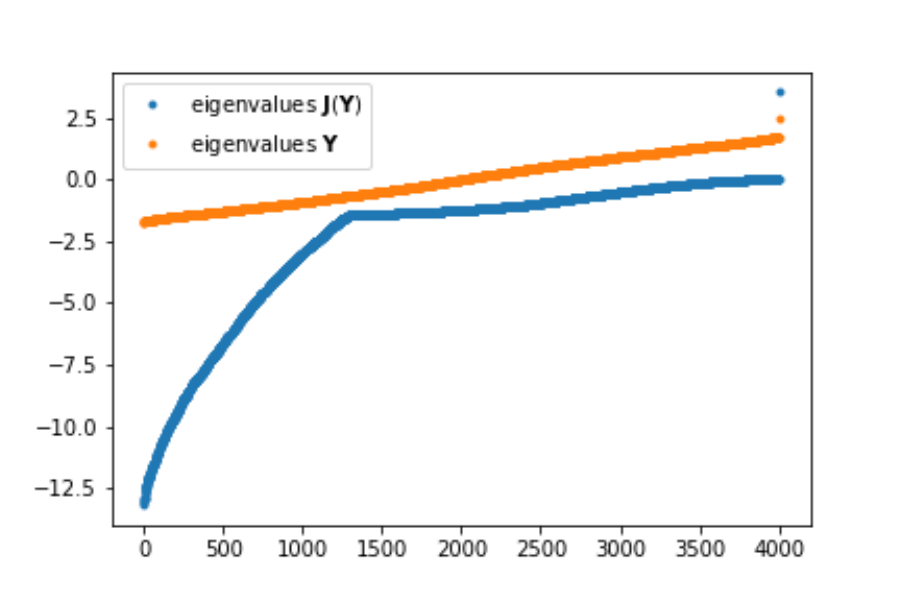}
                \includegraphics[width=0.48\linewidth,trim={0.6cm  0 1cm 0},clip]{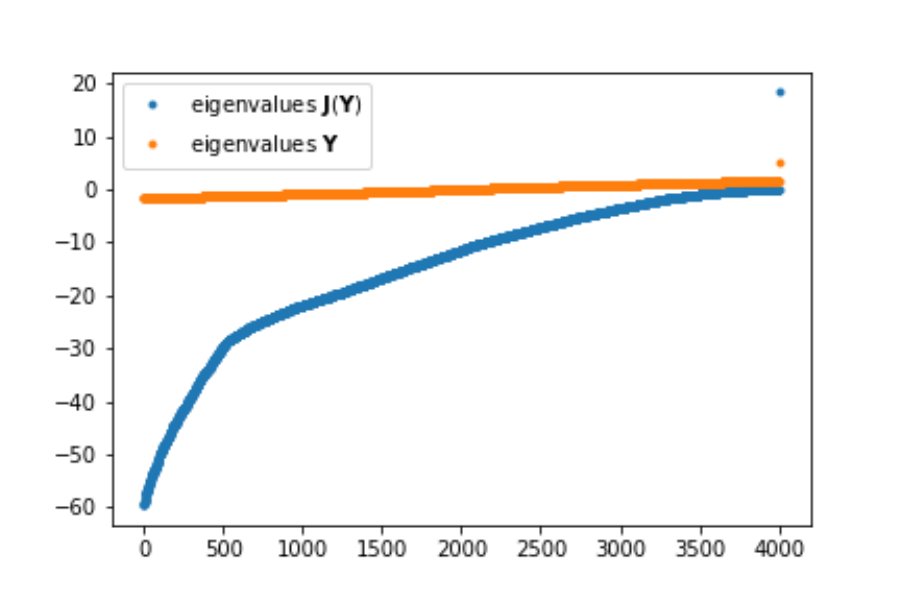}
                \caption{Ranked eigenvalues of the data matrix $\bY$ (orange) and the optimally pre-processed matrix $\bJ(\bY)$ (blue) for $N=4000 \ \mbox{for (left)} \ \lambda=2 \ \mbox{and (right)} \ \lambda=5 $. The gap between the largest detached eigenvalue on the extreme right and the second highest one is much bigger for the pre-processed matrix. Moreover, all the eigenvalues of $\bJ(\bY)$ in its non-informative bulk are negative.}
                \label{fig:spectral}
\end{center}
\end{figure}
                
\begin{figure}[p]
\begin{center}                
                \subfloat[Empirical spectral density of $\bY$. The largest, informative, eigenvalue is emphasized.]{
                \includegraphics[width=0.455\linewidth,trim={0.7cm  0 1cm 1cm},clip]{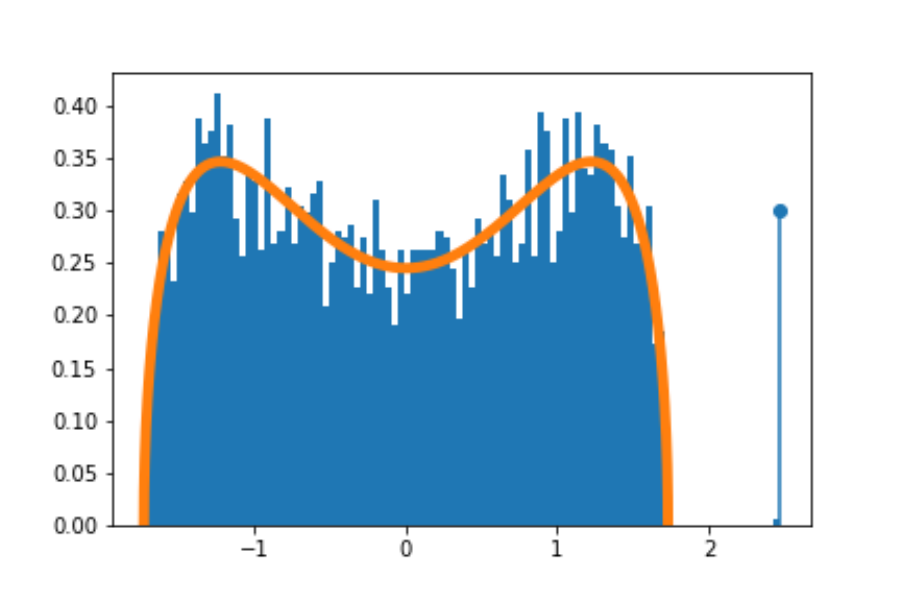}
                \includegraphics[width=0.455\linewidth,trim={0.7cm 0 1cm 1cm},clip]{figures/histo1_2.pdf}
                }\\
                \vspace{-0.1cm}
                                    \subfloat[The function $J(x)=\mu\lambda x-\gamma\lambda^2 x^2+\gamma\lambda x^3$ with $(\mu=0,\gamma(0)=16/27)$ is used to optimally pre-process the (eigenvalues of the) data $\bY$ and obtain $\bJ(\bY)$. The dashed curve indicates $0$. By comparison with the plots (a) above, we understand that the noise bulk will be pushed to negative values, while the spike towards the right, which results in a  ``cleaning'' effect. \label{fig:eigenvals}]{
             \includegraphics[width=0.45\linewidth,trim={1cm  0 1cm 1cm},clip]{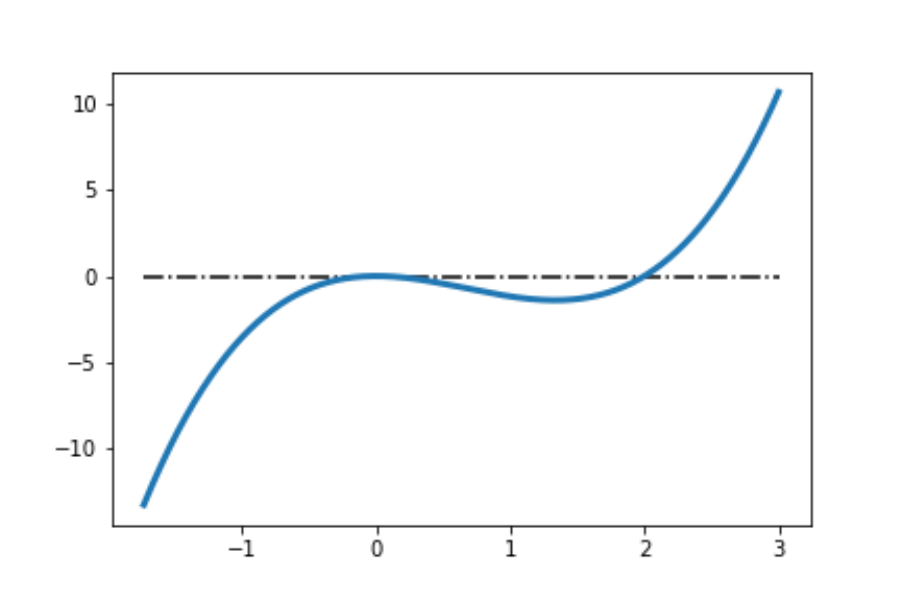}
             \includegraphics[width=0.45\linewidth,trim={1cm  0 1cm 1cm},clip]{figures/preprocess_2.pdf}}\\
             \vspace{-0.1cm}
                 \subfloat[Empirical spectral density of the pre-processed matrix $\bJ(\bY)$. The largest eigenvalue is emphasized and well separated from the negative bulk by the application of $J(x)$.]{
                \includegraphics[width=0.455\linewidth,trim={0.7cm  0 1cm 1cm},clip]{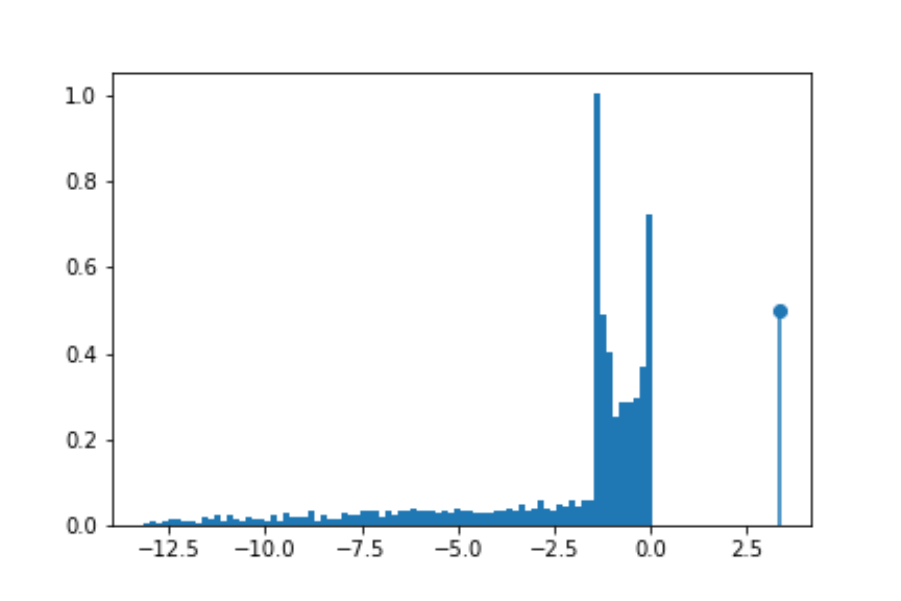}
                \includegraphics[width=0.455\linewidth,trim={0.7cm 0 1cm 1cm},clip]{figures/histo2_2.pdf}}
                    \end{center}
            \caption{Effect of the optimal pre-processing $J(x)$ on the eigenvalues of $\bY$. All experiments are for the most structured noise ensemble  $(\mu=0,\gamma(0)=16/27)$ and $N=4000$. The left column corresponds to $\lambda=2$, while the right column to $\lambda=5$.}
            \label{fig:spectral2}
        \end{figure}
        
        Let us discuss the effect on the spectrum of $\bY$ that has the application of the optimal pre-processing function $J(\cdot)$; clearly, this function does not influence the eigenvectors of $\bY$ which therefore has the same basis as $\bJ(\bY)$. From Figures \ref{fig:spectral} and \ref{fig:spectral2}, the effect is clear: the function $J$ (Figure \ref{fig:spectral2}, middle plots (b)) ``cleans'' the eigenvalues of the data $\bY$ (Figure \ref{fig:spectral2}, upper plots (a)) by shifting the non-informative bulk eigenvalues of $\bY$ to negative values, while the largest, informative, eigenvalue is further separated from the bulk. This results in the histograms (Figure \ref{fig:spectral2}, lower plots (c)) for the processed data $\bJ(\bY)$. It thus becomes much easier to distinguish the informative eigenvalue, which may be of interest for smaller instances where the finite-size effects are stronger.

\subsection{BAMP and AMP-AP improve over the existing AMP and match the replica prediction for the MMSE, and empirical universality of the rotational invariance assumption}\label{sec:7.2}
The plots of Figure \ref{fig:BAMP} consider the quartic ensemble discussed in Section \ref{sec:quartic} for three values of the parameter $\mu$, namely, $\mu\in \{0, 0.5, 1\}$ (recall $\gamma=\gamma(\mu)$ is fixed by relation \eqref{gamma(mu)}), and the power six ensemble \eqref{eq:sestic_potential}. The signal has Rademacher prior, i.e., i.i.d. entries $X_i^*\sim \frac12(\delta_1+\delta_{-1})$. 
The estimators of the spike $\bX^*\bX^{*\intercal}$ are compared in terms of the MSE ($y$-axis) achieved at the fixed point, as a function of the SNR $\lambda$ ($x$-axis). All algorithms are run for $N=8000$ and the results are averaged over $n_{\rm trials}=50$ independent trials; the state evolution recursions (and the replica prediction as well) correspond to $N\to\infty$. We compare the following inference procedures:
\begin{itemize}
    \item In black, we plot the replica prediction \eqref{replicaMMSE}, 
    obtained as the fixed point of \eqref{m_replica_simplified}--\eqref{hatmReplica}.
    
    \item In red, we plot the performance of the BAMP algorithm described in Section \ref{sec:AMP}, where $g_{t+1}$ is the posterior mean denoiser \eqref{eq:postBAMP}. More specifically, the red line corresponds to the fixed point of the MSE given by the state evolution recursion discussed in Section \ref{subsec:OnsSE} (cf. \eqref{eq:overlapMSE}), and the red stars denote the MSE obtained by running the BAMP algorithm \eqref{eq:AMPnew}. 
    
    \item In blue, we plot the performance of the AMP proposed in \cite{fan2022approximate}. More specifically, the blue line corresponds to the fixed point of the MSE \eqref{ZF_FP} obtained by choosing the posterior mean denoiser with a single-step memory term \eqref{eq:ZFden}. The blue diamonds denote the MSE obtained by running the AMP \eqref{eq:AMPZF} with this single-step denoiser.

    \item The ochre squares are MSE values obtained by the AMP of \cite{zhong2021approximate} (without the pre-processing of $\mathbf{Y}$), which employs a full memory posterior mean denoiser: \begin{equation}
        h_{t+1}(f_1, \ldots, f_t) = \mathbb E[X^*\mid (F_1, \ldots, F_t) = (f_1, \ldots, f_t)]\,.
    \end{equation}
    
    \item Finally, the green triangles denote the performance of BAMP when the uniformly distributed matrix $\bO$ (appearing in the spectral decomposition of the noise $\bZ$) is replaced by the product of the Hadamard-Walsh matrix and a diagonal matrix with i.i.d.\ Rademacher entries as in \cite{dudeja2022universality}.

\end{itemize}

\begin{figure}[t!]
            \begin{center}
                    \subfloat[Quartic potential with $\mu=1$.\label{fig:u1}]{
                \includegraphics[width=0.5\linewidth,trim={0 0 1.5cm 1cm},clip]{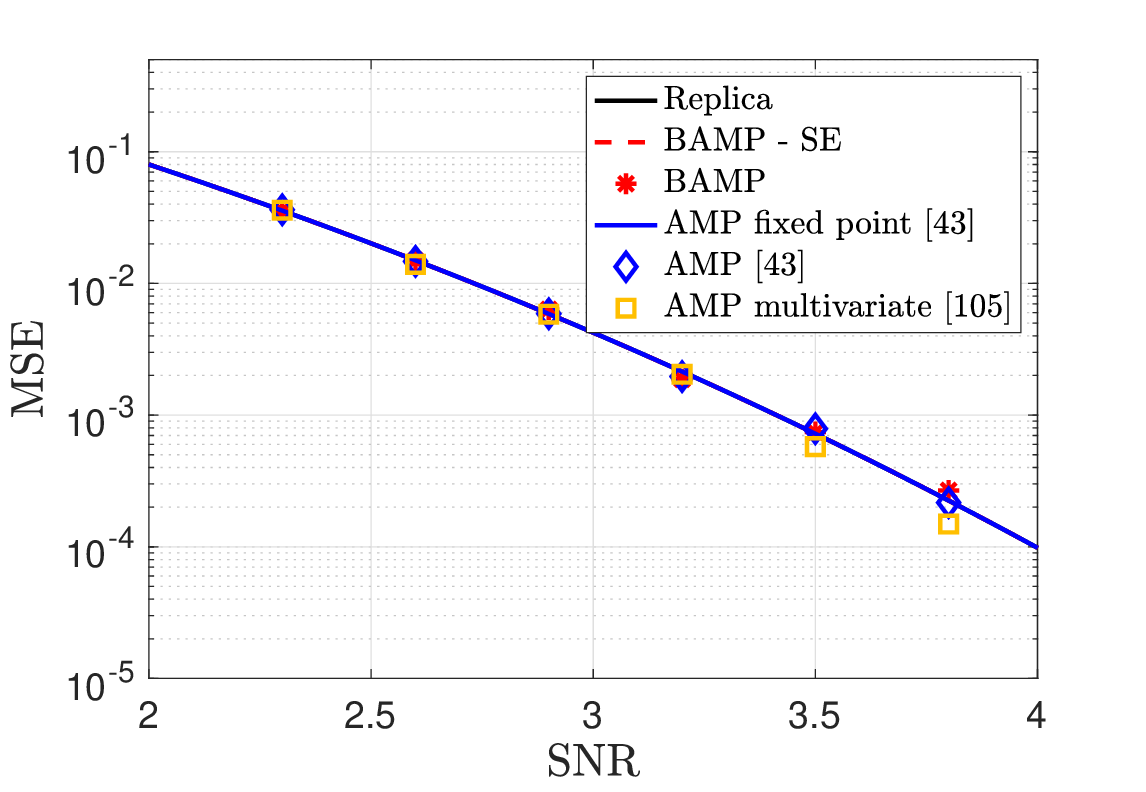}
                }
            \subfloat[Quartic potential with $\mu=0.5$.\label{fig:u0dot5}]{
                \includegraphics[width=0.5\linewidth,trim={0 0 1.5cm 1cm},clip]{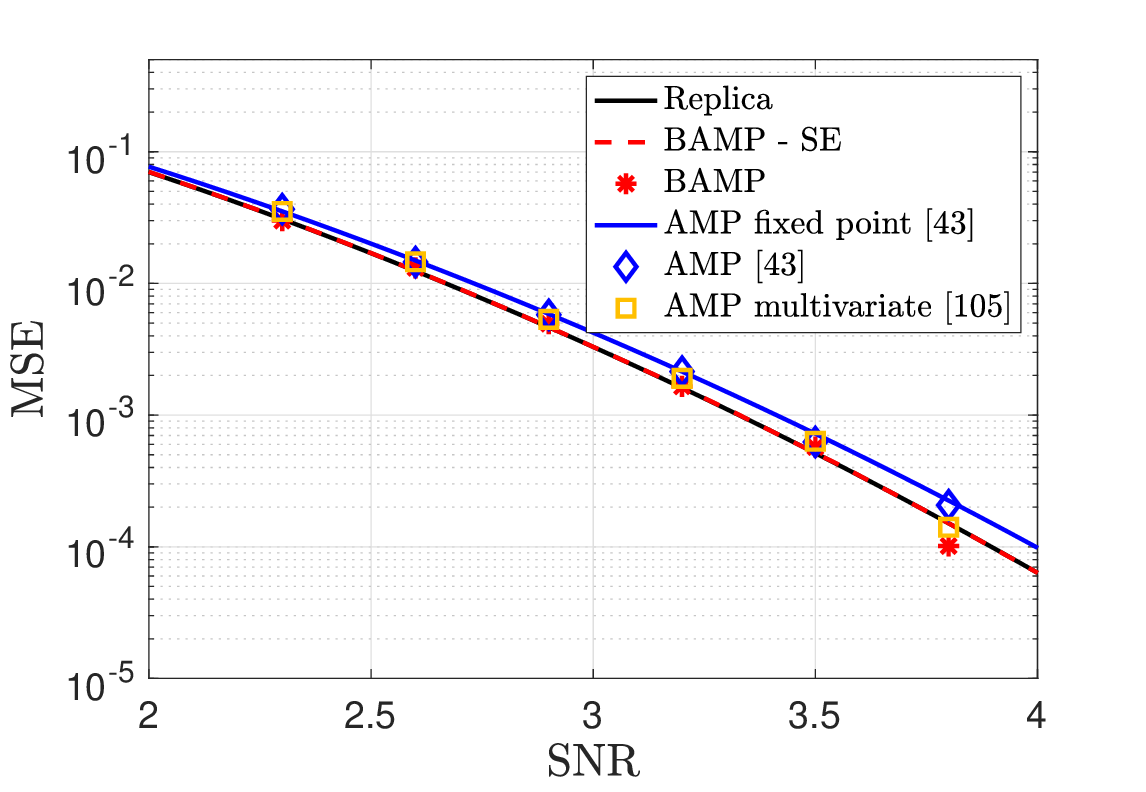}
                }
                \\
            \subfloat[Quartic potential with $\mu=0$.\label{fig:u0}]{
                \includegraphics[width=0.5\linewidth,trim={0 0 1.5cm 1cm},clip]{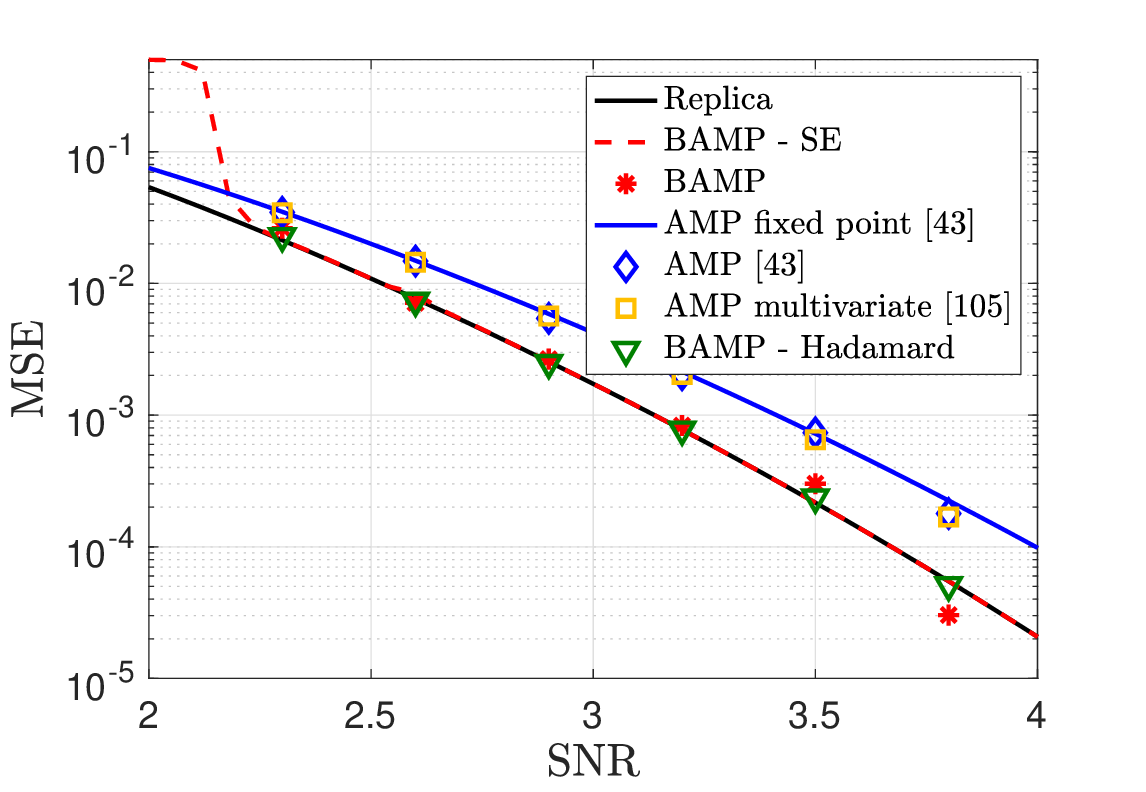}
                }
                \subfloat[Power six potential with $\xi=27/80$.\label{fig:sestic_plot}]{
                \includegraphics[width=0.54\linewidth,clip]{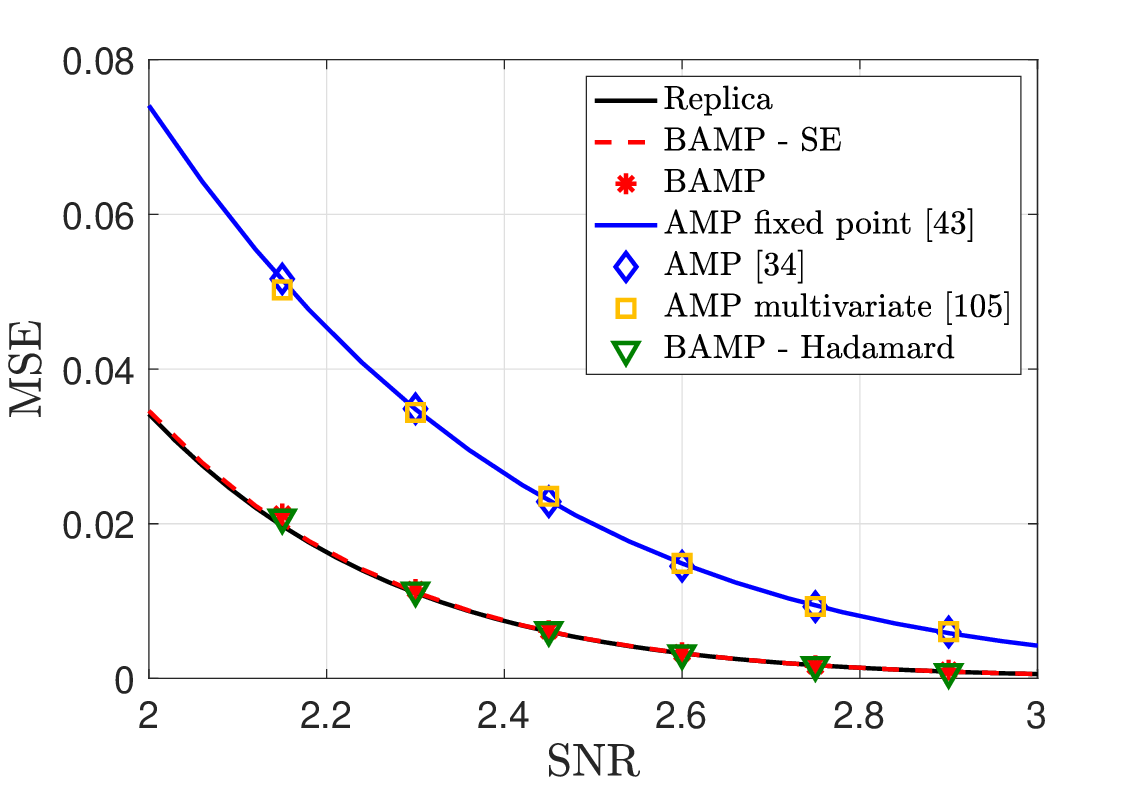}
                }
            \end{center}

            \caption{Performance comparison between the replica prediction for the MMSE (in black), the proposed BAMP (in red), and the existing AMPs \cite{fan2022approximate,zhong2021approximate} (in blue and green). BAMP matches the Bayes-optimal MSE predicted via the replica method, and it outperforms the existing AMP when the noise is not Gaussian. This improvement is more evident as the noise distribution gets further from a Wigner distribution. }\label{fig:BAMP}
            \label{fig:perfo}
        \end{figure}

 We note that all algorithms converge rapidly: $10$ iterations are sufficient to reach the corresponding fixed points. A few remarks concerning the numerical results displayed in Figure \ref{fig:BAMP} are now in order:
 
 \begin{itemize}
     \item {For both the quartic and the sestic potential}, the fixed point of the BAMP state evolution (in red) matches the replica prediction (in black). This is a strong numerical evidence supporting our conjecture that the proposed BAMP algorithm is Bayes-optimal. These theoretical curves for $N\to\infty$ are also remarkably close to the MSE achieved by the BAMP algorithm \eqref{eq:AMPnew} at $N=8000$.
     
     \item When $\mu=0$ in the quartic potential, i.e., the noise is sufficiently far from being independent Gaussian, there is a clear performance gap between our proposed BAMP (in red) and the existing AMP algorithms \cite{fan2022approximate,zhong2021approximate} (single-step denoiser in blue, and multi-step in green). As predicted by our theory, this gap is reduced for $\mu=0.5$, and all curves collapse for $\mu=1$.
   An even greater gap occurs when we consider the power six ensemble in Figure \ref{fig:sestic_plot}, which is ``further'' from the Wigner ensemble. 

     \item Finally, we note that the BAMP algorithm exhibits a numerical instability for low SNR. More specifically, when $\mu=0$ in the quartic potential and $\lambda =2.3$, 5 out of the 50 trials of the iteration in \eqref{BAMP} do not reach the fixed point of state evolution (and are therefore discarded). Furthermore, by inspecting Figure \ref{fig:u0}, one notices that the curve representing the BAMP state evolution detaches from the replica prediction as the SNR gets smaller than $2.3$. As expected, considering an initialization closer to the fixed point mitigates the issue. This numerical instability is likely due to BAMP's state evolution corresponding to the recursion of an auxiliary AMP that \emph{triples} the number of iterations. This leads to an amplification of numerical errors. The same phenomenon occurs with the power six potential, where the number of iterations is multiplied by five. The problem is again mitigated by providing an initialization close to the fixed point. Nevertheless, for $\lambda=2.15$ and $\lambda=2.3$ respectively 2 and 1 BAMP iterations do not reach the fix point of state evolution and are discarded.

 \end{itemize}

Let us re-emphasize that all these results hold in the Bayesian-optimal setting where all hyper-parameters of the model are known and optimally used. In practical situations this may not be the case. In particular the statistical properties of the correlated noise $\bZ$ may be only partially known, preventing one to obtain the coefficients $(c_k)$ defining the optimal pre-processing of the data $\bJ(\bY)=\sum_{k\le K} c_k \bY^k$ as done in Section \ref{sec:AdaTAPtoward}. In Appendix \ref{sec:EM} we provide a learning procedure based on expectation maximization to overcome this issue and which can be of help to practitioners aiming at using BAMP in more realistic situations. Its testing is left for future work.

\begin{figure}[t!]
            \begin{center}
                    \subfloat[Quartic potential with $\mu=0$.\label{fig:uniq}]{
                \includegraphics[width=0.5\linewidth,clip]{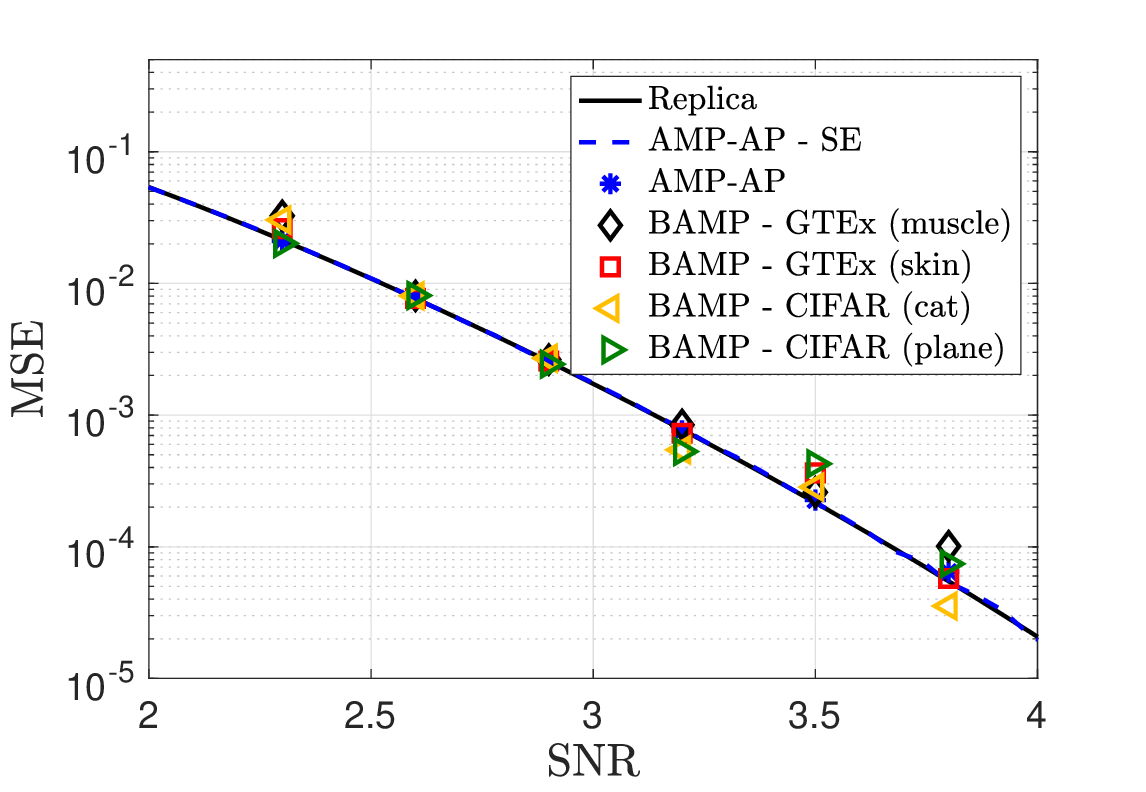}
                }
                \subfloat[Power six potential with $\xi=27/80$.\label{fig:unis}]{
                \includegraphics[width=0.5\linewidth,clip]{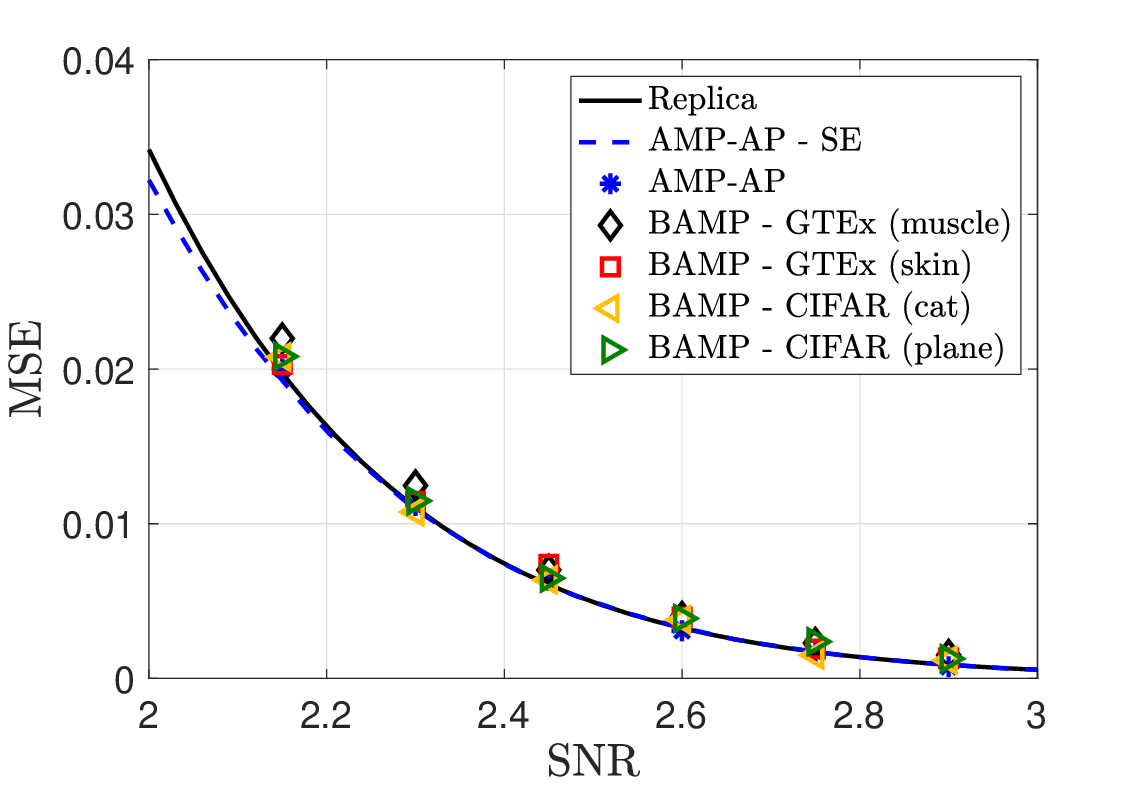}
                }
            \end{center}

            \caption{{Performance comparison between the replica prediction for the MMSE (in black), AMP-AP (in blue), and BAMP run when the noise matrix is not rotationally invariant (black, red, ochre and green symbols). AMP-AP matches the Bayes-optimal MSE predicted via the replica method and, hence, it provides an efficient alternative to BAMP. Furthermore, BAMP displays a remarkable universality behavior, in the sense that its performance is close to the state evolution prediction even when the eigenbasis of the noise is taken from the covariance matrix of datasets commonly employed in practice.} }
            \label{fig:perfouni}
        \end{figure}

{AMP-AP provides an algorithmic alternative that does not require the computation of the coefficients $(c_k)$. Its performance for the quartic potential with $\mu=0$ (left plot) and for the sestic potential (right plot) is represented in  blue in Figure \ref{fig:perfouni}. More specifically, the blue stars denote the MSE obtained by running the algorithm \eqref{eq:AMPZF} with the denoisers given by \eqref{eq:AMPAPden}, and the blue curve is the corresponding state evolution. We remark the excellent agreement with the minimum MSE predicted by the replica formula. In Figure \ref{fig:perfouni}, we also run BAMP when the noise matrix is \emph{not rotationally invariant}, but its eigenbasis comes from the covariance matrix of datasets commonly used in computer vision and quantitative genetics. In particular, we report the results for two CIFAR-10 classes (``plane'' and ``cat''), and two GTEx datasets (``muscle skeletal'' and ``skin sun exposed lower leg'') \cite{lonsdale2013genotype}. For the two CIFAR-10 classes, we have $N=1024$. The two GTEx datasets are matrices of 56200 rows and, respectively, 803 and 701 columns; we pick the first 8000 rows and construct a covariance matrix (hence, $N=8000$). Again, the BAMP performance matches the replica predictions, thus providing an empirical confirmation of the universality of our results.}

\begin{figure}[t!]
            \begin{center}
                \includegraphics[width=0.55\linewidth,clip]{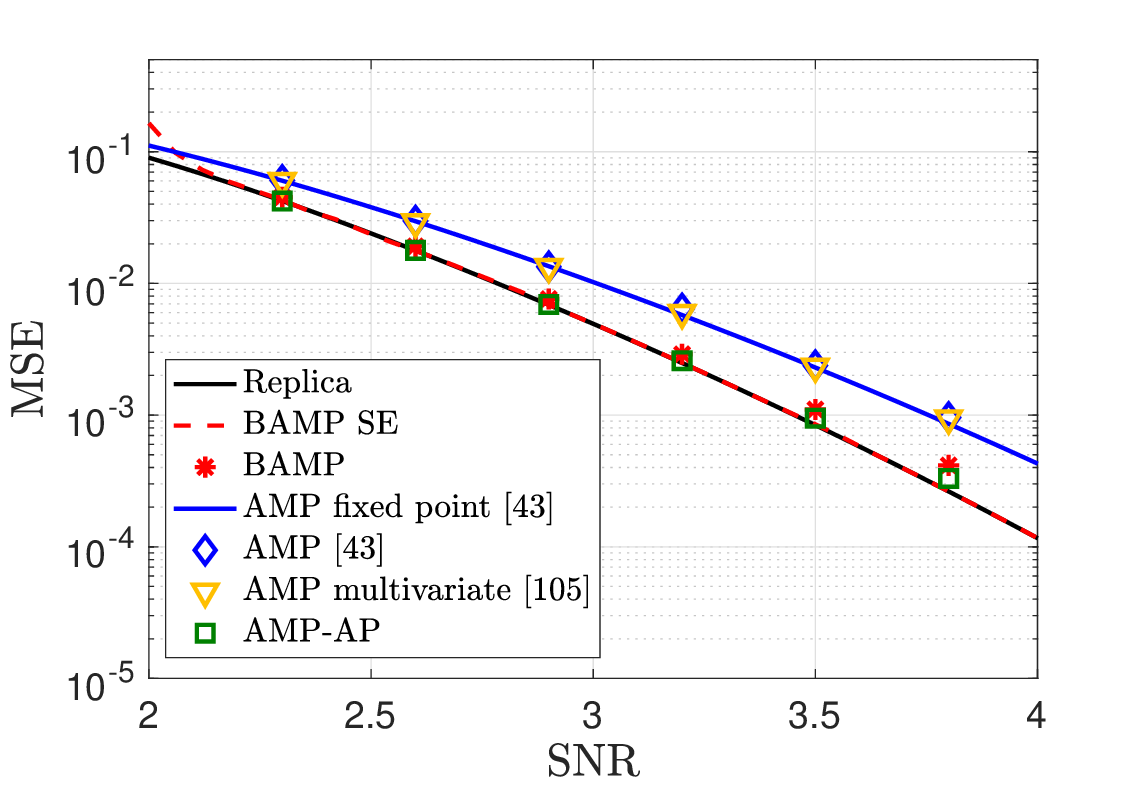}
                \end{center}

            \caption{{Performance comparison between the replica prediction for the MMSE (in black), BAMP (in red), AMP-AP (in green), and the existing AMPs \cite{fan2022approximate,zhong2021approximate} (in blue and ochre). We consider a quartic potential with $\mu=0$ and a signal with a sparse Rademacher prior with $\rho=0.3$. Once again, both BAMP and AMP-AP match the replica prediction and improve upon previously proposed algorithms.} }
            \label{fig:sparse}
        \end{figure}

{Finally, in Figure \ref{fig:sparse}, we consider the quartic potential with $\mu=0$ and a signal having a sparse Rademacher prior, i.e., i.i.d.\ entries $X_i^*\sim (1-\rho)\delta_0+\frac{\rho}{2}\delta_{-1/\sqrt{\rho}}+\frac{\rho}{2}\delta_{1/\sqrt{\rho}}$. We pick $\rho=0.3$. As in the previous cases, BAMP (red) and AMP-AP (green) meet the MMSE predicted via the replica method (black), and they outperform the AMPs previously proposed in \cite{fan2022approximate,zhong2021approximate} (blue and ochre). All algorithms are run for $N=8000$, except the point $\lambda=3.8$ for which we use $N=12000$ in order to improve the convergence to state evolution.}

{Taken all together, the numerical results of Figures \ref{fig:BAMP}-\ref{fig:sparse} provide a clear empirical confirmation of the (Bayes-)optimality of the proposed algorithms, as well as of the universality of BAMP.}


\appendix

\newpage

\section{Approximation of non-polynomial potentials}\label{sec:approx_potential}
In this appendix we will argue that the general strategy presented to study the inference task associated to noise coming from random matrix ensembles with polynomial potentials can be used to approximate the MMSE of noise ensembles with general analytic potentials by considering a proper sequence of polynomials that converges point-wise. In the following argument we will assume that:
\begin{itemize}
    \item[(i)] the potential $V:\R\mapsto\R$ is analytic,
    \item[(ii)] there is a constant $C > 0$ such that, for all $x\in\R$, we have that $V \geq C x^2/2$,
    \item[(iii)] and the coordinates of $\mathbf{X}^*$ are i.i.d. of density $P_X$ with bounded support.
\end{itemize}
Although condition (iii) can be weakened at the cost of some extra technicalities, here we will include it to keep the presentation more simple.

By condition (i), the potentials considered are analytic. Then, there is some sequence $(c_k)_{k\geq1}$ such that, for all $x\in\R$, 
\begin{equation*}
    V(x) = \sum_{k\geq1} c_k x^k.
\end{equation*}
Let $p\geq1$ and define $V_p:\R\mapsto\R$ according to $V_p(x) := \sum_{k\leq p} c_k x^k$. In this way, we will define $\mathbf{Z}_p\in\R^{N\times N}$ to be a random matrix of probability distribution
\begin{align}\label{Zp-ensemble}
    dP_{Z_p}(\mathbf{Z}_p)= C_{V_p} \exp\Big(-\frac{N}{2}\Tr V_p(\mathbf{Z}_p)\Big)\prod_{i\leq j}dZ_{p,ij}\,;
\end{align}
where, as before, $C_{V_p}>0$ is just a normalizing constant. Also define a new data matrix $\mathbf{Y}_p\in\R^{N\times N}$ according to $\mathbf{Y}_p := \frac{\lambda}{N} \mathbf{P}^* + \mathbf{Z}_p$. Here we will introduce the posterior measures
\begin{align}
    dP^{(p)}_{X\mid Y_p}(\mathbf{x}\mid \bY_p)\!=\!\frac{dP_X(\mathbf{x})}{Z_p(\mathbf{Y}_p)}\exp\Big(\!-\frac{N}{2}\Tr V_p\Big(\bY_p-\frac{\lambda}{N}\mathbf{P}\Big)\!\Big)\label{posterior-p}
\end{align}
and
\begin{align}
    d{P'}^{(p)}_{X\mid Y}(\mathbf{x}\mid \bY)\!=\!\frac{dP_X(\mathbf{x})}{Z'_p(\mathbf{Y})}\exp\Big(\!-\frac{N}{2}\Tr V_p\Big(\bY-\frac{\lambda}{N}\mathbf{P}\Big)\!\Big)\label{posterior-p2};
\end{align}
with $Z_p(\mathbf{Y}_p),Z'_p(\mathbf{Y}) >0$ normalization constants. Notice that $P^{(p)}_{X\mid Y_p}(\cdot)$ corresponds to the Bayes-optimal posterior of data with noise $\mathbf{Z}_p$ and ${P'}^{(p)}_{X\mid Y}(\cdot)$ is the mismatched posterior obtained from a signal generated with noise from an ensemble of potential $V$ but wrongly modeled as having noise from potential $V_p$. The free entropies associated with these posteriors will then be $F_N^{(p)}:=\E\ln Z_p(\mathbf{Y}_p)$ and ${F'}_N^{(p)}:=\E\ln Z'_p(\mathbf{Y})$. And finally, the associated mutual information between data and signal for the first of the two posteriors, which is Bayes-optimal, will be given by
\begin{equation*}
    I_p(\bP^*; \bY_p):=-F^{(p)}_N(\bY_p) - \frac{N}{2}\mathbb{E}\Tr  V_p(\mathbf{Z}_p).
\end{equation*}

In this appendix we will argue that
\begin{equation}\label{eq:equiv_mutual_inf}
    \lim_{p\to\infty} \lim_{N\to\infty} \frac{1}{N}|I(\bP^*; \bY)-I_p(\bP^*; \bY_p)| = 0.
\end{equation}
By including side information of the form $\mathbf{\tilde Y} = \tilde \lambda \mathbf{X}_* + \mathbf{\tilde Z}$ with $\tilde \lambda >0$ and $\mathbf{\tilde Z} \in \R^N$ a standard Gaussian vector, the magnetization $m$ of both models can be obtained as a derivative with respect to $\tilde \lambda$ of each asymptotic mutual information. The strategy to derive the free entropy limit in the main text can be easily adapted to include this side. Furthermore, if $\tilde \lambda$ is taken to be small (i.e., the side information has a low signal-to-noise ratio in some proper sense), the asymptotic value of $m$ is not modified by this side information. See for example \cite[Section 5.1.1]{barbier2018optimal}, for more details on this strategy. Finally, because the free entropies are convex functions of $\tilde \lambda$, equation \eqref{eq:equiv_mutual_inf} implies that the asymptotic values of the magnetization $m$ of both models coincide when $p$ goes to infinity whenever the signal-to-noise ratio is not taking a critical value. This then means that the MMSE of the model with noise of potential $V$ differs with respect to the one of the model with noise of potential $V_p$ by a term that is vanishing in $p$. This therefore justifies the fact of studying only models with polynomial potentials.

In the rest of the section we will justify \eqref{eq:equiv_mutual_inf}. To see that this should hold, we will first bound
\begin{equation}\label{eq:2_fe_bounds}
    \begin{split}
        \frac{1}{N}|I(\bP^*; \bY)-I_p(\bP^*; \bY_p)| & \leq \frac{1}{N}\Big|F_N-{F'}^{(p)}_N+\frac{N}{2}\Tr(V(\bZ)-V_p(\bZ))\Big| \\
        & \,\,\,\,\,\,\,\,\,\,\,\,\,\,\,\,\,\,\,\, + \frac{1}{N} \Big|I_p(\bP^*; \bY_p)+{F'}^{(p)}_N + \frac{N}{2} \mathbb{E}\Tr  V_p(\mathbf{Z}_p)\Big|.
    \end{split}
\end{equation}

For bounding the first term on the right of \eqref{eq:2_fe_bounds} we will introduce, for every $t\in[0,1]$, the interpolating measure of mean $\langle\cdot\rangle_t$ corresponding to the Hamiltonian 
\begin{equation*}
    H_{N,t}(\mathbf{X}) = -\frac{N}{2} \Tr\, V_t\Big(\mathbf{Y}-\frac{\lambda}{N}\mathbf{P}\Big);
\end{equation*}
where $V_t(x):= V_p(x)+tE_p(x)$ and $E_p(x) := \sum_{k\geq p+1} c_k x^k$. Let $F_{N,t}$ be the free entropy associated to $H_{N,t}$ and define $G_t := -F_{N,t}-N/2 \Tr V_t(\bZ)$. Clearly, we have that $G_0 = -{F'}^{(p)}_N-N/2 \Tr V_p(\bZ)$ and $G_1=-F_N-N/2 \Tr V(\bZ)$. Notice that for all $t\in[0,1]$,
\begin{equation*}
    \frac{1}{N}\frac{d G_t}{dt} = \, \frac{1}{2}\E\Big\langle \Tr \, E_p\big(\mathbf{Y}-\frac{\lambda}{N}\mathbf{P}\big) \Big\rangle_t - \frac{1}{2}\E\Tr \, E_p\big(\bZ\big) . 
\end{equation*}
We would now like to see that the absolute value of the right hand side of the last equation is $o_p(1)$. For this, denote by $D_1,\dots D_N$ the eigenvalues of $\bZ$ ordered from largest to smallest. Likewise, denote by $\tilde D_1,\dots,\tilde D_N$ the ones of $\bY - \lambda/N \bP$. By Weyl's interlacing inequalities we have that
\begin{equation*}
    D_i \leq \tilde D_i \leq D_{i-2} \,\,\,\,\, \mbox{ for all } i=3,\dots,N.
\end{equation*}
This means that, if we denote (for $i=2,\dots,N$) $\delta_i:=D_i-D_{i-1}$, we then have that, for all $i=3,\dots,N$,
\begin{equation*}
    |E_p(\tilde D_i) - E_p(D_i)| \leq  (\delta_{i-1}+\delta_i) E'_p(\xi_i),
\end{equation*}
for some $D_i\leq\xi_i\leq D_{i-2}$. By condition (ii) above, we know that the limiting distribution $\rho$ of the eigenvalues of $\mathbf{Z}$ and the distribution $\rho_p$ of the ones of $\mathbf{Z}_p$ are, for large enough $p$, both contained in the compact interval $[-2/C,2/C]$. This is so because condition (ii) implies that the potentials $V$ and $V_p$ are more confining than $C x^2/2$ (see \cite[Section 5.2]{potters2020first} for more details). Thus, under $\langle\cdot\rangle_t$, by this and condition (iii) the eigenvalues of $\mathbf{Y}-\lambda/N\mathbf{P}$ and $\bZ$ are contained in the interval $[-(1/C+2\lambda),1/C+2\lambda]$. On $[-(1/C+2\lambda),1/C+2\lambda]$ we will have $|E'_p|\leq \epsilon_p$ for some vanishing sequence $(\epsilon_p)_{p\geq1}$. From which we get that, for all $i=3,\dots,N$,
\begin{equation*}
    |E_p(\tilde D_i) - E_p(D_i)| \leq  (\delta_{i-1}+\delta_i)\epsilon_p
\end{equation*}
Moreover, on the interval $[-(1/C+2\lambda),1/C+2\lambda]$, we have that there is another vanishing sequence $(\epsilon'_p)_{p\geq1}$ such that $|E_p|\leq\epsilon'_p$. From this we have
\begin{equation*}
    \begin{split}
        \Big|\Tr \Big[E_p\big(\mathbf{Y}-\frac{\lambda}{N}\mathbf{P}\big)  -  E_p\big(\bZ\big)\Big]\Big| & \leq |E_p(\tilde D_1)|+|E_p(\tilde D_2)|+|E_p(D_1)|\\
        & \,\,\,\,\,\,\,\,\,\, +|E_p(D_2)|+\sum_{i=3}^N|E_p(\tilde D_i)-E_p(D_i)|\\
        & \leq 4 \epsilon'_p + \epsilon_p \sum_{i=3}^N (\delta_{i-1}+\delta_i) \leq  4 \epsilon'_p + \frac{4\epsilon_p}{C}\xrightarrow{p\to\infty}0.
    \end{split}
\end{equation*}
From this we then conclude that $N^{-1}|G_1-G_0| = o_p(1)$ which means that the first term on the right hand side of \eqref{eq:2_fe_bounds} is vanishing in $p$.

For the second term on the right hand side of \eqref{eq:2_fe_bounds} we will draw some of the conclusions from \cite[Theorem 2.6.1]{anderson2010introduction}. If we define a functional $\Sigma(\cdot)$ over the probability distributions on the line according to $\Sigma(\mu) = \int\int \ln|x-y|d\mu(x)d\mu(y)$ if $\int \ln(|x|+1)d\mu(x)<\infty$ and $\Sigma(\mu)=-\infty$ otherwise, as a consequence of the theorem we have that the empirical eigenvalue measure of $\mathbf{Z}$ obeys a large deviation principle of speed $N^2$ and good rate function
\begin{equation*}
    I_V(\mu) :=
    \begin{cases}
        \int V(x) d\mu(x) - \frac{1}{2}\Sigma(\mu) - c_V & \mbox{if } \int V(x) d\mu(x) < \infty \\
        \infty & \mbox{o.w.};
    \end{cases}
\end{equation*}
where $c_V := \inf_\mu \int V(x) d\mu(x) - 1/2\Sigma(\mu)$. Similarly, the empirical eigenvalue measure of $\mathbf{Z}_p$ obeys a large deviation principle of speed $N^2$ and good rate function
\begin{equation*}
    I_{V_p}(\mu) :=
    \begin{cases}
        \int V_p(x) d\mu(x) - \frac{1}{2}\Sigma(\mu) - c_{V_p} & \mbox{if } \int V_p(x) d\mu(x) < \infty \\
        \infty & \mbox{o.w.};
    \end{cases}
\end{equation*}
with $c_{V_p} := \inf_\mu \int V_p(x) d\mu(x) - 1/2\Sigma(\mu)$. By \cite[Lemma 2.6.2]{anderson2010introduction}, both $I_V$ and $I_{V_p}$ are strictly convex. By condition (ii), for finding the minimum of the rate functions, we can restrict the optimization problem to densities supported on the interval $[-1/C,1/C]$. Then, because $V_p\xrightarrow{p\to\infty}V$, we have that the minimizer of $I_{V_p}$ has to converge to that of $I_V$. We then have that $\rho_p$ converges point-wise to $\rho$. Finally, notice that, for general polynomial potentials, the function $f_\rho$ defining the optimization problem that gives the limiting free entropy should be continuous with respect to $\rho$. Therefore, when the minimizer of $f_\rho$ is unique, we should then have that the minimizer of $f_{\rho_p}$ should approach it when $p$ goes to infinity. Here we implicitly assumed that, if $p$ is sufficiently large, the posterior ${P'}^{(p)}_{X\mid Y}(\mathbf{x}\mid \bY)$ is replica symmetric. This then means that $\mbox{extr}_\tau f_{\rho_p}(\tau) \approx \mbox{extr}_\tau f_{\rho}(\tau)$. We then have that $N^{-1}|F^{(p)}_N-{F'}^{(p)}_N|$ is $o_p(1)$. Finally, $N^{-1}|\mathbb{E}\Tr  V_p(\mathbf{Z}_p)-\mathbb{E}\Tr  V_p(\mathbf{Z})|$ is also $o_p(1)$ because of the convergence of $\rho_p$ towards $\rho$. This means have that the second term in \eqref{eq:2_fe_bounds} is $o_p(1)$ from which we conclude \eqref{eq:equiv_mutual_inf}.

\section{Learning the optimal pre-processing $\bJ(\bY)$}\label{sec:EM}
 
Until now we have assumed that we are in the Bayesian-optimal setting where, in particular, the polynomial potential $V$ defining the noise statistics is completely known and correctly exploited. As seen from section~\ref{statCondADATAP}, given a potential $V$ we could deduce from the AdaTAP formalism an optimal polynomial $$\bJ=\bJ(\bY)=\sum_{k\le K} c_k \bY^k$$ to pre-process the data $\bY$ before using it in AMP. The Bayes-optimal case corresponds to matrix \eqref{optimalJ}, i.e., $\bJ=c_1 \bY+c_2 \bY^2+c_3 \bY^3$ with $\bc=(\mu\lambda,-\gamma\lambda^2 ,\gamma\lambda)$.

We here consider an extension of the previously derived AMP to a case where $V$ is not known and therefore the optimal $\bJ$ cannot be deduced by the AdaTAP approach as we did in Section \ref{statCondADATAP}. What is known instead is an upper bound on the order of $V$. In the base-case model studied in details in the present paper the order is four. The procedure we propose below will not be tested numerically yet, but we believe it may be of interest to practitioners eager to improve the Bayes-optimal AMP for more practical settings than the specific ones studied here.

To directly learn the coefficients $(c_k)_{k\le K}$ from the data, we propose to use an approach inspired by the expectation maximization (EM) algorithm, with a routine inside AMP performing the parameter estimation by maximizing the current estimate of the free entropy, i.e., of the log-likelihood of the observed data $\ln P(\bY\mid \bc )$. 

Assume that, at the AMP iterate $t$, the current estimate of the unknown coefficients $\bc=(c_k)_{k\le K}$ is $\bc(t)=(c_k(t))_{k\le K}$, the AMP estimate of the marginal means is $\bmm(t)$, and of the Onsager reaction term is $\bar V(t)$ (which is related to the set of Onsager coefficients, see Section \ref{subsec:OnsSE}). Let also the data matrix polynomial currently used by AMP be
$$\bJ(t):=\sum_{k\le K} c_k(t) \bY^k.$$

From the analysis of Section \ref{sec:selfaveraging} we know that at the saddle point we can safely replace the Onsager reaction term $V_i$ by $\bar V$ in the AdaTAP equations. When this is plugged back into \eqref{usefulId}, this identity implies that also the following concentration is consistently valid: $\EE(\tau_i-m_i^2)=\tau_i-m_i^2$, which is also equal by exchangeability to $N^{-1}\sum_{i\le N}\EE(\tau_i-m_i^2)$. Let us call $\bar \chi(t)$ the AMP estimate of the variance $\EE(\tau_i-m_i^2)$. Applying these simplifications to the AMP iterates we get that the matrix $\boldsymbol\Omega(t):={\rm diag}(\bV(t)+(\btau(t)-\bmm(t)^2)^{-1})$ can be simplified as  $$\boldsymbol\Omega(t)=(\bar V(t)+\bar \chi(t)^{-1})I_N.$$ 
From section~\ref{sec:adatapFreeEn} the AdaTAP approximation to the free entropy at iterate $t$ then reads, using these simplifications, as
\begin{align}
\Phi_N(t,\bc(t))&= \frac12\bmm(t)^\intercal \bJ(t)\bmm(t)+\frac12\ln \det\big(\boldsymbol\Omega(t)-\bJ(t)\big)-\frac12 \bar V(t) \sum_{i\le N}m_i(t)^2+\frac12\bar \chi(t)\nonumber\\
&\hspace{-1cm}-\sum_{i\le N}\ln\int dP_X(x)\exp\Big(\frac12\bar V(t)x^2+\big((\bJ(t)\bmm(t))_i -\bar V(t)m_i(t)\big) x\Big).\label{adatap_f(t)}
\end{align}
The free entropy $\Phi_N(t,\bc(t))$ is the current best approximation to the marginal log-likelihood of the data $\ln P(\bY\mid \bc )$, which we thus aim at maximizing with respect to the unknown parameters, all other quantities being fixed at their current values:
\begin{align}
    \partial_{c_k} \Phi_N|_{t,\bc(t)}&=\bmm(t)^\intercal\bY^k\Big(\frac12\bmm(t)-\eta(\bJ(t),\bmm(t),\bar V(t))\Big)\nonumber \\
    &\qquad-\frac12\Tr\big(\bY^k(\boldsymbol\Omega(t) - \bJ(t))^{-1}\big),
\end{align}
where we used \eqref{TAPm} and the notation $\eta(\bJ(t),\bmm(t),\bar V(t))= (\eta_i(\bJ(t),\bmm(t),\bar V(t))_{i\le N}$. Because $\bJ$ is diagonalizable in the same basis as the data $\bY$, the eigenvalues of which are denoted $\sigma_i=\sigma_i(\bY)$, we have
\begin{align}
  \Tr\big(\bY^k(\boldsymbol\Omega(t) - \bJ(t))^{-1}\big)  =\sum_{i\le N} \frac{\sigma_i^k}{\bar V(t)+\bar \chi(t)^{-1}-\sum_{\ell\le K} c_\ell(t) \sigma_i^\ell}.
\end{align}
Then
\begin{align}
    \partial_{c_k} \Phi_N|_{t,\bc(t)}&=\bmm(t)^\intercal\bY^k\Big(\frac12\bmm(t)-\eta(\bJ(t),\bmm(t),\bar V(t))\Big)\nonumber \\
    &\qquad-\frac12\sum_{i\le N} \frac{\sigma_i^k}{\bar V(t)+\bar \chi(t)^{-1}-\sum_{\ell\le K} c_\ell(t) \sigma_i^\ell}.
\end{align}
We aim at maximizing the free entropy so given a learning rate $\zeta>0$ the learning rule finally reads
\begin{align}
    c_k(t+1)=c_k(t)+\zeta \partial_{c_k} \Phi_N|_{t,\bc(t)}.
\end{align}

\section{Proofs for BAMP}\label{appsec:BoptAMP}

\subsection{Auxiliary AMP}\label{appsubsec:aux}

The iterates of the auxiliary AMP are denoted by $\tilde{\bz}^t, \tilde{\bu}^t\in\mathbb R^N$, and they are computed as follows, for $t\ge 1$:
\begin{equation}\label{eq:auxAMP}
    \tilde{\bz}^t = \bZ \tilde{\bu}^t - \sum_{i=1}^t \bar {\sf b}_{t, i}\tilde{\bu}^i, \quad \tilde{\bu}^{t+1} = \tilde{h}_{t+1}(\tilde{\bz}^1, \ldots, \tilde{\bz}^t, \bu^1, \bX^*).
\end{equation}
The iteration \eqref{eq:auxAMP} is initialized with $\tilde{\bu}^1=\bu^1$, where $\bu^1$ satisfies \eqref{eq:AMPinit}. For $t\ge 1$, the functions $\tilde{h}_{t+1}:\mathbb R^{t+2}\to\mathbb R$ are applied component-wise, and they are recursively defined as
\begin{equation}\label{eq:deftildeh}
    \begin{split}
        \tilde{h}_{K(t-1)+1+\ell}(&z_1, \ldots, z_{K(t-1)+\ell}, u_1, x^*) = z_{K(t-1)+\ell}+(\tilde\bB_{K(t-1)+\ell})_{K(t-1)+\ell, 1}\, u_1 \\
        &\hspace{-4em}+\hspace{-.5em} \sum_{i=2}^{K(t-1)+\ell}\hspace{-.5em} (\tilde\bB_{K(t-1)+\ell})_{K(t-1)+\ell, i}\,\tilde{h}_{i}\Big(z_1, \ldots, z_{i-1}, u_1, x^*\Big)+\tilde\mu_{K(t-1)+\ell} x^*, \quad \ell\in [K-1],\\
        \tilde{h}_{Kt+1}(&z_1, \ldots, z_{Kt}, u_1, x^*) = g_{t+1}\Big(\mu_t x^*+\sum_{i=1}^{Kt}\theta_{t, i}z_i\Big).
    \end{split}
\end{equation}
The idea is that the choice \eqref{eq:deftildeh} for the denoisers $\{\tilde{h}_{t+1}\}_{t\ge 1}$ ensures that $\tilde \bu^{K(t-1)+\ell}$ tracks the quantity $\bY^{\ell-1}\bu^t$ for $\ell\in [K]$ and $t\ge 1$, where $\{\bu^t\}$ are the iterates of the AMP  iteration \eqref{eq:AMPnew} we are interested in analyzing.

In \eqref{eq:deftildeh}, $g_{t+1}$ is the denoiser of the AMP \eqref{eq:AMPnew}. The parameters $(\tilde \bB_{K(t-1)+\ell}$, $\tilde\mu_{K(t-1)+\ell}$, $\mu_t, \theta_{t, i})$ come from the state evolution recursion detailed in Section \ref{subsec:OnsSE}: $\tilde \bB_{K(t-1)+\ell}$ is given by \eqref{eq:Bup}, $\tilde\mu_{K(t-1)+\ell}$ by \eqref{eq:tildemuup}, $\mu_t$ by \eqref{eq:muup} and $\theta_{t, i}$ by \eqref{eq:thetaup}. We now discuss how to obtain the coefficients $\{\bar{\sf b}_{t, i}\}_{i=1}^t$ needed in \eqref{eq:auxAMP}. Let us define the matrix $\bar\bPhi_{t}\in\mathbb R^{t\times t}$ as
\begin{equation}\label{eq:defbPhi}
        (\bar\bPhi_{t})_{i, j}=0, \quad \mbox{ for }i\le j, \qquad (\bar\bPhi_{t})_{i, j}=\langle\partial_j \tilde{\bu}^i\rangle, \quad \mbox{ for }i> j,
\end{equation}
where, for $j<i$, the vector $\langle\partial_j \tilde{\bu}^i\rangle\in \mathbb R^N$ denotes the partial derivative of $\tilde{h}_i:\mathbb R^{i+1}\to \mathbb R$ with respect to the $j$-th input (applied component-wise). Then, the vector $(\bar{\sf b}_{t, 1}, \ldots, \bar{\sf b}_{t, t})$ is given by the last row of the matrix $\bar{\bB}_{t}\in\mathbb R^{t\times t}$ defined as
\begin{equation}\label{eq:deftildeB}
    \bar{\bB}_{t} = \sum_{j=0}^{t-1} \kappa_{j+1}\bar\bPhi_{t}^j.
\end{equation}
where $\{\kappa_k\}_{k\ge 1}$ denotes the sequence of free cumulants associated to the matrix $\bZ$.

\subsection{State evolution of auxiliary AMP}\label{appsubsec:auxAMPSE}
Using Theorem 2.3 in \cite{zhong2021approximate}, we provide a state evolution result for the auxiliary AMP \eqref{eq:auxAMP}. In particular, we show in Proposition \ref{prop:auxSE} that the joint empirical distribution of $(\tilde{\bz}^1, \ldots, \tilde{\bz}^t)$ converges to a $t$-dimensional Gaussian $\mathcal N(\bzero, \hat{\bSigma}_t)$. 

The covariance matrices $\{\hat{\bSigma}_t\}_{t\ge 1}$ are defined recursively, starting with $\hat{\bSigma}_1=\bar{\kappa}_2\mathbb E[U_1^2]$, where $U_1$ is defined in \eqref{eq:AMPinit}. Given $\hat{\bSigma}_t$, let 
\begin{equation}\label{eq:SElawsaux}
    \begin{split}
        (\hat Z_1&, \ldots, \hat Z_t)\sim\mathcal N(\bzero, \hat{\bSigma}_t) \mbox{ and independent of }(X^*, U_1),\\
         \hat U_{s} &= \tilde{h}_{s}\Big(\hat Z_1, \ldots, \hat Z_{s-1}, U_{1}, X^*\Big), \quad \mbox{ for } s\in \{2, \ldots, t+1\},
    \end{split}
\end{equation}
where $\tilde{h}_s$ is defined via \eqref{eq:deftildeh} and we set $\hat U_{1}=U_1$. Let $\hat{\bPhi}_{t+1}, \hat{\bDelta}_{t+1}\in\mathbb R^{(t+1)\times (t+1)}$ be matrices with entries given by  
\begin{equation}
    \begin{split}
        (\hat{\bPhi}_{t+1})_{i, j}&=0, \quad \mbox{ for }i\le j,\qquad\qquad (\hat{\bPhi}_{t+1})_{i, j}=\mathbb E[\partial_j \hat U_i], \quad \mbox{ for }i> j,\\
        (\tilde{\bDelta}_{t+1})_{i, j} &= \mathbb E[\hat U_{i}\,\hat U_{j}], \quad 1\le i, j\le t+1,
    \end{split}
\end{equation}
where $\partial_j \hat U_i$ denotes the partial derivative $\partial_{\hat Z_j} \tilde{h}_i(\hat Z_1, \ldots, \hat Z_{i-1}, U_1, X)$. Then, we compute the covariance matrix $\hat{\bSigma}_{t+1}$ as
\begin{equation}\label{eq:defSigmat}
    \hat{\bSigma}_{t+1} = \sum_{j=0}^{2t}\bar{\kappa}_{j+2} \sum_{i=0}^j (\hat{\bPhi}_{t+1})^i\hat{\bDelta}_{t+1}(\hat{\bPhi}_{t+1}^\intercal)^{j-i}.
\end{equation}
It can be verified that the $t\times t$ top left sub-matrix of $\hat{\bSigma}_{t+1}$ is given by $\hat{\bSigma}_t$.

\begin{proposition}[State evolution for auxiliary AMP]
    Consider the auxiliary AMP in \eqref{eq:auxAMP} and the state evolution random variables defined in \eqref{eq:SElawsaux}. Let $\tilde{\psi}:\mathbb R^{2t+2} \to \mathbb R$ be a $\mathrm{PL}(2)$ function. Then, for each $t\ge 1$, we almost surely have
    \begin{align}
    & \lim_{N \to \infty}  \frac{1}{N} \sum_{i=1}^N  \tilde{\psi}(\tilde{z}^1_i, \ldots, \tilde{z}^t_i, \tilde{u}^1_i, \ldots, \tilde{u}^{t+1}_i, X^*_i) \notag
    \\
    &\hspace{4em}= \E[ \tilde{\psi}(\hat Z_1, \ldots, \hat Z_t, \hat U_1, \ldots, \hat U_{t+1}, X^*) ]. \label{eq:PL2aux}
    \end{align}
    Equivalently, as $N\to\infty$, almost surely:
    \begin{align}
        & (\tilde{\bz}^1, \ldots, \tilde{\bz}^{t}, \,  \tilde{\bu}^1, \ldots, \tilde{\bu}^{t+1}, \, \bX^*) \, \stackrel{\mathclap{W_2}}{\longrightarrow}  \,  (\hat Z_1, \ldots, \hat Z_t, \hat U_1, \ldots, \hat U_{t+1}, X^*).\label{eq:PL2aux2}
    \end{align}
    Furthermore, 
    \begin{equation}\label{eq:equiv3}
       (\hat Z_1, \ldots, \hat Z_t, \hat U_1, \ldots, \hat U_{t+1}, X^*) \stackrel{\rm d}{=}  (\tilde Z_1, \ldots, \tilde Z_t, \tilde U_1, \ldots, \tilde U_{t+1}, X^*),
    \end{equation}
    where $(\tilde Z_1, \ldots, \tilde Z_t, \tilde U_1, \ldots, \tilde U_{t+1}, X^*)$ are obtained via \eqref{eq:3ip0}--\eqref{eq:Zdef}. 
    
    \label{prop:auxSE}
\end{proposition}
\begin{proof}
The result follows from Theorem 2.3 in \cite{zhong2021approximate}. In fact, Assumption 2.1 of \cite{zhong2021approximate} holds because of the model assumptions on $\bZ$, Assumption 2.2(a) holds because $(\bX^*, \tilde\bu^1)=(\bX^*, \bu^1)\stackrel{\mathclap{W_2}}{\longrightarrow} (X^*, U_1)$ from \eqref{eq:AMPinit}, and Assumption 2.2(b) follows from the definition of $\tilde{h}_{t+1}$ in \eqref{eq:deftildeh} and the fact that $g_{t+1}$ is continuously differentiable and Lipschitz. As the auxiliary AMP in \eqref{eq:auxAMP} is of the standard form for which the state evolution result of Theorem 2.3 in \cite{zhong2021approximate} holds, we readily obtain \eqref{eq:PL2aux2}. The equivalence between \eqref{eq:PL2aux2} and \eqref{eq:PL2aux} follows from \cite[Corollary 7.21]{feng2022unifying}. Finally, by inspecting the state evolution recursions \eqref{eq:3ip0}--\eqref{eq:Zdef} and \eqref{eq:SElawsaux} giving $(\tilde Z_1, \ldots, \tilde Z_t, \tilde U_1, \ldots, \tilde U_{t+1}, X^*)$ and $(\hat Z_1, \ldots, \hat Z_t, \hat U_1, \ldots, \hat U_{t+1}, X^*)$ respectively, \eqref{eq:equiv3} is readily obtained.
\end{proof}

Proposition \ref{prop:auxSE} gives that the state evolution recursion discussed in Section \ref{subsec:OnsSE} (cf. \eqref{eq:3ip0}--\eqref{eq:Zdef}) coincides with the state evolution tracking the iterates of the auxiliary AMP algorithm \eqref{eq:auxAMP}. In particular, $\tilde\bDelta_{3t}=\hat\bDelta_{3t}$, $\tilde\bPhi_{3t}=\hat\bPhi_{3t}$, and $\tilde\bSigma_{3t}=\bar\bSigma_{3t}$. Furthermore, in the proof of Theorem \ref{th:SE} contained in Appendix \ref{appsubsec:proof}, we will show that $\bar\bB_{3t}\to \tilde\bB_{3t}$ as $N\to \infty$.

\subsection{Proof of Theorem \ref{th:SE}}\label{appsubsec:proof}

We start by presenting a useful technical lemma. 

\begin{lemma}
    \label{lem:lipderiv}
    Let $F : \mathbb R^t \to\mathbb R$ be a Lipschitz function, and let $\partial_k F$ denote its derivative with respect to the $k$-th argument, for $1 \le k \le t$. Assume that $\partial_k F$ is continuous almost everywhere in the $k$-th argument,  for each $k$. Let $(V_1^{(m)}, \ldots, V_t^{(m)})$  be a sequence of random vectors in $\mathbb R^t$ converging in distribution to the random vector $(V_1, \ldots,V_t)$ as $m \to \infty$. Furthermore, assume that the distribution of $(V_1, \ldots, V_t)$ is absolutely continuous with respect to the Lebesgue measure. Then, 
    \begin{equation}
         \lim_{m \to \infty}  \mathbb E[ \partial_k F(V_1^{(m)}, \ldots, V_t^{(m)})] = \mathbb E[ \partial_k F(V_1, \ldots, V_t) ], \qquad 1 \le k \le t.
    \end{equation}
\end{lemma}
The result was proved for $t=2$ in \cite[Lemma 6]{Bayati_Montanari11}. The proof for $t > 2$ is basically the same, see also \cite[Lemma 7.14]{feng2022unifying}.
At this point, we are ready to give the proof of Theorem \ref{th:SE}.

\begin{proof}[Proof of Theorem \ref{th:SE}]
We show that, for any $\mathrm{PL}(2)$ function  $\psi: \mathbb R^{2t+2} \to \mathbb R$, the following limit holds almost surely for $t \ge 1$: 
\begin{equation}\label{eq:PLsecondstep}
    \begin{split}
        \lim_{N \to \infty}  \Big\vert 
        \frac{1}{N}& \sum_{i=1}^{N} \psi\big(u^1_i,u^2_i,\ldots, u^{t+1}_i, f^1_i, f^2_i, \ldots, f^t_i,  X^*_i\big)
        \\
        & -   \frac{1}{N} \sum_{i=1}^{N} \psi\big(\tilde u^1_i,\tilde u^{K+1}_i, \ldots, \tilde u^{Kt+1}_i, \tilde f^1_i, \tilde f^2_i,\ldots, \tilde f^t_i,  X^*_i\big)  \Big \vert =0,
    \end{split}
\end{equation}
where we have defined for $s\in \{1, \ldots, t\}$,
\begin{equation}\label{eq:defu}
    \tilde \bff^s = \mu_s \bX^*+\sum_{i=1}^{Ks}\theta_{s, i}\tilde \bz^i.
\end{equation}
From here till the end of the argument, all the limits hold almost surely, and we use $C$ to denote a generic positive constant, which can change from line to line and is independent of $N$. By using that $\psi$ is pseudo-Lipschitz, we have that
\begin{equation}
    \begin{split}
         \Big\vert 
         &\frac{1}{N} \sum_{i=1}^{N} \psi\big(u^1_i,u^2_i,\ldots, u^{t+1}_i, f^1_i, f^2_i, \ldots, f^t_i,  X^*_i\big)
        \\
        & \hspace{3em}-   \frac{1}{N} \sum_{i=1}^{N} \psi\big(\tilde u^1_i,\tilde u^{K+1}_i, \ldots, \tilde u^{Kt+1}_i, \tilde f^1_i, \tilde f^2_i,\ldots, \tilde f^t_i,  X^*_i\big)  \Big \vert\\
        & \le \frac{C}{N} \sum_{i=1}^{N} \Big( 1+ |X^*_i| +2|u^1_i| +\sum_{k=1}^t \Big( |f^k_i| +|\tilde f^{k}_i|+ |u^{k+1}_i|  + |\tilde{u}^{Kk+1}_i|  \Big) \Big) \\
        &\hspace{14em} \cdot \Big( \sum_{k=1}^t  \big(|f^k_i - \tilde f^k_i |^2  +  |u^{k+1}_i - \tilde u^{Kk+1}_i |^2\big) \Big)^{1/2} \\
        & \leq C(4t+3)\Big[ 1 + \frac{\| \bX^* \|^2}{N} +  
        \sum_{k=1}^t \Big( \frac{\|\bff^k \|^2}{N} + \frac{\|\tilde \bff^k \|^2}{N} + \frac{\| \bu^{k+1} \|^2}{N} + 
        \frac{\| \tilde{\bu}^{Kk+1} \|^2}{N} \Big) \Big]^{1/2} \\
        &\hspace{13em}\cdot \Big( \sum_{k=1}^t \Big(\frac{\| \bff^k - \tilde\bff^{k}\|^2}{N} + \frac{\|\bu^{k+1} - \tilde{\bu}^{Kk+1}\|^2}{N} \, \Big)\Big)^{1/2},  \label{eq:xz_diff}
    \end{split}
\end{equation}
where the last step uses twice Cauchy-Schwarz inequality. We now inductively show that as $N \to \infty$: \emph{(i)} each of the terms in the last line of \eqref{eq:xz_diff} converges to zero, and \emph{(ii)} the terms within the square brackets in \eqref{eq:xz_diff} all converge to finite, deterministic limits. To achieve this goal, we will also show that, for $k\in [t]$ and $\ell\in [K-1]$, 
\begin{align}
    \lim_{N\to\infty}&\frac{\|\bY^\ell\bu^{k}-\tilde\bu^{K(k-1)+1+\ell}\|^2}{N}=0,\label{eq:auxAMP1}\\
    \lim_{N\to\infty}&\frac{\|\tilde\bu^{K(k-1)+1+\ell}\hspace{-.2em}-\hspace{-.2em}\sum_{j=1}^{K(k-1)+\ell}\hspace{-.2em}\alpha_{K(k-1)+1+\ell, j}\tilde\bz^j\hspace{-.2em}-\hspace{-.2em}\sum_{j=1}^{k}\hspace{-.2em}\beta_{K(k-1)+1+\ell, j}\bu^j\hspace{-.2em}-\hspace{-.2em}\gamma_{K(k-1)+1+\ell}\bX^*\|^2}{N}\notag\\
    &\hspace{32.5em}=0.  \label{eq:auxAMPnew2}
\end{align}
The limit \eqref{eq:auxAMP1} formalizes the idea discussed in Section \ref{subsec:BoptAMP} (see \eqref{eq:mapping1}) that the iterate $\tilde\bu^{K(k-1)+1+\ell}$ of the auxiliary AMP tracks the quantity $\bY^\ell\bu^{k}$, where $\bu^{k}$ is the iterate of the AMP we wish to analyze, up to an $o_N(1)$ error. The limit \eqref{eq:auxAMPnew2} formalizes the interpretation of the coefficients $\{\alpha_{i, j}\}$, $\{\beta_{i, j}\}$, $\{\gamma_i\}$ provided at the end of Section \ref{subsec:OnsSE} (see \eqref{eq:charac}).

\noindent \underline{Base case ($t=1$).} We have that
\begin{equation}\label{eq:b1man}
    \begin{split}
        \bY\bu^1-\tilde \bu^2 &= \bZ\bu^1+\lambda\frac{\langle \bX^*, \bu^1\rangle}{N}\bX^*-\tilde{\bz}^1-(\tilde \bB_1)_{1, 1}\bu^1-\tilde\mu_1\bX^* \\
        &=\Big(\lambda\frac{\langle \bX^*, \bu^1\rangle}{N}-\tilde\mu_1\Big)\bX^*+\big(\bar {\sf b}_{1, 1}-(\tilde \bB_1)_{1, 1}\big)\bu^1,
    \end{split}
\end{equation}
where the first equality uses the definition of $\bY$ and of $\tilde h_2$ (see \eqref{eq:deftildeh}), and the second equality uses \eqref{eq:auxAMP} and that $\tilde\bu^1=\bu^1$. Hence, by triangle inequality,
\begin{equation}\label{eq:base1}
    \begin{split}
        \frac{\|\bY\bu^1-\tilde \bu^2\|^2}{N}&\le 2\Big(\lambda\frac{\langle \bX^*, \bu^1\rangle}{N}-\tilde\mu_1\Big)^2\frac{\|\bX^*\|^2}{N}+2\big(\bar {\sf b}_{1, 1}-(\tilde \bB_1)_{1, 1}\big)^2\frac{\|\bu^1\|^2}{N}\\
        &\le C\Big(\Big(\lambda\frac{\langle \bX^*, \bu^1\rangle}{N}-\tilde\mu_1\Big)^2+\big(\bar {\sf b}_{1, 1}-(\tilde \bB_1)_{1, 1}\big)^2\Big),
    \end{split}
\end{equation}
where the last inequality uses that $(\bX^*, \bu^1)$ converges in $W_2$ to a pair of random variables with finite second moments. As $\tilde \mu_1=\lambda\epsilon$ (cf. \eqref{eq:SEnewinit}), we have 
\begin{equation}\label{eq:base2}
    \lim_{N\to\infty}\lambda\frac{\langle \bX^*, \bu^1\rangle}{N}=\lambda\mathbb E[ U_1X^*]=\lambda\epsilon=\tilde\mu_1.
\end{equation}
Furthermore, note that  $(\tilde\bB_1)_{1, 1}=\bar{\kappa}_1$ (cf. \eqref{eq:SEnewinit}) and $\bar{\sf b}_{1, 1}=\kappa_1$ (cf. \eqref{eq:deftildeB}). Hence, by the model assumptions, as $N\to\infty$, $\kappa_1\to\bar{\kappa}_1$ and, therefore, $\bar{\sf b}_{1, 1}\to (\tilde\bB_1)_{1, 1}$. By combining this observation with \eqref{eq:base1} and \eqref{eq:base2}, we obtain that \eqref{eq:auxAMP1} holds for $k=1$ and $\ell=1$.

By using \eqref{eq:alphaup}--\eqref{eq:gammaup}, we readily obtain that $\alpha_{2, 1}=1$, $\beta_{2, 1}=(\tilde\bB_1)_{1, 1}$ and $\gamma_{2}=\tilde\mu_1$. Hence, by using the definition \eqref{eq:deftildeh} of $\tilde h_2$, we obtain that \eqref{eq:auxAMPnew2} holds for $k=1$ and $\ell=1$.

Next, by using the definitions of $\bY$, of the auxiliary AMP \eqref{eq:auxAMP} and of $\tilde h_3$ (cf. \eqref{eq:deftildeh}), we have
\begin{equation}\label{eq:y2u2}
    \begin{split}
        &\bY^2\bu^1-\tilde\bu^3 = \bY(\bY\bu^1-\tilde\bu^2)+\bY\tilde\bu^2-\tilde\bz^2-(\tilde\bB_2)_{2, 1}\bu^1-(\tilde\bB_2)_{2, 2}\tilde\bu^2-\tilde\mu_2\bX^*\\
        &=\bY(\bY\bu^1-\tilde\bu^2) + \bZ\tilde\bu^2 -\tilde\bz^2-(\tilde\bB_2)_{2, 1}\bu^1-(\tilde\bB_2)_{2, 2}\tilde\bu^2+ \Big(\lambda\frac{\langle \bX^*, \tilde\bu^2\rangle}{N}-\tilde\mu_2\Big)\bX^*\\
        &=\bY(\bY\bu^1-\tilde\bu^2) + \big(\bar{\sf b}_{2, 1}-(\tilde\bB_2)_{2, 1}\big)\bu^1+\big(\bar{\sf b}_{2, 2}-(\tilde\bB_2)_{2, 2}\big)\tilde\bu^2+ \Big(\lambda\frac{\langle \bX^*, \tilde\bu^2\rangle}{N}-\tilde\mu_2\Big)\bX^*.
    \end{split}
\end{equation}
Hence, by triangle inequality,
\begin{equation}\label{eq:y2u1}
\begin{split}
    \frac{\|\bY^2\bu^1-\tilde\bu^3\|^2}{N}&\le C\Big(\frac{\|\bY(\bY\bu^1-\tilde\bu^2)\|^2}{N} + \big(\bar{\sf b}_{2, 1}-(\tilde\bB_2)_{2, 1}\big)^2\frac{\|\bu^1\|^2}{N} \\
    &+ \big(\bar{\sf b}_{2, 2}-(\tilde\bB_2)_{2, 2}\big)^2\frac{\|\tilde\bu^2\|^2}{N}+\Big(\lambda\frac{\langle \bX^*, \tilde\bu^2\rangle}{N}-\tilde\mu_2\Big)^2\frac{\|\bX^*\|^2}{N}\Big)\\
    &:=C(T_1+T_2+T_3+T_4).
\end{split}
\end{equation}
Consider the first term. As $\bY$ has bounded operator norm and \eqref{eq:auxAMP1} holds for $k=1$ and $\ell=1$, we have that $T_1\to 0$ as $N\to\infty$.

Consider the second and third terms. The following chain of equalities holds
\begin{equation}\label{eq:phiarg}
    \lim_{N\to \infty} (\bar\bPhi_{2})_{2, 1} =\lim_{N\to\infty}\langle\partial_1 \tilde{\bu}^2\rangle =\mathbb E[\partial_1\hat{U}_2]=\mathbb E[\partial_1\tilde{U}_2] = (\tilde{\bPhi}_{2})_{2, 1}.
\end{equation}
Here, the first equality uses the definition \eqref{eq:defbPhi}; the second equality follows from Lemma \ref{lem:lipderiv}, as $\tilde{\bu}^2$ converges in $W_2$ (and therefore in distribution) to $\tilde{U}_2$ and $\partial_1 \tilde U_2$ is continuous; the third equality uses \eqref{eq:equiv3}; and the fourth equality uses the definition of $(\tilde{\bPhi}_{2})_{2, 1}$ in \eqref{eq:Phiup}. By the model assumptions, as $N\to\infty$, $\kappa_j\to\bar{\kappa}_j$ for all $j$. Thus, by combining \eqref{eq:phiarg} with the definitions of $\bar\bB_{2}$ and $\tilde{\bB}_{2}$ in \eqref{eq:deftildeB} and \eqref{eq:Bup}, respectively, we conclude that, as $N\to\infty$, $\bar{\sf b}_{2, i}\to(\tilde\bB_{2})_{2, i}$ for $i\in \{1, 2\}$. By Proposition \ref{prop:auxSE}, $\|\tilde{\bu}^2\|^2/N$ converges to a finite limit, hence we conclude that $T_2, T_3\to 0$ as $N\to\infty$.

Consider the fourth term. Then,
\begin{equation*}
    \lim_{N\to \infty}\lambda\frac{\langle \bX^*, \tilde{\bu}^{2}\rangle}{N}=\lambda\mathbb E[X^*\,\tilde{U}_{2}]=\tilde\mu_{2}.
\end{equation*}
Here, the first equality uses Proposition \ref{prop:auxSE} and the second equality uses the definition of $\tilde\mu_{2}$ in \eqref{eq:tildemuup}. As $\|\bX^*\|^2/ N=1$, we conclude that $T_4\to 0$ as $N\to\infty$. This proves that the RHS of \eqref{eq:y2u1} vanishes and gives that \eqref{eq:auxAMP1} holds for $k=1$ and $\ell=2$.

By using \eqref{eq:alphaup}--\eqref{eq:gammaup}, we readily obtain that $\alpha_{3, 1}=(\tilde\bB_2)_{2, 2}$, $\alpha_{3, 2}=1$, $\beta_{3, 1}=(\tilde\bB_2)_{2, 1}+(\tilde\bB_2)_{2, 2}\,(\tilde\bB_2)_{1, 1}$ and $\gamma_{3}=\tilde\mu_2+\tilde\mu_1\,(\tilde\bB_2)_{2, 2}$. Hence, by using the definition \eqref{eq:deftildeh} of $\tilde h_3$, we obtain that \eqref{eq:auxAMPnew2} holds for $k=1$ and $\ell=2$.

The proof of \eqref{eq:auxAMP1}--\eqref{eq:auxAMPnew2} for $k=1$ and $\ell\in \{3, \ldots, K-1\}$ follows from similar arguments. In particular, we write 
\begin{equation}\label{eq:yellu2}
    \begin{split}
        &\bY^\ell\bu^1-\tilde\bu^{1+\ell} = \bY(\bY^{\ell-1}\bu^1-\tilde\bu^\ell)+\bY\tilde\bu^\ell-\tilde\bz^\ell-(\tilde\bB_\ell)_{\ell, 1}\bu^1-\sum_{j=2}^\ell (\tilde\bB_\ell)_{\ell, j}\tilde\bu^j-\tilde\mu_\ell\bX^*\\
        &=\bY(\bY^{\ell-1}\bu^1-\tilde\bu^\ell) + \bZ\tilde\bu^\ell -\tilde\bz^\ell-(\tilde\bB_\ell)_{\ell, 1}\bu^1-\sum_{j=2}^\ell (\tilde\bB_\ell)_{\ell, j}\tilde\bu^j+ \Big(\lambda\frac{\langle \bX^*, \tilde\bu^\ell\rangle}{N}-\tilde\mu_\ell\Big)\bX^*\\
        &=\bY(\bY^{\ell-1}\bu^{1}-\tilde\bu^\ell) + \big(\bar{\sf b}_{\ell, 1}-(\tilde\bB_\ell)_{\ell, 1}\big)\bu^1+\sum_{j=2}^\ell\big(\bar{\sf b}_{\ell, j}-(\tilde\bB_\ell)_{\ell, j}\big)\tilde\bu^j\\
        &\hspace{20em}+ \Big(\lambda\frac{\langle \bX^*, \tilde\bu^\ell\rangle}{N}-\tilde\mu_\ell\Big)\bX^*,
    \end{split}
\end{equation}
which by triangle inequality gives 
\begin{equation}\label{eq:y2u1ell}
\begin{split}
    \frac{\|\bY^\ell\bu^1-\tilde\bu^{1+\ell}\|^2}{N}&\le C\Big(\frac{\|\bY(\bY^{\ell-1}\bu^1-\tilde\bu^\ell)\|^2}{N} + \big(\bar{\sf b}_{\ell, 1}-(\tilde\bB_\ell)_{\ell, 1}\big)^2\frac{\|\bu^1\|^2}{N} \\
    &+ \sum_{j=2}^\ell\big(\bar{\sf b}_{\ell, j}-(\tilde\bB_\ell)_{\ell, j}\big)^2\frac{\|\tilde\bu^j\|^2}{N}+\Big(\lambda\frac{\langle \bX^*, \tilde\bu^\ell\rangle}{N}-\tilde\mu_\ell\Big)^2\frac{\|\bX^*\|^2}{N}\Big).
\end{split}
\end{equation}
As $\bY$ has bounded operator norm and $\|\bY^{\ell-1}\bu^1-\tilde\bu^\ell\|^2/N\to 0$ (by the previous step), we have that 
\begin{equation*}
\lim_{N\to\infty}    \frac{\|\bY(\bY^{\ell-1}\bu^1-\tilde\bu^\ell)\|^2}{N}=0.
\end{equation*}
Next, by following passages analogous to those in \eqref{eq:phiarg}, we have that $\lim_{N\to\infty}\bar\bPhi_{\ell}=\tilde\bPhi_\ell$. As $\kappa_j\to\bar\kappa_j$ for all $j$, this implies that $\lim_{N\to\infty}\bar\bB_\ell=\tilde\bB_\ell$. Hence, for all $j\in [\ell]$, as $\|\tilde\bu^j\|/N$ is bounded, we have that
\begin{equation*}
    \lim_{N\to\infty}\Bigg(\big(\bar{\sf b}_{\ell, 1}-(\tilde\bB_\ell)_{\ell, 1}\big)^2\frac{\|\bu^1\|^2}{N}+ \sum_{j=2}^\ell\big(\bar{\sf b}_{\ell, j}-(\tilde\bB_\ell)_{\ell, j}\big)^2\frac{\|\tilde\bu^j\|^2}{N}\Bigg)=0.
\end{equation*}
Finally, as
\begin{equation*}
     \lim_{N\to \infty}\lambda\frac{\langle \bX^*, \tilde{\bu}^{\ell}\rangle}{N}=\lambda\mathbb E[X^*\,\tilde{U}_{\ell}]=\tilde\mu_{\ell}, 
\end{equation*}
we conclude that the last term in the RHS of \eqref{eq:y2u1ell} vanishes as well, which proves that \eqref{eq:auxAMP1} holds for $k=1$ and a generic $\ell\in \{3, \ldots, K-1\}$. Furthermore, by using \eqref{eq:alphaup}--\eqref{eq:gammaup} and the definition \eqref{eq:deftildeh} of $\tilde h_{\ell+1}$, one can readily verify that \eqref{eq:auxAMPnew2} holds for $k=1$ and a generic $\ell\in \{3, \ldots, K-1\}$.

By using \eqref{eq:auxAMP} and the definition of $\bY$, we have that
\begin{equation}
\begin{split}
    \bY^K\bu^1-\tilde\bz^K-\sum_{i=1}^K\bar{\sf b}_{K, i}\tilde\bu^i-\tilde\mu_K\bX^*&=\bZ\big(\bY^{K-1}\bu^1- \tilde\bu^K\big)+ \Big(\lambda\frac{\langle \bX^*, \bY^{K-1}\bu^1\rangle}{N}-\tilde\mu_K\Big)\bX^*.
\end{split}
\end{equation}
Hence, by using the definition of $\tilde\mu_K$ in \eqref{eq:tildemuup} and \eqref{eq:auxAMP1} with $k=1$, $\ell=K-1$, we obtain
\begin{equation}\label{eq:auxAMP3}
    \lim_{N\to\infty}\frac{\|\bY^K\bu^1-\tilde\bz^K-\sum_{i=1}^K\bar{\sf b}_{K, i}\tilde\bu^i-\tilde\mu_K\bX^*\|^2}{N}=0.
\end{equation}
Recall that $\bJ(\bY)=\sum_{j=1}^K c_j \bY^j$. Then, by combining \eqref{eq:auxAMP1} with $k=1$ and \eqref{eq:auxAMP3}, we have
\begin{equation}\label{eq:intp1}
    \lim_{N\to\infty}\frac{\|\bJ(\bY)\bu^1-\sum_{j=1}^K c_j \big(\tilde \bz^{j}+\sum_{i=1}^j\bar{\sf b}_{j, i}\tilde\bu^i+\tilde\mu_j\bX^*\big)\|^2}{N} = 0.
\end{equation}
By following the same argument as in \eqref{eq:phiarg}, we have that $\lim_{N\to\infty}\bar\bPhi_K=\tilde\bPhi_K$. As $\kappa_j\to\bar\kappa_j$ for all $j$, this implies that $\lim_{N\to\infty}\bar\bB_K=\tilde\bB_K$. Therefore,
\begin{equation}
\begin{split}
    \lim_{N\to\infty}\frac{\bigg\|\displaystyle\sum_{j=1}^K c_j \big(\tilde \bz^{j}+\sum_{i=1}^j\bar{\sf b}_{j, i}\tilde\bu^i+\tilde\mu_j\bX^*\big) - \sum_{j=1}^K c_j \big(\tilde \bz^{j}+\sum_{i=1}^j(\tilde\bB_j)_{j, i}\tilde\bu^i+\tilde\mu_j\bX^*\big)\bigg\|^2}{N}= 0.
    \end{split}
\end{equation}
Recall that $\tilde\bu^1 = \bu^1$ and \eqref{eq:auxAMPnew2} holds for $k=1$. Hence, by plugging in the formulas for ${\sf c}_{1, 1}$, $\mu_1$ and $\{\theta_{1, i}\}_{i\in [K]}$ (cf. \eqref{eq:Ons}, \eqref{eq:muup} and \eqref{eq:thetaup}), we have
\begin{equation}\label{eq:intp3}
\begin{split}
    \lim_{N\to\infty}\frac{\bigg\|\displaystyle\sum_{j=1}^K c_j \big(\tilde \bz^{j}+\sum_{i=1}^j(\tilde\bB_j)_{j, i}\tilde\bu^i+\tilde\mu_j\bX^*\big) - {\sf c}_{1, 1}\bu^1 - \mu_1\bX^* - \sum_{i=1}^K\theta_{1, i}\tilde\bz^i\bigg\|^2}{N} = 0.
    \end{split}
\end{equation}
By combining \eqref{eq:intp1}--\eqref{eq:intp3} with the definitions of $\bff^1$ and $\tilde\bff^1$(cf. \eqref{eq:AMPnew} and \eqref{eq:defu}), we conclude that 
\begin{equation}\label{eq:it1}
\lim_{N\to\infty}\frac{\|\bff^1-\tilde\bff^1\|^2}{N} =0.   
\end{equation}
As $g_2$ is Lipschitz, \eqref{eq:it1} immediately implies that
\begin{equation}
    \lim_{N\to\infty}\frac{\|\bu^2-\tilde\bu^{K+1}\|^2}{N}=0.
\end{equation}
An application of the triangle inequality gives that, for any $i\ge 1$,
\begin{equation}\label{eq:tr1}
    \begin{split}
        \|\tilde\bff^{i}\|- \|\bff^{i}-\tilde\bff^{i}\|&\le \|\bff^{i}\|\le \|\tilde\bff^{i}\|+ \|\bff^{i}-\tilde\bff^{i}\|, \\
         \|\tilde\bu^{Ki+1}\|- \|\bu^{i+1}-\tilde\bu^{Ki+1}\|&\le \|\bu^{i+1}\|\le \|\tilde\bu^{Ki+1}\|+ \|\bu^{i+1}-\tilde\bu^{Ki+1}\|.
    \end{split}
\end{equation}
Thus, by using \eqref{eq:tr1} with $i=1$ and Proposition \ref{prop:auxSE}, we obtain that
\begin{equation}
    \begin{split}
        &\lim_{N\to\infty}\frac{\|\bff^1\|^2}{N}=\lim_{N\to\infty}\frac{\|\tilde\bff^1\|^2}{N}=\mathbb E\Big[\Big(\mu_1 X^*+\sum_{i=1}^K\theta_{1, i}\tilde{Z}_i\Big)^2\Big], \\
        &\lim_{N\to\infty}\frac{\|\bu^2\|^2}{N}=\lim_{N\to\infty}\frac{\|\tilde{\bu}^{K+1}\|^2}{N}=\mathbb E[(\tilde{U}_{K+1})^2],
    \end{split}
\end{equation}
which concludes the base step.

\noindent \underline{Induction step.} Assume towards induction that \eqref{eq:auxAMP1}--\eqref{eq:auxAMPnew2} hold for $k\in [t]$, $\ell\in [K-1]$ and that, for $k\in [t]$,
\begin{align}
    &\lim_{N\to\infty}\frac{\|\bff^{k}-\tilde\bff^{k}\|^2}{N}=0,\label{eq:ind0}\\
    &\lim_{N\to\infty}\frac{\|\bu^{k+1}-\tilde{\bu}^{Kk+1}\|^2}{N}=0,\label{eq:ind1}\\
    &\lim_{N\to\infty}\frac{\|\bff^{k}\|^2}{N}=\lim_{N\to\infty}\frac{\|\tilde\bff^{k}\|^2}{N}=\mathbb E\Big[\Big(\mu_k X^*+\sum_{i=1}^{Kk}\theta_{k, i}\tilde{Z}_i\Big)^2\Big],\label{eq:ind1bis}\\
    &\lim_{N\to\infty}\frac{\|\bu^{k+1}\|^2}{N}=\lim_{N\to\infty}\frac{\|\tilde{\bu}^{Kk+1}\|^2}{N}=\mathbb E[\tilde{U}_{Kk+1}^2].\label{eq:ind2}
\end{align}
We now show that \eqref{eq:ind0}--\eqref{eq:ind2} hold for $k= t+1$, and that  \eqref{eq:auxAMP1}--\eqref{eq:auxAMPnew2} hold for $k= t+1$, $\ell\in [K-1]$.
By doing so, we will have proved also the induction step and consequently that \eqref{eq:PLsecondstep} holds. 

Using similar passages as in \eqref{eq:b1man}, we obtain
\begin{equation}
    \begin{split}
        &\bY\bu^{t+1}-\tilde \bu^{Kt+2} = \bZ\bu^{t+1}+\lambda\frac{\langle \bX^*, \bu^{t+1}\rangle}{N}\bX^*-\tilde{\bz}^{Kt+1}-\sum_{i=1}^{Kt+1}(\tilde \bB_{Kt+1})_{Kt+1, i}\tilde\bu^i-\tilde\mu_{Kt+1}\bX^* \\
        &=\bZ(\bu^{t+1}-\tilde\bu^{Kt+1})+\Big(\lambda\frac{\langle \bX^*, \bu^{t+1}\rangle}{N}-\tilde\mu_{Kt+1}\Big)\bX^*+\sum_{i=1}^{Kt+1}\big(\bar {\sf b}_{Kt+1, i}-(\tilde \bB_{Kt+1})_{Kt+1, i}\big)\tilde\bu^i.
    \end{split}
\end{equation}
Hence, by triangle inequality,
\begin{equation}\label{eq:decT}
    \begin{split}
&        \frac{\|\bY\bu^{t+1}-\tilde \bu^{Kt+2}\|^2}{N}\le C\Big(\frac{\|\bZ(\bu^{t+1}-\tilde\bu^{Kt+1})\|^2}{N}\\
        &\hspace{1em}+\Big(\lambda\frac{\langle \bX^*, \bu^{t+1}\rangle}{N}-\tilde\mu_{Kt+1}\Big)^2\frac{\|\bX^*\|^2}{N}+\sum_{i=1}^{Kt+1}\big(\bar {\sf b}_{Kt+1, i}-(\tilde \bB_{Kt+1})_{Kt+1, i}\big)^2\frac{\|\tilde{\bu}^i\|^2}{N}\Big)\\
        &\hspace{2em}:= C(\bar T_1+\bar T_2+\bar T_3).
    \end{split}
\end{equation}
Consider the first term. Since $\|\bZ\|_{\rm op}\le C$, the induction hypothesis \eqref{eq:ind1} implies that $\bar T_1\to 0$ as $N\to\infty$.

Consider the second term. The following chain of equalities holds:
\begin{equation}\label{eq:T2}
    \lim_{N\to \infty}\lambda\frac{\langle \bX^*, \bu^{t+1}\rangle}{N}=\lim_{N\to \infty}\lambda\frac{\langle \bX^*, \tilde{\bu}^{Kt+1}\rangle}{N}=\lambda \mathbb E[X\,\tilde{U}_{Kt+1}]=\tilde\mu_{Kt+1}.
\end{equation}
Here, the first equality uses \eqref{eq:ind1} together with the fact that $\|\bX^*\|^2/ N=1$; the second equality follows from Proposition \ref{prop:auxSE}; and the third equality uses the definition of $\tilde\mu_{Kt+1}$ in \eqref{eq:tildemuup}. Finally, using \eqref{eq:T2} and again that $\|\bX^*\|^2/ N=1$ gives that $\bar T_2\to 0$ as $N\to\infty$.

Consider the third term. By following the same argument as in \eqref{eq:phiarg}, we have that $\lim_{N\to\infty}\bar\bPhi_{Kt+1}=\tilde\bPhi_{Kt+1}$. As $\kappa_j\to\bar\kappa_j$ for all $j$, this implies that $\lim_{N\to\infty}\bar\bB_{Kt+1}=\tilde\bB_{Kt+1}$. 
 By using the induction hypothesis \eqref{eq:ind2}, which shows that $\|\tilde{\bu}^i\|^2/N$ converges to a finite limit, we conclude that $\bar T_3\to 0$ as $N\to\infty$. This proves that the RHS of \eqref{eq:decT} vanishes and gives that \eqref{eq:auxAMP1} holds for $k=t+1$ and $\ell=1$.
 
 For $\ell\in \{2, \ldots, K-1\}$, by following passages similar to \eqref{eq:yellu2}, we have
\begin{equation*}
    \begin{split}
        \bY^\ell\bu^{t+1}-\tilde\bu^{Kt+\ell+1} &=\bY(\bY^{\ell-1}\bu^{t+1}-\tilde\bu^{Kt+\ell}) + \sum_{i=1}^{Kt+\ell}\big(\bar{\sf b}_{Kt+
        \ell, i}-(\tilde\bB_{Kt+\ell})_{Kt+\ell, i}\big)\tilde\bu^i\\
        &+ \Big(\lambda\frac{\langle \bX^*, \tilde\bu^{Kt+\ell}\rangle}{N}-\tilde\mu_{Kt+\ell}\Big)\bX^*,
    \end{split}
\end{equation*}
which by triangle inequality gives
\begin{equation}\label{eq:y2ut1}
\begin{split}
    &\frac{\|\bY^\ell\bu^{t+1}-\tilde\bu^{Kt+\ell+1}\|^2}{N}\le C\Big(\frac{\|\bY(\bY^{\ell-1}\bu^{t+1}-\tilde\bu^{Kt+\ell})\|^2}{N}  \\
    &\hspace{3em}+ \sum_{i=1}^{Kt+\ell}\big(\bar{\sf b}_{Kt+\ell, i}-(\tilde\bB_{Kt+\ell})_{Kt+\ell, i}\big)^2\frac{\|\tilde\bu^i\|^2}{N}+ \Big(\lambda\frac{\langle \bX^*, \tilde\bu^{Kt+\ell}\rangle}{N}-\tilde\mu_{Kt+\ell}\Big)^2\frac{\|\bX^*\|^2}{N}\Big).
\end{split}
\end{equation}
The first term on the RHS of \eqref{eq:y2ut1} vanishes as $\bY$ has bounded operator norm and we have just proved in the previous step that $\|\bY^{\ell-1}\bu^{t+1}-\tilde\bu^{Kt+\ell}\|^2/N\to 0$. To bound the second term, note that, by following the same argument as in \eqref{eq:phiarg}, we have that $\lim_{N\to\infty}\bar\bPhi_{Kt+\ell}=\tilde\bPhi_{Kt+\ell}$. As $\kappa_j\to\bar\kappa_j$ for all $j$, this implies that $\lim_{N\to\infty}\bar\bB_{Kt+\ell}=\tilde\bB_{Kt+\ell}$. 
 By using the induction hypothesis \eqref{eq:ind2}, we have that $\|\tilde{\bu}^i\|^2/N$ converges to a finite limit for $i\in [Kt+\ell-1]$. Furthermore, as $\|\bY^{\ell-1}\bu^{t+1}-\tilde\bu^{Kt+\ell}\|^2/N\to 0$, we also have that $\|\tilde{\bu}^{Kt+\ell}\|^2/N$ converges to a finite limit. As a result, the second term on the RHS of \eqref{eq:y2ut1} vanishes. Finally, we can write a chain of equalities analogous to \eqref{eq:T2} with $Kt+\ell$ in place of $Kt+1$, from which we deduce that the third term vanishes. This concludes the proof that \eqref{eq:auxAMP1} holds for $k=t+1$ and $\ell\in [K-1]$.

For $\ell\in [K-1]$, by definition \eqref{eq:deftildeh} of $h_{Kt+1+\ell}$, we have
\begin{equation}\label{eq:hatuexp0}
\begin{split}
    \tilde\bu^{Kt+1+\ell} &= \tilde\bz^{Kt+\ell} + \tilde\mu_{Kt+\ell}\bX^* + \sum_{i=1}^{Kt+\ell}(\tilde\bB_{Kt+\ell})_{Kt+\ell, i}\tilde\bu^i.
\end{split}
\end{equation}
Let us define:
\begin{equation}\label{eq:hatuexp}
\begin{split}
    \hat\bu^{Kt+1+\ell} &:=  \tilde\bz^{Kt+\ell} + \tilde\mu_{Kt+\ell}\bX^* + \sum_{i=1}^{t+1}(\tilde\bB_{Kt+\ell})_{Kt+\ell, K(i-1)+1}\bu^{i} \\
    &\hspace{1em}+ \sum_{\substack{i=1\\ i\not\equiv 1 ({\rm mod} K)}}^{Kt+\ell}(\tilde\bB_{Kt+\ell})_{Kt+\ell, i} 
    \Big(\sum_{j=1}^{i-1}\alpha_{i, j}\tilde\bz^j+\sum_{j=1}^{\lceil (i-1)/K\rceil}\beta_{i, j}\bu^j+\gamma_{i}\bX^*\Big).
\end{split}
\end{equation}
Then, by using the recursive definitions \eqref{eq:alphaup}--\eqref{eq:gammaup}, we readily have that the RHS of \eqref{eq:hatuexp} is equal to 
\begin{equation}\label{eq:hatuexp1}
\sum_{j=1}^{Kt+\ell}\alpha_{Kt+1+\ell, j}\tilde\bz^j + \sum_{j=1}^{t+1}\beta_{Kt+1+\ell, j}\bu^j + \gamma_{Kt+1+\ell}\bX^*.
\end{equation}
Recall that, by induction hypothesis, \eqref{eq:ind1} holds for $k\in [t]$, and \eqref{eq:auxAMPnew2} holds for $k\in [t]$ and $\ell\in [K-1]$. Thus, by using the expressions in \eqref{eq:hatuexp0} and \eqref{eq:hatuexp} for $\ell=1$, one readily obtains that
\begin{equation}
    \lim_{N\to\infty}\frac{\|\tilde\bu^{Kt+2}-\hat\bu^{Kt+2}\|^2}{N}=0.
\end{equation}
Since the RHS of \eqref{eq:hatuexp} is equal to the expression in \eqref{eq:hatuexp1} for $\ell=1$, we conclude that \eqref{eq:auxAMPnew2} holds for $k=t+1$ and $\ell=1$. At this point, we have that \eqref{eq:auxAMPnew2} holds for $k\in [t]$, $\ell\in [K-1]$ and also for $k=t+1, \ell=1$. Hence, by using the expressions in \eqref{eq:hatuexp0} and \eqref{eq:hatuexp} for $\ell=2$, we obtain
\begin{equation}
    \lim_{N\to\infty}\frac{\|\tilde\bu^{Kt+3}-\hat\bu^{Kt+3}\|^2}{N}=0.
\end{equation}
Since the RHS of \eqref{eq:hatuexp} is equal to the expression in \eqref{eq:hatuexp1} for $\ell=2$, we conclude that \eqref{eq:auxAMPnew2} holds for $k=t+1$, $\ell=2$. By iterating this procedure for $\ell\in \{3, \ldots, K-1\}$, we obtain that \eqref{eq:auxAMPnew2} holds for $k=t+1$, $\ell\in [K-1]$.

By using \eqref{eq:auxAMP} and the definition of $\bY$, we have that
\begin{equation}
\begin{split}
    \bY^K\bu^{t+1}-\tilde\bz^{K(t+1)}-&\sum_{i=1}^{K(t+1)}\bar{\sf b}_{K(t+1), i}\tilde\bu^i-\tilde\mu_{K(t+1)}\bX^*\\
    &=\bZ\big(\bY^{K-1}\bu^{t+1}- \tilde\bu^{K(t+1)}\big)+ \Big(\lambda\frac{\langle \bX^*, \bY^{K-1}\bu^{t+1}\rangle}{N}-\tilde\mu_{K(t+1)}\Big)\bX^*.
\end{split}
\end{equation}
Hence, by using \eqref{eq:auxAMP1} with $k=t+1$, $\ell=K-1$ and the definition of $\tilde\mu_{K(t+1)}$ in \eqref{eq:tildemuup}, we obtain
\begin{equation}\label{eq:auxAMP3ind}
    \lim_{N\to\infty}\frac{\|\bY^K\bu^{t+1}-\tilde\bz^{K(t+1)}-\sum_{i=1}^{K(t+1)}\bar{\sf b}_{K(t+1), i}\tilde\bu^i-\tilde\mu_{K(t+1)}\bX^*\|^2}{N}=0.
\end{equation}
As $\bJ(\bY)=\sum_{j=1}^K c_j \bY^j$, by combining \eqref{eq:auxAMP3ind} with \eqref{eq:auxAMP1} with $k=t+1$, $\ell\in [K-1]$, we have
\begin{equation}\label{eq:intp1ind}
    \lim_{N\to\infty}\frac{\|\bJ(\bY)\bu^{t+1}-\sum_{j=1}^K c_j \big(\tilde \bz^{Kt+j}+\sum_{i=1}^{Kt+j}\bar{\sf b}_{Kt+j, i}\tilde\bu^i+\tilde\mu_{Kt+j}\bX^*\big)\|^2}{N} = 0.
\end{equation}
By following the same argument as in \eqref{eq:phiarg}, we have that $\lim_{N\to\infty}\bar\bPhi_{Kt+j}=\tilde\bPhi_{Kt+j}$ for all $j\in [K]$. As $\kappa_j\to\bar\kappa_j$ for all $j$, this implies that $\lim_{N\to\infty}\bar\bB_{Kt+j}=\tilde\bB_{Kt+j}$ for all $j\in [K]$. Therefore, \eqref{eq:intp1ind} implies that 
\begin{equation}
\begin{split}
    \lim_{N\to\infty}\frac{\|\bJ(\bY)\bu^{t+1} - \sum_{j=1}^K c_j \big(\tilde \bz^{Kt+j}+\sum_{i=1}^{Kt+j}(\tilde\bB_{Kt+j})_{Kt+j, i}\tilde\bu^i+\tilde\mu_{Kt+j}\bX^*\big)\|^2}{N}= 0.
    \end{split}
\end{equation}
Recall that \eqref{eq:ind1} holds for $k\in [t]$ by the induction hypothesis and \eqref{eq:auxAMPnew2} holds for $k\in [t+1]$, $\ell\in [K-1]$ (thanks to the induction hypothesis and the argument above). Hence, by plugging in the formulas for $\{{\sf c}_{t+1, i}\}_{i\in [t+1]}$, $\mu_{t+1}$ and $\{\theta_{t+1, i}\}_{i\in [K(t+1)]}$ (cf. \eqref{eq:Ons}, \eqref{eq:muup} and \eqref{eq:thetaup}), we have
\begin{equation}\label{eq:intp3ind}
\begin{split}
    \lim_{N\to\infty}\frac{\|\bJ(\bY)\bu^{t+1} - \sum_{i=1}^{t+1}{\sf c}_{t+1, i}\bu^i - \mu_t\bX^* - \sum_{i=1}^{K(t+1)}\theta_{t+1, i}\tilde\bz^i\|^2}{N} = 0.
    \end{split}
\end{equation}
By recalling the definitions of $\bff^{t+1}$ and $\tilde\bff^{t+1}$ (cf. \eqref{eq:AMPnew} and \eqref{eq:defu}), \eqref{eq:intp3ind} implies that 
\begin{equation}\label{eq:it1ind}
\lim_{N\to\infty}\frac{\|\bff^{t+1}-\tilde\bff^{t+1}\|^2}{N} =0.   
\end{equation}
As $g_{t+2}$ is Lipschitz, \eqref{eq:it1ind} also gives that
\begin{equation}
    \lim_{N\to\infty}\frac{\|\bu^{t+2}-\tilde\bu^{K(t+1)+1}\|^2}{N}=0.
\end{equation}
Then, by using \eqref{eq:tr1} with $i=t+1$ and Proposition \ref{prop:auxSE}, we obtain that \eqref{eq:ind1bis} and \eqref{eq:ind2} hold for $k=t+1$, thus concluding the inductive proof. The result we have just proved by induction, combined with \eqref{eq:xz_diff}, gives that \eqref{eq:PLsecondstep} holds. 

Another application of Proposition \ref{prop:auxSE}, together with \eqref{eq:PLsecondstep}, gives that
\begin{equation}\label{eq:fin1}
    \begin{split}
         \lim_{N \to \infty}  
         \frac{1}{N} \sum_{i=1}^{N} \psi\big(&\tilde u^1_i,\tilde u^{K+1}_i, \ldots, \tilde u^{Kt+1}_i, \tilde f^1_i, \tilde f^2_i,\ldots, \tilde f^t_i,  X^*_i\big) \\
         &=\mathbb E [ \psi(\tilde U_1, \tilde U_{K+1}, \ldots, \tilde U_{Kt+1}, F_1, \ldots, F_t, X^*) ],
     \end{split}
\end{equation}
where we recall that, by the definition in the theorem statement, for $s\in \{1, \ldots, t\}$,
\begin{equation}
    F_s = \mu_s X^*+\sum_{i=1}^{Ks} \theta_{s, i}\tilde Z_i.
\end{equation}
As $U_{s+1}=g_{s+1}(F_s)$, we have $\tilde U_{Ks+1}=U_{s+1}$ for all $s \in [t]$, and the proof is complete.
\end{proof}

\bibliographystyle{plain}
\bibliography{refs}

\end{document}